\documentclass[12pt]{amsart}
\usepackage[margin=2cm]{geometry}

\usepackage{amsmath, amsthm, amssymb}
\input xy
\xyoption{all}
\usepackage{mathrsfs}
\usepackage{enumerate}
\usepackage{caption}
\usepackage{subcaption}
\usepackage{graphicx}
\usepackage{color}
\usepackage{tikz}
\usepackage{tkz-euclide}
\usetkzobj{all}

\usepackage[driverfallback=hypertex]{hyperref}
\usepackage{nameref,zref-xr}                    
\usepackage{footmisc}

\usepackage{cancel}
\usepackage[normalem]{ulem}

\newtheorem{dummy}{}[section]
\newtheorem{thm}[dummy]{Theorem}
\newtheorem{prop}[dummy]{Proposition}
\newtheorem{pr}[dummy]{Proposition}

\newtheorem{lemma}[dummy]{Lemma}
\newtheorem{cor}[dummy]{Corollary}

\newtheorem{construction}[dummy]{Construction/Notation}

\theoremstyle{definition}
\newtheorem{definition}[dummy]{Definition}
\newtheorem{nn}[dummy]{Notation}

\theoremstyle{remark}
\newtheorem{rmk}[dummy]{Remark}

\newtheorem{obs}[dummy]{Observation}

\newcommand{\CM}{{\mathcal{M}}}
\newcommand{\CMm}{\mathcal{M}^{\text{main}}}
\newcommand{\oCM}{{\overline{\mathcal{M}}}}
\newcommand{\oCMm}{\overline{\mathcal{M}}^{\text{main}}}
\newcommand{\oCMr}{{{\overline{\mathcal{M}}}^{\frac{1}{r}}_{0,k,\{a_1\ldots,a_l\}}}}
\newcommand{\CMr}{{{{\mathcal{M}}}^{\frac{1}{r}}_{0,k,\{a_1\ldots,a_l\}}}}
\newcommand{\Gammar}{{\Gamma_{0,k,\{a_1\ldots,a_l\}}}}

\newcommand{\CL}{{\mathbb{L}}}

\newcommand{\alt}{{\text{alt}}}

\newcommand{\detach}{{\text{detach} }}

\newcommand{\CB}{{\mathcal{B}}}

\newcommand{\cS}{\mathcal{S}}
\newcommand{\cJ}{\mathcal{J}}
\newcommand{\cW}{\mathcal{W}}

\newcommand{\tw}{\text{tw}}
\newcommand{\twE}{\text{tw}}
\newcommand{\NNN}{{n}}

\newcommand{\N}{\mathbb{N}}

\newcommand{\TTT}{{\mathcal{T}}}
\newcommand{\RRR}{{\mathcal{R}}}

\newcommand{\Conn}{{{\text{Conn}}}}

\newcommand{\Z}{\ensuremath{\mathbb{Z}}}

\newcommand{\C}{\ensuremath{\mathbb{C}}}
\newcommand{\R}{\ensuremath{\mathbb{R}}}

\newcommand{\M}{\ensuremath{\overline{\mathcal{M}}}}
\renewcommand{\O}{\ensuremath{\mathcal{O}}}

\renewcommand{\d}{\ensuremath{\partial}}
\renewcommand{\subset}{\ensuremath{\subseteq}}

\newcommand{\<}{\left<}
\renewcommand{\>}{\right>}

\DeclareMathOperator{\ior}{int}

\usepackage{geometry}
\numberwithin{equation}{section}

\begin{document}

\title{Open $r$-spin theory I: Foundations}

\author{Alexandr Buryak}
\address{A.~Buryak:\newline School of Mathematics, University of Leeds, Leeds, LS2 9JT, United Kingdom}\email{a.buryak@leeds.ac.uk}

\author{Emily Clader}
\address{E.~Clader:\newline San Francisco State University, San Francisco, CA 94132-1722, USA}
\email{eclader@sfsu.edu}

\author{Ran J. Tessler}
\address{R.~J.~Tessler, Incumbent of the Lilian and George Lyttle Career Development Chair:\newline Department of Mathematics, Weizmann Institute of Science, POB 26, Rehovot 7610001, Israel}
\email{ran.tessler@weizmann.ac.il}

\begin{abstract}
We lay the foundation for a version of $r$-spin theory in genus zero for Riemann surfaces with boundary. In particular, we define the notion of $r$-spin disks, their moduli space, and the Witten bundle, we show that the moduli space is a compact smooth orientable orbifold with corners, and we prove that the Witten bundle is canonically relatively oriented relative to the moduli space.  In the sequel to this paper, we use these constructions to define open $r$-spin intersection theory and relate it to the Gelfand--Dickey hierarchy, thus providing an analogue of Witten's $r$-spin conjecture in the open setting.
\end{abstract}

\maketitle

\setcounter{tocdepth}{1}
\section{Introduction}
On a smooth marked curve $(C;z_1, \ldots, z_n)$, an $r$-spin structure is a line bundle $S$ together with an isomorphism
\[S^{\otimes r} \cong \omega_{C}\left(-\sum_{i=1}^n a_i[z_i]\right),\]
where $a_i \in \{0,1,\ldots, r-1\}$.  There is a natural compactification $\M_{g,\{a_1, \ldots, a_n\}}^{1/r}$ of the moduli space of $r$-spin structures on smooth curves, and this space admits a virtual fundamental class $c_W$ known as Witten's class. In genus zero, Witten's class is defined by
\begin{gather}\label{eq:Witten's class}
c_W:= e((R^1\pi_*\mathcal{S})^{\vee}),
\end{gather}
where $\pi: \mathcal{C} \rightarrow \M_{0,\{a_1, \ldots, a_n\}}^{1/r}$ is the universal curve and $\mathcal{S}$ the universal $r$-spin structure.  In higher genus, on the other hand, $R^1\pi_*\mathcal{S}$ may not be a vector bundle, and there are several (all much more intricate) versions of the definition of Witten's class \cite{PV,ChiodoWitten,Moc06,FJR,CLL}.

Given any of these definitions, one defines the {\it closed $r$-spin intersection numbers} by
\begin{gather}\label{eq:closed r-spin}
\<\tau^{a_1}_{d_1}\cdots\tau^{a_n}_{d_n}\>^{\frac{1}{r},c}_g:=r^{1-g}\int_{\M^{1/r}_{g,\{a_1, \ldots, a_n\}}} \hspace{-1cm} c_W \cap \psi_1^{d_1} \cdots \psi_n^{d_n},
\end{gather}
where $\psi_1, \ldots, \psi_n \in H^2(\M^{1/r}_{g,\{a_1, \ldots, a_n\}})$ are the first Chern classes of the cotangent line bundles at the $n$ marked points.  This theory has received a great deal of attention in recent years; for example, it led to a proof of a conjecturally complete set of tautological relations on $\M_{g,n}$ \cite{PPZ}, and it is a special case of Fan--Jarvis--Ruan--Witten theory \cite{FJR} as well as the gauged linear sigma model \cite{FJRGLSM}.  For our purposes, perhaps the most interesting feature of $r$-spin theory was proven by Faber--Shadrin--Zvonkine \cite{FSZ10}: after a simple change of variables, the generating function of the closed $r$-spin intersection numbers becomes a tau-function of the $r$th Gelfand--Dickey hierarchy.  This statement generalizes Witten's celebrated conjecture (proven by Kontsevich) regarding the generating function of $\psi$-integrals on $\M_{g,n}$.

A different direction in which the intersection theory of $\M_{g,n}$ can be generalized is the consideration of Riemann surfaces with boundary.  This work was undertaken by Pandharipande, Solomon, and the third author in \cite{PST14}, in which a moduli space $\M_{0,k,l}$ was constructed that parameterizes tuples $(\Sigma; x_1, \ldots, x_k; z_1, \ldots, z_l)$ where $\Sigma$ is a stable disk, $x_i \in \d\Sigma$ are boundary marked points, and $z_j \in \Sigma \setminus \d\Sigma$ are internal marked points.  Furthermore, intersection numbers on $\M_{0,k,l}$ were constructed, which can be viewed as integrals of $\psi$-classes at the internal marked points. This construction was extended in \cite{ST1} to all genera, yielding a generating function $F^o$ of open intersection numbers. In order to define the extension, Solomon and the third author introduced graded $2-$spin structures and proved that the moduli of $2-$spin surfaces with boundary is canonically oriented. The open analogue of Witten's conjecture was proved by the first and third authors in \cite{BT17,Bur16,Tes15}, relating $\exp(F^o)$ to the wave function of the KdV hierarchy.

Combining $r$-spin theory with open theory, one might hope to define {\it open $r$-spin theory} and generalize Witten's conjecture to this setting.  In order to do so even in genus zero, though, one first must define an appropriate open $r$-spin moduli space~$\M_{0,k,\{a_1, \ldots, a_l\}}^{1/r}$ and an open analogue of Witten's bundle $(R^1\pi_*\mathcal{S})^{\vee}$.  We carry out the construction of such a moduli space in this paper, referring to its objects as graded $r$-spin disks.

\subsection{The moduli space of graded $r$-spin disks}

Graded $r$-spin disks are defined, roughly, as follows.  Let $C$ be an orbifold curve equipped with an involution  $\phi: C \rightarrow C$ that realizes the coarse underlying curve $|C|$ as a union of two Riemann surfaces $\Sigma$ and $\overline{\Sigma}$ (where $\overline{\Sigma}$ is obtained from $\Sigma$ by reversing the complex structure) glued along their common boundary:
\[|C| = \Sigma \cup_{\d\Sigma} \overline{\Sigma}.\]
Let $z_1, \ldots, z_l \in \Sigma \setminus \d\Sigma$ be a collection of internal marked points, let $\overline{z}_i:= \phi(z_i) \in \overline{\Sigma}$ be their conjugates, and let $x_1, \ldots, x_k \in \d\Sigma$ be a collection of boundary marked points.  On such a curve, a {\it graded $r$-spin structure} with {\it twists} $a_1, \ldots, a_l$ is an orbifold line bundle $S$ on $C$ together with an isomorphism
\[|S|^{\otimes r} \cong \omega_{|C|} \otimes \O\left(-\sum_{i=1}^l a_i[z_i] - \sum_{i=1}^l a_i[\overline{z}_i] - \sum_{j=1}^k (r-2)[x_j]\right)\]
on the coarse underlying curve $|C|$, an involution $\widetilde{\phi}: S \rightarrow S$ lifting $\phi$, and a certain orientation of $\left(S|_{\d\Sigma \setminus \{x_j\}}\right)^{\widetilde\phi}$ that we refer to as a {\it grading}.  In what follows, we prove that there exists a moduli space $\M_{0,k,\{a_1, \ldots, a_l\}}^{1/r}$ of graded $r$-spin disks with twists $a_1, \ldots, a_l$, and that this moduli space is a compact, orientable, effective, smooth orbifold with corners.

On $\M_{0,k,\{a_1, \ldots, a_l\}}^{1/r}$, there is an {\it open Witten bundle}, defined roughly as
\[
\mathcal{W}: = (R^0\pi_*(\mathcal{S}^{\vee} \otimes \omega_{\pi}))_+ = (R^1\pi_*\mathcal{S})^{\vee}_-,
\]
where ``$+$" denotes the space of $\widetilde{\phi}$-invariant sections and ``$-$" the space of $\widetilde{\phi}$-anti-invariant sections.
There are also cotangent line bundles $\mathbb{L}_1 \ldots, \mathbb{L}_l$ at the internal marked points.  We define these bundles carefully below and explore their behavior under forgetful morphisms and restriction to boundary strata.

It is straightforward to show that the cotangent line bundles have canonical complex orientations. The open Witten bundle, on the other hand, is a real vector bundle and hence it is not clear that it is orientable at all. One of the main results of this paper is that not only is $\mathcal{W}$ orientable, but it carries a \emph{canonical relative orientation} relative to the moduli space; the grading plays a central role in the construction of this canonical orientation. 
We also analyze the behavior of the canonical orientation under restriction to boundary strata.

\subsection{Companion works}

This paper lays the foundations for the sequel \cite{BCT2}, in which we use the construction of $\M_{0,k,\{a_1, \ldots, a_l\}}^{1/r}$ and its associated bundles to define genus-zero open $r$-spin intersection numbers and prove the open $r$-spin version of Witten's conjecture.  In particular, in \cite{BCT2}, we calculate all open $r$-spin numbers and prove an explicit relationship between their generating function and the genus-zero part of the Gelfand--Dickey wave function.  In addition to verifying the generalization of Witten's conjecture in genus zero, this leads to a conjecture for the higher-genus intersection numbers.

The content of \cite{BCT2} also illuminates an intriguing connection between open $r$-spin theory and an extension of closed $r$-spin theory, in which one allows a single marked point with twist $-1$.  We define this ``closed extended $r$-spin theory" carefully in the companion paper \cite{BCT_Closed_Extended} to this work, and in \cite{BCT2}, we make the correspondence between the two theories precise.

\subsection{Plan of the paper}

The structure of the current paper is as follows.  In Section \ref{sec:objects}, we define graded $r$-spin disks, and in Section \ref{sec:moduli}, we describe their moduli space, its orbifold structure, and its orientation.  Section \ref{sec:bundles} contains the definition of the cotangent line bundles $\CL_i$ and the open Witten bundle $\mathcal{W}$, as well as an investigation of the behavior of these bundles under certain key morphisms.  Finally, in Section~\ref{sec:or}, we establish the canonical relative orientation of $\cW$ and analyze its behavior under the relevant morphisms.

\subsection{Acknowledgements}

The authors would like to thank R. Pandharipande, D. Ross, J. Solomon, and E. Witten for interesting discussions related to open $r$-spin theory. The authors are grateful to J. Gu\'er\'e for enlightening discussions concerning the Ramond sector, and to A. Netser-Zernik for his enriching explanations about orbifolds with corners.

A.~B. received funding from the European Union's Horizon 2020 research and innovation programme under the Marie Sk\l odowska-Curie grant agreement No 797635 and was also supported by grant ERC-2012-AdG-320368-MCSK in the group of Rahul Pandharipande at ETH Zurich and grant RFFI-16-01-00409. E.~C. was supported by NSF DMS grant 1810969.  R.T. was supported by a research grant from the Center for New Scientists of Weizmann Institute, by Dr. Max R\"ossler, the Walter Haefner Foundation, and the ETH Z\"urich Foundation, and by the ISF (grant No. 335/19).

\section{Graded $r$-spin disks}
\label{sec:objects}
We denote by $[n]$ the set $\{1,2, \ldots, n\}$. Write $2^\N_{*}=2^{\N}\setminus\{\emptyset\}$, where $\N=\Z_{\ge 1}$. Throughout what follows, a {\it marking} of a set $A$ is a function
\[m: A \rightarrow \{0\}\cup2^\N_{*}\]
such that, for all distinct $a,a' \notin m^{-1}(0)$, we have $m(a) \cap m(a') = \emptyset$.  A marking is \emph{strict} if $0$ is not in its image.

Given a marking, we identify elements of $A\setminus m^{-1}(0)$ with their images in $2^\N_{*}$, and if the image is a singleton, we identify it with an element of $\N$.  Such functions are used in what follows to label the marked points on a curve; the possibility of marking some points by $0$ or with a set is desired to handle marked points that arise via normalization of a nodal curve.

\subsection{Smooth $r$-spin surfaces}
\label{subsec:smooth}

Recall that an {\it orbifold Riemann surface} is a smooth, proper, possibly disconnected, one-dimensional Deligne--Mumford stack over $\C$.  We sometimes refer to such a surface as {\it closed}, to distinguish it from the Riemann surfaces considered below that may have boundary.

A (smooth) {\it marked orbifold Riemann surface with boundary} is a tuple
\[(C, \phi, \Sigma, \{z_i\}_{i \in I}, \{x_j\}_{j \in B}, m^I, m^B),\]
in which:
\begin{enumerate}[(i)]
\item $C$ is a (closed) orbifold Riemann surface;
\item $\phi: C \rightarrow C$ is an anti-holomorphic involution ({\it conjugation}) that realizes the coarse underlying Riemann surface $|C|$ topologically as two Riemann surfaces $\Sigma$ and $\bar{\Sigma}$ (where $\bar{\Sigma}$ is obtained from $\Sigma$ by reversing the complex structure) glued along their common boundary $\d\Sigma=\d\bar\Sigma= \text{Fix}(|\phi|)$:
\[|C| = \Sigma \cup_{\d \Sigma} \overline{\Sigma};\]
\item $z_i \in C$ are a collection of distinct points (the {\it internal marked points}) labeled by the set $I$, whose images in $|C|$ lie in $\Sigma \setminus \d\Sigma$, with {\it conjugate marked points} $\bar{z}_i:= \phi(z_i)$;
\item $x_j \in \text{Fix}(\phi)$ are a collection of distinct points (the {\it boundary marked points}) labeled by the set $B$, whose images in $|C|$ lie in $\d\Sigma$;
\item the only nontrivial isotropy of $C$ occurs at the (internal, conjugate, and boundary) marked points;
\item $m^I:I\to\{0\}\cup 2^\N_{*}$ and $m^B:B\to\{0\}\cup 2^\N_{*}$ are maps such that, for any connected component $C'$ of $C$ with marked points labeled by $I'\subseteq I$ and $B'\subseteq B$, the restrictions $m^I|_{I'}$ and $m^B|_{B'}$ are markings, and whenever $C'\cap\phi(C')\neq\emptyset$, the marking $m^I|_{I'}$ is strict.
\end{enumerate}
A marked orbifold Riemann surface $C$ is \emph{stable} if each genus-zero connected component has at least three marked points (including conjugate marked points) and each genus-one connected component has at least one marked point.

We observe that the choice of a preferred half $\Sigma \subset |C|$, which is part of the data of an orbifold Riemann surface with boundary, endows $\d\Sigma$ with a canonical orientation.  In what follows, we typically suppress $\phi$ from the notation and write $\overline{x}$ for $\phi(x)$ when $x$ lies in the preimage of $\Sigma$ in $C$.

An {\it isomorphism} of marked orbifold Riemann surfaces with boundary
\[(C_1, \phi_1, \Sigma_1, \{z_{1,i}\}_{i \in I}, \{x_{1,j}\}_{j \in B}, m_1^{I_1}, m_1^{B_1}) \cong (C_2, \phi_2, \Sigma_2, \{z_{2,i}\}_{i \in I_2}, \{x_{2,j}\}_{j \in B_2}, m_2^{I_2}, m_2^{B_2})\]
consists of an isomorphism $s: C_1\rightarrow C_2$ and bijections $f^I: I_1 \rightarrow I_2,~f^B: B_1 \rightarrow B_2$ such that
\begin{enumerate}[(i)]
\item $s \circ \phi_1 = \phi_2 \circ s$,
\item $s(\Sigma_1) = \Sigma_2$,
\item $s(z_{1,j}) = z_{2,f^B(j)}$ for all $j \in B_1$ and $s(x_{1,i}) = x_{2,f^I(i)}$ for all $i \in I_1$,
\item $m_1^{I_1} = m_2^{I_2} \circ f^I$ and $m_1^{B_1} = m_2^{B_2} \circ f^B$.
\end{enumerate}

Let $C$ be a marked orbifold Riemann surface with boundary in which every marked point has isotropy group $\Z/r\Z$.  We denote
\begin{equation}
\label{eq:omegalog}
\omega_{C, \log} := \rho^*\omega_{|C|,\log} =  \omega_{C} \left(\sum_{i \in I} [z_i] + \sum_{i \in I} [\overline{z_i}] +  \sum_{j \in B} [x_j] \right),
\end{equation}
where $\rho: C \rightarrow |C|$ is the morphism to the coarse underlying Riemann surface and $[z_i]$,  $[\overline{z}_i]$, and $[x_j]$ denote the degree-$1/r$ orbifold divisors of the marked points.  An {\it $r$-spin structure} on $C$ is an orbifold line bundle $L$ together with an isomorphism
\begin{equation}
\label{Lr}
\tau: L^{\otimes r} \cong \omega_{C, \log}
\end{equation}
and an involution $\widetilde{\phi}: L \rightarrow L$ lifting $\phi$ such that $\widetilde{\phi}^{\otimes r}$ agrees under $\tau$ with the involution on $\omega_{C, \log}$ induced by $\phi$.  We denote by $\text{mult}_p(L)$ the {\it multiplicity} of $L$ at a point $p$, which is defined as the integer $m \in \{0,1,\ldots, r-1\}$ such that the local structure of the total space of $L$ near $p$ is $[\C^2/(\Z/r\Z)]$ with action
\[\zeta\cdot(x,v) = (\zeta x,\zeta^mv)\]
by the canonical generator $\zeta$ of the isotropy group.

Associated to an $r$-spin structure $L$ on $C$, there is a unique {\it twisted $r$-spin structure} $S$ on $C$, defined as the complex line bundle
\begin{equation}
\label{eq:S}
S := L \otimes \O\left(-\sum_{i \;|\;\text{mult}_{z_i}(L) = 0} r[z_i] - \sum_{i \;|\; \text{mult}_{\overline{z}_i}(L) = 0} r[\overline{z_i}] - \sum_{j \;|\; \text{mult}_{x_j}(L) = 0} r[x_j]\right).
\end{equation}
This bundle satisfies
\begin{equation}\label{eq:Snew}
S^{\otimes r} \cong \rho^*\left(\omega_{|C|, \log}\otimes \O\left(-\sum_{i \;|\;\text{mult}_{z_i}(L) = 0} r[z_i] - \sum_{i \;|\; \text{mult}_{\overline{z}_i}(L) = 0} r[\overline{z_i}] - \sum_{j \;| \; \text{mult}_{x_j}(L) = 0} r[x_j]\right)\right),
\end{equation}
in which, on the right-hand side, $[z_i]$, $[\overline{z}_i]$, and $[x_j]$ now denote the corresponding divisors on $|C|$.  It follows that the coarse underlying bundle $|S|:= \rho_*S$ satisfies
\begin{equation}
\label{Sr}
|S|^{\otimes r} \cong \omega_{|C|} \otimes \O\left(-\sum_{i \in I} a_i [z_i] -\sum_{i \in I} a_i [\overline{z_i}] - \sum_{j \in B} b_j [x_j]\right)
\end{equation}
with $a_i, b_j \in \{0,1,\ldots, r-1\}$
defined by the requirement that
\[a_i \equiv \text{mult}_{z_i}(L) -1\mod r, \;\;\;\;\;\; b_j \equiv \text{mult}_{x_j}(L)-1\mod r.\]
We refer to the numbers $a_i$ and $b_j$ as internal and boundary {\it twists}.  In fact, pushforward under $\rho$ defines an equivalence of categories between bundles $S$ satisfying \eqref{eq:Snew} and bundles $|S|$ satisfying \eqref{Sr} for some choice of $a_i$ and $b_j$; hence, in particular, the data of $|S|$ and its twists is equivalent to the data of an $r$-spin structure on $C$.  See the appendix of \cite{CR} for a more detailed discussion of this equivalence.

Let
\begin{equation}
\label{eq:J}
J:= S^{\vee} \otimes \omega_{C},
\end{equation}
which inherits an involution that we also denote by $\widetilde{\phi}$.  Using the fact that
\[\omega_C = \rho^*\omega_{|C|} \otimes \O\left(\sum_{i \in I} (r-1)[z_i]  + \sum_{i \in I} (r-1)[\overline{z_i}] + \sum_{j \in B} (r-1)[x_j]\right),\]
which follows directly from \eqref{eq:omegalog}, one shows easily that
\[J = \rho^*\bigg(|S|^{\vee}\otimes \omega_{|C|}\bigg) \otimes \O\left(\sum_{i \in I} \text{mult}_{z_i}(J)[z_i] + \sum_{i \in I} \text{mult}_{\overline{z_i}}(J)[\overline{z_i}] + \sum_{j \in B} \text{mult}_{x_j}(J)[x_j]\right),\]
and hence
\begin{equation}
\label{eq:|J|}
|J| = |S|^{\vee} \otimes \omega_{|C|}.
\end{equation}
In particular, since $|J|$ is a bundle on a non-orbifold curve and hence has integral degree, the following observation is immediate.

\begin{obs}
\label{obs:open_rank1}
The twists $a_i, b_j$ for a twisted $r$-spin structure on a smooth marked orbifold Riemann surface with boundary satisfy the following congruence condition:
\begin{equation}\label{eq:open_rank1}
e:=\frac{2\sum a_i + \sum b_j+(g-1)(r-2)}{r}\in \Z,
\end{equation}
where $g$ is the genus of $C$.  (In case $C$ is disconnected, we define $g=\sum_i g(C_i)-l+1$, where $C_1, \ldots, C_l$ are the connected components of $C$.)
\end{obs}

All of the above also works in the more familiar setting of closed marked orbifold Riemann surfaces $(C, \{z_i\}_{i \in I}, m^I)$.  However, in this case, it is important in what follows to allow the possibility of limited $-1$ twists; see Observation \ref{obs:twist-1}.  Thus, we define a {\it closed twisted $r$-spin structure} as a closed marked orbifold Riemann surface equipped with an orbifold line bundle $S$ satisfying
\[S^{\otimes r} \cong \rho^*\left(\omega_{|C|, \log}\otimes \O\left(-\sum_{i \in I_0} r[z_i]\right)\right),\]
where $I_0\subseteq I$ is a subset of the marked points such that $\text{mult}_{z_i}(S) = 0$ for all $i \in I_0$ and each connected component of $C$ contains at most one marked such point $z_i$ with $i \notin I_0$.  In this case, one has
\[|S|^{\otimes r} \cong \omega_{|C|} \otimes \O\left(-\sum_{i \in I} a_i[z_i]\right)\]
with $a_i \in \{-1,0, \ldots, r-1\}$ and $a_i = -1$ for at most one marked point $z_i$ in each connected component of $C$.  Analogously to Observation~\ref{obs:open_rank1}, we have:

\begin{obs}\label{obs:closed}
The twists $a_i$ for a closed twisted $r$-spin structure satisfy
\begin{equation}\label{eq:close_rank1}
\frac{\sum a_i +(g-1)(r-2)}{r}\in \Z.
\end{equation}
\end{obs}

An {\it isomorphism} of $r$-spin structures consists of an isomorphism of marked orbifold Riemann surfaces with boundary, as defined above, together with an isomorphism $\widetilde{s}: s^*L_2 \cong L_1$ of the spin bundles such that:
\begin{enumerate}[(i)]
\item $\widetilde{s}$ commutes with the involutions, i.e. $\widetilde{s} \circ \widetilde{\phi}_1 = \widetilde{\phi}_2 \circ \widetilde{s}$;
\item the diagram
\[\xymatrix{
s^*L_2^{\otimes r} \ar[r]^{\tau_1}\ar[d]_{\widetilde{s}^{\otimes r}} & s^*\omega_{C_2,\log}\ar[d]\\
L_1^{\otimes r} \ar[r]^{\tau_2} & \omega_{C_1,\log}
}\]
commutes, where the right-hand vertical arrow is induced by $s$.
\end{enumerate}
Similarly, an isomorphism of twisted $r$-spin structures is an isomorphism $\widetilde{s}: s^*S_2 \cong S_1,$ satisfying the analogues of (i) and (ii).

\subsection{Gradings on smooth $r$-spin disks}

From here forward, we restrict to Riemann surfaces with boundary in which each connected component has genus zero, which we refer to as {\it disks}.  The definitions extend to higher genus, as well, but they are not needed for the current work.  We denote by $(C, \phi, \Sigma, \{z_i\}_{i \in I}, \{x_j\}_{j \in B}, m^I, m^B)$ a smooth marked orbifold disk.  Note that we allow the case $\d\Sigma=\emptyset$, and that if a smooth marked orbifold disk is connected (meaning that $C$ is connected), then $\d\Sigma\ne\emptyset$.

Let $A \subseteq \d\Sigma \setminus \{x_j\}_{j \in B}$ be an open subset (where we identify the boundary marked points with their images in $|C|$), and let $s \in \Gamma(A, \omega_{|C|})^{\phi}$ be a section fixed under the fiberwise involution on $\omega_{|C|}|_A$ induced by $\phi$.  We call $s$ {\it positive} if, for any $p \in A$ and any tangent vector $v \in T_p(\d \Sigma)$ in the direction of orientation, we have $\langle s(p),v\rangle > 0$.  A similar notion of positivity applies to $\phi$-fixed sections of any tensor power of $\omega_{|C|}$ over $A$.
\begin{definition}\label{def:lifting}
Given a twisted $r$-spin structure $S$, a  {\it lifting of $S$ over $A$} is a continuous, $\widetilde\phi$-invariant section
\[v \in \Gamma\left(A, |S|^{\widetilde\phi}\right)\]
such that the image of $v^{\otimes r}$ under the map on sections induced by the injection
\begin{equation}
\label{vr}
|S|^{\otimes r} \rightarrow \omega_{|C|}
\end{equation}
is positive.
A {\it lifting of $J$ over $A$} is a $\widetilde\phi$-invariant section
\[w \in \Gamma\left(A, |J|^{\widetilde\phi}\right)\] for which there exists a lifting $v$ of $S$ over $A$ with $\langle w,v\rangle \in \Gamma(A,\omega_{|C|})$ positive on $A$, where $\langle-,-\rangle$ is induced by the natural pairing between $|S|^{\vee}$ and $|S|$.  We consider two liftings $v$ and $v'$ (of either $S$ or $J$) {\it equivalent} if $v = c v'$ for a continuous function $c: A \rightarrow \R^+$. We write $[v]$ for the equivalence class of $v.$
\end{definition}

Observe that a twisted $r$-spin structure admits a lifting of $S$ over $A$ precisely if it admits a lifting of $J$ over $A$.  Moreover, there is a bijection between equivalence classes of liftings of $S$ and of $J$, in which $[v]$ corresponds to $[w]$ if $\langle w,v \rangle$ is everywhere positive for all representatives $v$ and $w$ of $[v]$ and $[w]$.  If a twisted $r$-spin structure admits a lifting (of either $S$ or $J$) over all of $\d \Sigma \setminus \{x_j\}_{j \in B}$, we call it {\it compatible}.

We now define the notion of a lifting ``alternating" at a boundary marked point.
\begin{definition}
\label{def:legal}
Let $w$ be a lifting of $J$ over $\partial\Sigma\setminus\{x_i\}_{i\in B}$.  A boundary marked point $x_j$ is said to be \emph{illegal} if there exists a lifting $w'\in[w]$ that can be continuously extended to $x_j$ without vanishing.\footnote{One can define legality in the exact same manner for liftings of $S$, and it is straightforward to see that $x_j$ is legal for the class $[w]$ of liftings of $J$ precisely if it is legal for the corresponding class $[v]$ of liftings of $S$.}  If $x_j$ is not illegal, we say that it is {\it legal} and that $w$ {\it alternates} at $x_j$.
\end{definition}

It is immediate that legality is well-defined under equivalence of $v$.  Furthermore, legality and compatibility are closely related to the twists:
\begin{prop}
\label{prop:compatibility_lifting_parity}
\begin{enumerate}[(1)]
\item\label{it:compatibility_odd} When $r$ is odd, any twisted $r$-spin structure is compatible, and there is a unique equivalence class of liftings.
\item\label{it:compatibility_even} When $r$ is even, the boundary twists $b_j$ in a compatible twisted $r$-spin structure must be even.  Whenever the boundary twists are even, either the $r$-spin structure is compatible or it becomes compatible after replacing $\widetilde\phi$ by $\xi \circ \widetilde\phi \circ \xi^{-1}$ for $\xi$ an $r$th root of $-1$, which yields an isomorphic $r$-spin structure.
\item\label{it:lifting and parity_odd} Suppose $r$ is odd and $v$ is a lifting over a punctured neighborhood of a boundary marked point $x_j$.  Then $x_j$ is legal if and only if its twist is odd.
\item\label{it:lifting and parity_even}
Suppose $r$ is even.  If a lifting over $\d \Sigma \setminus \{x_j\}_{j \in B}$ alternates precisely at a subset $D \subset \{x_j\}_{j \in B}$, then
\begin{equation}
\label{parity}
\frac{2\sum a_i + \sum b_j+2}{r} \equiv |D| \mod 2.
\end{equation} If~\eqref{parity} holds, then there exist exactly two liftings (up to equivalence) that alternate precisely at $D \subset \{x_j\}_{j \in B}$, one of which is the negative of the other.
\end{enumerate}
\end{prop}
\begin{proof}
We begin by choosing trivializations
\begin{equation}
\label{triv}
\omega_{|C|}^{\phi}\bigg|_{\d \Sigma} \cong \d\Sigma \times \R
\end{equation}
and
\[\omega_{|C|}\left(-\sum_{i \in I} a_i [z_i] -\sum_{i \in I} a_i [\overline{z_i}] - \sum_{j \in B} b_j [x_j]\right)^{\phi}\bigg|_{\d \Sigma \setminus \{x_j\}_{j \in B}} \cong ( \d \Sigma \setminus \{x_j\}_{j \in B} ) \times \R\]
such that a section of either of these bundles is positive precisely if its image lies in the positive ray $\R^+ \subset \R$ in each fiber.
Let $I_j$ be the connected component of $\partial\Sigma\setminus\{x_j\}_{j\in B}$ defined by the property that $x_j$ is the left endpoint of the closure of $I_j$ with respect to the orientation of $\partial\Sigma$.  Since the $r$th tensor power of a $\widetilde\phi$-invariant section of $|S|$ is $\phi$-invariant, one can see that on each $I_j$, either there is a section $v_j\in\Gamma(I_j,|S|^{\widetilde{\phi}})$ with $v_j^{\otimes r}$ mapping to $1$ under the composition of~\eqref{vr} and \eqref{triv}, or there is a section $v_j$ with $v^{\otimes r}_j$ mapping to~$-1.$

Now, for the first item, suppose $r$ is odd.  Then, by replacing $v_j$ with $-v_j$ if necessary, we find for each $I_j$ a section $v_j\in\Gamma(I_j,|S|^{\widetilde{\phi}})$ such that $v_j^{\otimes r}$ maps to $1$ under the composition of \eqref{vr} and \eqref{triv}.  Thus, there is always a lifting.

If $r$ is even, then if there is a section $v_j\in\Gamma(I_j,|S|^{\widetilde{\phi}})$ with $v_j^{\otimes r}$ mapping to the constant section $\epsilon_j=\pm 1$, then $(-v_j)^{\otimes r}$ maps to $(-1)^r\epsilon_j=\epsilon_j,$ and there is no real section that maps to $-\epsilon_j.$ We see that when $r$ is even, the structure is compatible precisely if $\epsilon_j=1$ for all $j$.

Suppose the elements of $B$ are enumerated cyclically so that $x_j$ follows $x_{j-1}$ in the cyclic order of $j \in \{1, \ldots, |B|\}$ around the boundary.  Choose a local coordinate $x$ on $\d\Sigma$ centered around a boundary marked point $x_j$.  Then the map on local sections induced by \eqref{vr} is multiplication by the local section $x^{b_j}$  of $\O(b_j[x_j])$.  The section
\[x^{-b_j} \in \Gamma\left(U\setminus\{x_j\}, \;\omega_{|C|}\left(-\sum_{i \in I} a_i [z_i] -\sum_{i \in I} a_i [\overline{z_i}] - \sum_{j \in B} b_j [x_j]\right)^{\phi}\right),\]
defined in a small punctured neighborhood $U \setminus \{x_j\}$ without any other markings, extends to all of $U$ as a nowhere-vanishing local section that we denote by $s$.  In addition, there exists a section $v\in\Gamma(|S|^{\widetilde \phi}|_{U\cap I_j})$ with $v^{\otimes r}$ mapping to $\epsilon_j s$.  Since $s$ does not vanish, $v$ can be extended to all of $U$ without vanishing in such a way that $v^{\otimes r}$ maps to $s.$  Given that $x^{-b_j}$ changes sign at $x_j$ precisely when $b_j$ is odd, we see that
\[\epsilon_j = (-1)^{b_j}\epsilon_{j-1}.\]
Thus, if each $\epsilon_j$ equals $1$, then all $b_j$ must be even.  In case all $b_j$ are even, it could still be the case that $\epsilon_j=-1$ for each $j$, but then replacing $\widetilde{\phi}$ by $\xi \circ \widetilde\phi \circ \xi^{-1}$ reverses the notion of positivity and hence ensures that the spin structure is compatible. The second item follows.

For the third item, by the considerations of the previous two items, for $r$ odd there is a unique $r$th root
\[ v(x) = x^{-b_j/r}.\]
Its $r$th power $v^{\otimes r}$ does not change sign after crossing $x_j$ (and hence gives rise to a lifting) precisely if $b_j$ is even.

The last item is a consequence of the real zero count for a section of $|J|$.  Using \eqref{eq:|J|},
\begin{align*}
\deg(|J|) &= \frac{2\sum a_i +\sum b_i +2-2r}{r}\\
&\equiv \frac{2\sum a_i +\sum b_i +2}{r} \mod 2.
\end{align*}
Viewing the degree as the number of zeroes minus the number of poles of a meromorphic section, one sees also that
\[\deg(|J|) \equiv \deg\left(|J|\big|_{\d \Sigma} \right) \mod 2,\]
since we may choose a $\widetilde\phi$-invariant meromorphic section, and for such sections non-real zeroes and poles come in conjugate pairs.

The number of zeroes minus the number of poles of a $\widetilde\phi$-invariant meromorphic section of $|J|\big|_{\partial\Sigma}$ is even precisely if the real subbundle $|J|^{\widetilde{\phi}}$ on ${\partial\Sigma}$ is orientable.  The orientability of $|J|^{\widetilde{\phi}}\to{\partial\Sigma}$, on the other hand, can be deduced from the number of legal marked points on $\partial\Sigma$: a section $w$ as in Definition~\ref{def:legal} gives a trivialization of $|J|^{\widetilde{\phi}}|_{\partial\Sigma}$ away from boundary marked points, and the transition functions between these trivializations are sign-reversing exactly at the legal marked points. Equation \eqref{parity} follows. The same considerations also allow us, assuming~\eqref{parity} holds, to construct a lifting that alternates precisely at the points of $D$.  The equivalence class of such a lifting is determined by choosing the lifting at any unmarked boundary points, and there are exactly two such choices.
\end{proof}

The last paragraph of the above proof also yields:
\begin{cor}\label{cor:grading_and_orientation_of_J}
The bundle $|J|^{\widetilde\phi}$ on $\d\Sigma$ is orientable if and only if the number of legal marked points on $\partial\Sigma$ is even.
\end{cor}

If $S$ is a compatible twisted spin structure on a smooth marked orbifold disk, then we define a {\it grading} of $S$ as a lifting $w$ for $J$ over $\partial\Sigma\setminus\{x_i\}_{i\in B}$ that alternates at every boundary marked point.  We apply the same notion of equivalence as above to the grading $w.$

Given this, we can now define our main objects of study, in the case where $C$ is smooth:

\begin{definition}
A {\it graded $r$-spin structure} on a smooth marked orbifold disk is a compatible twisted spin structure $S$ in which all boundary twists are $r-2$, together with an equivalence class $[w]$ of gradings.
\end{definition}

We often abuse terminology somewhat, referring to the choice of equivalence class simply as the ``grading."  The key point about the grading is that it induces a notion of positivity---that is, a choice of preferred real ray in $J^{\widetilde\phi}_q$ at any unmarked boundary point $q$.

A smooth marked orbifold disk together with a twisted spin structure and a grading is called a {\it smooth graded $r$-spin disk}.  An {\it isomorphism} of graded $r$-spin disks consists of an isomorphism of twisted spin structures, as defined above, such that the image of the equivalence class of gradings $[w_1]$ on $C_1$ under the map on sections induced by $s$ and $\widetilde{s}$ is the equivalence class $[w_2]$.

If $\text{Fix}(\phi) = \emptyset$, the above notion of grading is vacuous, but we require an additional datum:

\begin{definition}\label{def:graded_sphere}
A {\it smooth graded $r$-spin sphere} is a smooth $r$-spin sphere together with a choice of a distinguished marking $z_i$, referred to as the {\it anchor}, such that
\begin{enumerate}[(1)]
\item $m^I(z_i) = 0$ and $z_i$ is the only marked point marked zero;
\item if there is a marked point with twist $-1$, then it must be the anchor;
\item if the twist $a_i$ of the anchor is $r-1$, we have a map $\tau':\left(|S|\otimes \O\left([z_i]\right)\right)^{\otimes r}\big|_{z_i}\to \C$ defined as the composition
\[\left(|S|\otimes \O\left([z_i]\right)\right)^{\otimes r}\big|_{z_i}\rightarrow \omega_{|C|}([z_i])\big|_{z_i}\cong\C,\] where the second identification is the residue map.
In this case, we also fix an involution $\widetilde{\phi}$ on the fiber $\left(|S|\otimes \O\left([z_i]\right)\right)_{z_i}$ and require it to satisfy two properties: first,
\[\tau'(\widetilde{\phi}(v)^{\otimes r})=-\overline{\tau'(v^{\otimes r})}\]
for all $v\in\left(|S|\otimes \O\left([z_i]\right)\right)_{z_i}$, where $w\mapsto\overline{w}$ is the standard conjugation; and second,
\[\left\{\tau'(v^{\otimes r})\; | \; v\in \left(|S|\otimes \O\left([z_i]\right)\right)^{\widetilde\phi}_{z_i}\right\}\supseteq i\R_+,\]where $i$ is the root of $-1$ in the upper half plane. Finally, we pick a \emph{positive direction} on $\left(|S|\otimes \O\left([z_i]\right)\right)^{\widetilde\phi}_{z_i},$ meaning a connected component $V$ of $\left(|S|\otimes \O\left([z_i]\right)\right)^{\widetilde\phi}_{z_i}\setminus\{0\}$ such that $\tau'(v^{\otimes r})\in i\R_+$ for any $v \in V$.

Since
\[|J|_{z_i}\otimes \left(|S|\otimes \O\left([z_i]\right)\right)_{z_i}\cong \omega_{|C|}([z_i])\big|_{z_i},\]
the involution $\widetilde\phi$ induces a conjugation on the fiber $|J|_{z_i}$, also denoted by $\widetilde\phi$, by the requirement that
\[\langle \widetilde\phi(w),\widetilde\phi (v)\rangle = -\overline{\langle w, v\rangle}\]
under the identification of $\omega_{|C|}([z_i])_{z_i}$ with $\C$.  Similarly, a positive direction is defined on $|J|^{\widetilde\phi}_{z_i}$ by the requirement that $w\in |J|_{z_i}^{\widetilde\phi}$ is positive if, for any positive $v\in \left(|S|\otimes \O\left([z_i]\right)\right)_{z_i}^{\widetilde\phi}$, one has
\[\langle w,v\rangle\in i\R_+.\]
\end{enumerate}
An {\it isomorphism} of graded $r$-spin spheres is an isomorphism of $r$-spin spheres that preserves the anchor and, in the case where the anchor has twist $r-1$, also preserves the involution $\widetilde{\phi}$ and the positive direction.
\end{definition}

\begin{rmk}
The anchor is an auxiliary tool that is useful when sphere components arise from normalization.  The anchor of a sphere component should be thought of as the half-node at which the component met a disk component, or met a simple path of sphere components connecting it to a disk component. The additional requirement in the case of twist $r-1$ arises from the notion of grading at contracted boundary components; see the discussion preceding Definition \ref{def:stablegraded} below.
\end{rmk}

We observe the following parity constraint, which follows from Observation \ref{obs:open_rank1} when $r$ is odd and from item \ref{it:lifting and parity_even} of Proposition \ref{prop:compatibility_lifting_parity} when $r$ is even:

\begin{obs}\label{obs:open_rank2}
For any graded spin disk, we have
\begin{equation}\label{eq:open_rank2}
e\equiv |B|-1 \mod 2,
\end{equation}
where $e$ is defined by equation \eqref{eq:open_rank1} with $g=0$ and $b_j=r-2$ for all $j$.
\end{obs}

In both the closed (sphere) theory and the open (disk) theory, one has an existence and uniqueness result:

\begin{prop}\label{prop:existence_smooth}
If $(C, \{z_1, \ldots, z_l\})$ is a marked sphere and $a_1,\ldots,a_l \in \{-1,0, \ldots, r-1\}$ are twists satisfying \eqref{eq:close_rank1} such that at most one $a_i$ equals $-1$, then there exists a unique twisted $r$-spin structure on $C$ with twists $a_1, \ldots, a_l$.  The only automorphisms are given by scaling $S$ fiberwise by an $r$th root of unity. For any choice of $z_i$, there exists a unique graded $r$-spin structure on $(C,S)$ with anchor $z_i$, and it has no automorphisms if the anchor has twist $r-1$.

Suppose $(C,\phi,\Sigma, \{z_1, \ldots, z_l\}, \{x_1, \ldots, x_k\})$ is a connected marked disk and $a_1,\ldots,a_l \in \{0, 1, \ldots, r-1\}$, $b_1,\ldots,b_k\in\{0,1,\ldots,r-2\}$ are such that \eqref{eq:open_rank1} holds. If $r$ is even, let $D\subseteq[k]$ be an arbitrary set for which \eqref{parity} holds, and if $r$ is odd, let $D=\{i\in[k] \; |\; 2\nmid b_k\}.$ Then $C$ admits a unique $r$-spin structure with a lifting such that the internal twists are given by the integers $a_i$, the boundary twists are given by the integers $b_j$, and $D$ is the set of legal boundary marked points. The $r$-spin structure has no automorphisms that preserve the lifting.

In particular, when all $b_j$ are equal to $r-2$ and \eqref{eq:open_rank1} and \eqref{eq:open_rank2} hold, there exists a unique graded $r$-spin structure on this disk with internal twists given by the integers $a_i$, and this graded structure has no automorphisms.
\end{prop}
\begin{proof}
The first part is well-known; see, for example, the appendix of \cite{CR}. The statement regarding the grading in the closed setting is clear, except for the automorphism claim in case the anchor has twist $r-1.$ In this case, when $r$ is odd, scaling by a root of unity is not compatible with the involution $\widetilde\phi$ in the fiber of the anchor. When $r$ is even, scaling by $-1$ is compatible with the involution but not with the choice of positive direction.

Now, fix a connected marked disk and twists $a_1,\ldots,a_l$ such that \eqref{eq:open_rank1} and \eqref{parity} hold.  Note that the integrality requirement for the closed case (Observation \ref{obs:closed}) is satisfied on the closed marked Riemann surface
\[(C; \{z_1, \ldots, z_l, \overline{z}_1, \ldots, \overline{z}_l, x_1, \ldots, x_k\}),\]
so there exists a twisted $r$-spin structure $S$ on $C$.

To define a twisted $r$-spin structure in the sense of surfaces with boundary, we must produce an involution
\[\widetilde\phi: |S| \rightarrow |S|\]
lifting the involution $|\phi|$ on $|C|$.  In order to do so, choose a boundary point $p \in \d\Sigma \setminus \{x_1, \ldots, x_k\}$, and choose a vector $v \in |S|_p$ such that the image of $v^{\otimes r}$ in $\omega_{|C|}|_p,$ under the injection~\eqref{vr}, is positive.  For any other point $q \in \Sigma$ that is not a marked point, and any $u \in |S|_q$, we define $\widetilde\phi(q,u) = (\phi(q), \widetilde\phi(u))$ as follows.  Let $\widetilde\gamma$ be a path in the total space of
\[|S|\big|_{\Sigma\setminus\{z_i\}_{i\in[l]},\{x_j\}_{j\in[k]}}\]
with $\widetilde\gamma(0)=(p,v)$ and $\widetilde\gamma(1)=(q,u)$, and let $\gamma$ be the image of $\widetilde\gamma^{\otimes r}$ image in the total space of the bundle $\omega_{|C|}|_{\Sigma\setminus\{z_i\}_{i\in[l]},\{x_j\}_{j\in[k]}}$.  There is a unique path in the total space of $|S|\big|_{\phi(\Sigma)\setminus\{\bar{z}_i\}_{i\in[l]},\{x_j\}_{j\in[k]}}$, denoted by $\widetilde\phi(\widetilde\gamma)$ and starting at $(p,v)$, whose $r$th power maps to $\phi(\gamma)$.  Define $\widetilde\phi(q,u)$ to be its endpoint.

It is easy to see that the above definition is independent of choices, and the conjugation extends uniquely, up to isomorphism, to the marked points.  So far we have defined a twisted $r$-spin structure on the marked disk, and what remains is to define the grading.  For this, one can take as the lifting any smooth section $v$ of $|S|^{\widetilde\phi}|_{\d \Sigma\setminus\{x_j\}_{j\in[k]}}$ that alternates at each boundary marking; such a section exists because the parity constraint \eqref{parity} is satisfied. The lifting determines a grading $w.$

The absence of automorphisms, when $r$ is odd, is due to the involution: fiberwise multiplication by an $r$th root of unity does not commute with $\widetilde\phi$, so it is not an automorphism.  When $r$ is even, multiplication by $-1$ does commute with the involution, but it does not preserve the grading $w$, so it is again not an automorphism of the graded structure.
\end{proof}


\subsection{Stable graded $r$-spin disks}
\label{subsec:stabledisks}

Thus far, we have considered only graded $r$-spin structures on smooth curves.  In order to compactify the moduli space of such objects, however, we also need to allow the curves $C$ to be nodal.

A {\it nodal marked orbifold Riemann surface with boundary} is a tuple
\[(C, \phi, \Sigma, \{z_i\}_{i \in I}, \{x_j\}_{j \in B}, m^I, m^B),\]
defined exactly as in Section \ref{subsec:smooth} except that $C$ is a nodal, possibly disconnected, orbifold Riemann surface (curve) as in \cite{AV}, and $\text{Fix}(|\phi|)$ is locally homeomorphic near every point either to an open subset of $\R$, to the union of the coordinate axes
\[\{xy=0\} \subset \R^2,\]
or to a single point.  We recall from \cite{AV} that $C$ is only allowed nontrivial isotropy at marked points and nodes, each of which has isotropy group $\Z/r\Z$, and all nodes are required to be balanced in the sense that, in the local picture $\{xy=0\} \subset \C^2$ at the node, the action of the distinguished generator $\zeta$ of the isotropy group is given by
\[(x,y) \mapsto (\zeta x, \zeta ^{-1}y).\]

The nodes in a Riemann surface with boundary can be divided into three types:
\begin{enumerate}[(1)]
\item Internal nodes, which are nodes in the interior $\Sigma$ (together with a conjugate node in $\overline{\Sigma}$);
\item Boundary nodes, which are nodes in $\d \Sigma$, around which $\d \Sigma$ is locally homeomorphic to the union of the coordinate axes \[\{xy=0\} \subset \R^2;\]
\item Contracted boundary nodes, which are nodes arising in the limiting case where one component of $\text{Fix}(|\phi|)$ is a single point.
\end{enumerate}
The three types of nodes are illustrated in Figure \ref{fig:nodes}.

A nodal Riemann surface is \emph{stable} if every connected component of its normalization is stable.  
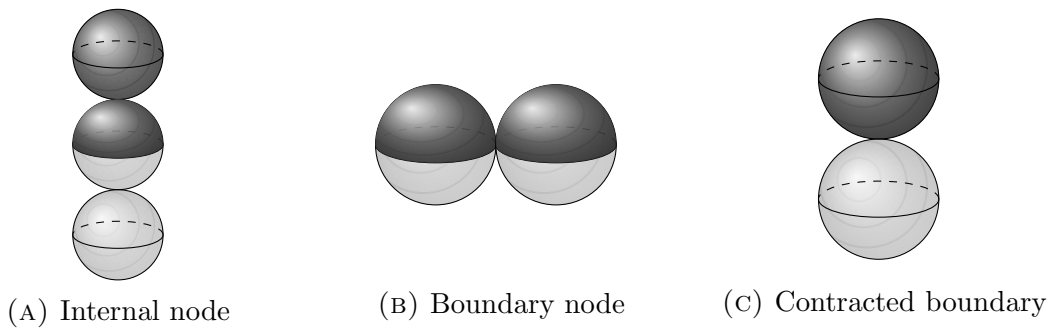
\begin{figure}
\centering
\begin{subfigure}{.3\textwidth}
  \centering

\begin{tikzpicture}[scale=0.3]
  \shade[ball color = gray, opacity = 0.5] (0,0) circle (2cm);
  \draw (0,0) circle (2cm);
  \draw (-2,0) arc (180:360:2 and 0.6);
  \draw[dashed] (2,0) arc (0:180:2 and 0.6);

  \shade[ball color = gray, opacity = 0.1] (0,-4) circle (2cm);
  \draw (0,-4) circle (2cm);
  \draw (-2,-4) arc (180:360:2 and 0.6);
  \draw[dashed] (2,-4) arc (0:180:2 and 0.6);
   \shade[ball color = gray, opacity = 0.6] (-2,-4) arc (180:360:2 and 0.6) arc (0:180:2);

    \shade[ball color = gray, opacity = 0.1] (0,-8) circle (2cm);
  \draw (0,-8) circle (2cm);
  \draw (-2,-8) arc (180:360:2 and 0.6);
  \draw[dashed] (2,-8) arc (0:180:2 and 0.6);
\end{tikzpicture}

  \caption{Internal node}
\end{subfigure}
\begin{subfigure}{.3\textwidth}
  \centering
\vspace{0.9cm}
\begin{tikzpicture}[scale=0.4]
  \shade[ball color = gray, opacity = 0.1] (0,0) circle (2cm);
  \draw (0,0) circle (2cm);
  \draw (-2,0) arc (180:360:2 and 0.6);
  \draw[dashed] (2,0) arc (0:180:2 and 0.6);
  \shade[ball color = gray, opacity = 0.6] (-2,0) arc (180:360:2 and 0.6) arc (0:180:2);

  \shade[ball color = gray, opacity = 0.1] (4,0) circle (2cm);
  \draw (4,0) circle (2cm);
  \draw (2,0) arc (180:360:2 and 0.6);
  \draw[dashed] (6,0) arc (0:180:2 and 0.6);
  \shade[ball color = gray, opacity = 0.6] (2,0) arc (180:360:2 and 0.6) arc (0:180:2);
\end{tikzpicture}
\vspace{0.9cm}

  \caption{Boundary node}
\end{subfigure}
\begin{subfigure}{.3\textwidth}
  \centering

\begin{tikzpicture}[scale=0.4]
\vspace{0.15cm}
  \shade[ball color = gray, opacity = 0.6] (0,0) circle (2cm);
  \draw (0,0) circle (2cm);
  \draw (-2,0) arc (180:360:2 and 0.6);
  \draw[dashed] (2,0) arc (0:180:2 and 0.6);

  \shade[ball color = gray, opacity = 0.1] (0,-4) circle (2cm);
  \draw (0,-4) circle (2cm);
  \draw (-2,-4) arc (180:360:2 and 0.6);
  \draw[dashed] (2,-4) arc (0:180:2 and 0.6);
\end{tikzpicture}
\vspace{0.15cm}

  \caption{Contracted boundary}
\end{subfigure}

\caption{The closed curve $|C|$, with the open disk $\Sigma$ shaded.}
\label{fig:nodes}
\end{figure}

An {\it $r$-spin structure} on a nodal marked orbifold Riemann surface with boundary is defined, exactly as in the smooth case, as a complex line bundle $L$ on $C$ with an isomorphism $\tau$ as in \eqref{Lr}, together with an involution $\widetilde{\phi}: L \rightarrow L$ lifting $\phi$ that is compatible with $\tau$.  There is an associated {\it twisted $r$-spin structure} $S$, defined by
\begin{equation}
\label{eq:S_nodal}
S := L \otimes \O\left(-\sum_{i \in I_0} r[z_i] - \sum_{i \in I_0} r[\overline{z_i}] - \sum_{j \; | \; \text{mult}_{x_j}(L) = 0} r[x_j]\right),
\end{equation}where $I_0\subseteq I$ is a subset of the marked points such that $\text{mult}_{z_i}(L) = 0$ for all $i \in I_0$, and such that connected components of $C$ not meeting the $\phi$-fixed locus contain at most one multiplicity-zero marked point $z_i$ with $i \notin I_0$, whereas connected components meeting the $\phi$-fixed locus do not contain any multiplicity-zero $z_i$ with $i \notin I_0$.  The bundle $S$ satisfies
\begin{equation}\label{eq:Snew_nodal}
S^{\otimes r} \cong \rho^*\left(\omega_{|C|,\log}\otimes \O\left(-\sum_{i \in I_0} r[z_i] - \sum_{i \in I_0} r[\overline{z_i}] - \sum_{j\; | \; \text{mult}_{x_j}(L) = 0} r[x_j]\right)\right),
\end{equation}
and we again define $J = S^{\vee} \otimes \omega_C$, as in \eqref{eq:J}.

\begin{rmk}\label{rmk:auts_of_spin}
Closed orbifold Riemann surfaces have additional ``ghost" automorphisms in the presence of nodes, which play a role in our calculation of the automorphism groups of stable graded $r$-spin disks below.  Specifically, in the local picture of a node as $\{xy=0\} \subseteq \C^2$, there is one ghost automorphism of the form
\begin{equation}
\label{eq:ghost}
(x,y) \mapsto (\xi x, y)
\end{equation}
for each $r$th root of unity $\xi$.  These act trivially on the coarse underlying curve $|C|$, but they act nontrivially on the orbifold $C$ and induce a nontrivial action on orbifold line bundles.  Indeed, let $q$ be a node of $C$ with branches $p$ and $p'$ (given locally by $y=0$ and $x=0$, respectively), and let $m$ be the multiplicity of $S$ at $p$ (where multiplicity at a branch of a node is defined analogously to multiplicity at a marked point).  Then, if $g$ is a ghost automorphism given by \eqref{eq:ghost}, the lift $\widetilde{g}: g^*S \cong S$ multiplies the fiber over $p$ by $\xi^m$, or in other words, changes the gluing of the fibers of $S$ over $p$ and $p'$ by a factor of $\xi^m$; see \cite[Proposition 2.5.3]{ChiodoStable}.
\end{rmk}

From here on, we restrict again to the case where each connected component has genus zero.  Let $\NNN:\widehat{C}\to C$ be the normalization morphism. Then $\NNN^*L\to\widehat{C}$ is an $r$-spin structure, but $\NNN^*S$ is not a twisted spin structure, in general, since it may not satisfy the requisite conditions on the subset $I_0$. Still, there is a canonical way to associate to $S$ a twisted spin structure on $\widehat{C}$, by setting
\begin{equation}\label{eq:widehatS}
\widehat{S} := \NNN^*S\otimes\O\left(-\sum_{q\in\RRR}r[q]\right).
\end{equation}
Here, $[q]$ is the degree-$1/r$ orbifold divisor of a point $q$, and $\RRR$ is the subset of the half-nodes $q \in \widehat{C}$ (thought of as marked points of $\widehat{C}$) with $\text{mult}_q(\NNN^*S)=0$ that satisfy one of the following:
\begin{enumerate}[(1)]
\item $\NNN(q)$ is a boundary node;
\item $\NNN(q)$ is a contracted boundary node;
\item $\NNN(q)$ is an internal node of $C$ that belongs to a connected component not containing any marked point of twist $-1$ and not meeting the $\phi$-fixed locus;
\item $\NNN(q)$ is an internal node of $C$ that belongs to a connected component containing a marked point of twist $-1$, and if one normalizes $C$ only at $n(q)$, then the half-node corresponding to $q$ is in the same connected component as the marked point of twist $-1$;
\item $\NNN(q)$ is an internal node of $C$ that belongs to a connected component meeting the $\phi$-fixed locus, and if one normalizes $C$ only at $n(q)$, then the half-node corresponding to $q$ is in the connected component meeting the $\phi$-fixed locus.
\end{enumerate}
Thus, for each irreducible component $C_l$ of $\widehat{C}$, if $\{z_i\}_{i \in I_l}, \{\overline{z_i}\}_{i \in I_l}$, and $\{x_j\}_{j \in B_l}$ are the marked points lying in $C_l$ and $\{p_k\}_{k \in N_l}$ are the branches of nodes in $C_l$, we have an equation
\begin{equation}\label{eq:nodal_curve}\left(|\widehat{S}|\big|_{|C_l|}\right)^{\otimes r} \cong \omega_{|C_{l}|} \otimes  \O\left(-\sum_{i \in I_l} a_i [z_i] -\sum_{i \in I_l} a_i [\overline{z_i}] - \sum_{j \in B_l} b_j [x_j] - \sum_{h \in N_l} c_h [p_h]\right)\end{equation}
with
\[a_i,c_h\in \{-1,0,\ldots, r-1\}, \;\;\;\; b_j \in \{0,\ldots, r-1\}.\]
The numbers $a_i,b_j$, and $c_h$ are the \emph{twists} of the corresponding internal marked points, boundary marked points, and half-nodes.

Note that if $p$ and $p'$ are the two branches of a node, then we have
\[c_p + c_{p'} \equiv r-2 \mod r.\]
If $c_p=-1\mod r$ (and hence $c_{p'} = -1\mod r$, as well), we say that the node is {\it Ramond}.  Otherwise, we have $c_p + c_{p'} = r-2$, and we say that the node is {\it Neveu--Schwarz}.

As long as $C$ does not have a contracted boundary node, the notion of {\it lifting} can be defined as before, as a continuous, $\widetilde{\phi}$-invariant section of $|S|$ on (an open subset of) the complement of the special points in $\d \Sigma$ whose $r$th power is positive, and we say that a twisted $r$-spin structure is {\it compatible} if it admits a lifting over the entire complement $A$ of the special points in $\d \Sigma$.  Two liftings $v$ and $v'$ are {\it equivalent} if there is a continuous function $c: A \rightarrow \R^+$ such that $v = cv'$. We similarly define liftings and equivalence for $J$, and the same correspondence between equivalence classes of liftings for $S$ and $J$ holds. If $w$ is a lifting of $J$ over $A,$ then it induces a
lifting of
\begin{equation*}
\widehat{J}:=\omega_{\widehat{C}}\otimes\widehat{S}^\vee
\end{equation*}
over the complement of the special points in the boundary $\partial\widehat\Sigma$ of the normalization. We say that $w$ \emph{alternates} at a marked point or half-node $q$, and that $q$ is \emph{legal}, if the induced lifting on $\widehat{J}$ alternates at $q$. Otherwise, the lifting does not alternate, and the point is said to be \emph{illegal}.  A lifting $w$ of $J$ over $A$ that alternates at all boundary marked points is a {\it grading} if, in addition, one of the two half-nodes of every Neveu--Schwarz boundary node is legal and the other is illegal.

\begin{obs}\label{obs:ramond_bdry_nodes}
In a compatible $r$-spin structure for even $r$, all boundary half-nodes have twists of even parity, by item \ref{it:compatibility_even} of Proposition~\ref{prop:compatibility_lifting_parity}. In particular, in this case there are no Ramond boundary nodes, since the twist of such a node is odd.  When $r$ is odd, Ramond boundary nodes may exist, but both of their half-nodes are necessary illegal. Indeed, since the twist of a Ramond boundary half-node is $r-1$ (by the definition of $\widehat{S}$ and the twists in \eqref{eq:widehatS} and \eqref{eq:nodal_curve}), which is even, the illegality of Ramond boundary half-nodes follows from item \ref{it:lifting and parity_odd} of Proposition~\ref{prop:compatibility_lifting_parity}.
\end{obs}

If $C$ has a contracted boundary node, on the other hand, then in any connected component with such a node, the complement of the special points in $\d\Sigma$ is empty, so we must adapt the definition of a lifting.  Given that each connected component of $C$ has genus zero, there can only be one contracted boundary node in each connected component. Restrict to one such component, and let $q$ be the contracted boundary node. Assume, additionally, that $q$ is Ramond.
Recall that the fiber $\omega_{|C|}\big|_q$ is canonically identified with $\C$ via the residue, and the involution $\phi$ is sent, under this identification, to the involution $z\mapsto -\bar{z},$ whose fixed points are the purely imaginary numbers.\footnote{The residue of a conjugation-invariant form $\zeta$ can be calculated as $\frac{1}{2\pi i} \oint_L\zeta$, where $L\subset\Sigma$ is a small loop surrounding $q$ whose orientation is such that $q$ is to the left of $L.$ Applying conjugation and using the invariance of $\zeta$ shows that the residue is imaginary.}
We define a {\it lifting} as a $\widetilde\phi$-invariant element
\[v \in \Gamma(\{q\}, |S|) = |S|\big|_q\]
such that the image of $v^{\otimes r}$ under the map
\[|S|^{\otimes r}\big|_q \rightarrow \omega_{|C|}\big|_q\]
is positive imaginary, meaning that it lies in $i\R_+.$  In this case, we call the twisted $r$-spin structure {\it compatible} if the contracted boundary node is Ramond and a lifting exists. Two liftings are equivalent if, at the contracted boundary node, they differ by multiplication by a positive number.  There always exists a $\widetilde{\phi}$-invariant $w \in |J|\big|_q$ such that $\langle v,w \rangle$ is positive imaginary, and we refer to this $w$ as a {\it grading}; this is the limiting case of the notion of grading for smooth curves.

With or without a contracted boundary node, we now have the following definition:

\begin{definition}
\label{def:stablegraded}
A {\it stable genus-zero graded $r$-spin surface} is a nodal marked orbifold Riemann surface with boundary whose coarse underlying surface $(|C|, \{z_i\}, \{\bar z_i\}, \{x_j\})$ is a stable Riemann surface in which each connected component has genus zero, together with:
\begin{enumerate}[(1)]
\item a compatible twisted $r$-spin structure $S$ in which all boundary marked points have twist $r-2$ and all contracted boundary nodes are Ramond;
\item an equivalence class of gradings;
\item a choice of one distinguished special point (called the {\it anchor} and marked zero) in each connected component $C'$ of $C$ that is either disjoint from the set $\text{Fix}(\phi)$ or meets the set $\text{Fix}(\phi)$ in a single contracted boundary node.  If either a contracted boundary node or marked point of twist $-1$ exists, we require the anchor to be this point; if not, the anchor is simply required to be a marked point.  We also require that the collection of anchors is $\phi$-invariant, so that it descends to $\Sigma$. Finally, if the twist of an anchor $z_i$ is $r-1,$ we fix an involution
    $\widetilde{\phi}$ on the fiber $\left(|S|\otimes \O\left([z_i]\right)\right)_{z_i}$ and an orientation of the $\widetilde\phi$-fixed subspace, as in Definition \ref{def:graded_sphere}.
\end{enumerate}
The internal and boundary marked points are required to satisfy the same properties as in the smooth case, and in particular, an anchor, if it is a marked point, is the only marked point in its connected component that is marked zero.
\end{definition}

We conclude this subsection with an existence and uniqueness result analogous to Proposition \ref{prop:existence_smooth}. Here, we refer to irreducible components of $C$ that do not meet the preimage of $\text{Fix}(|\phi|)$, or that meet the preimage of $\text{Fix}(|\phi|)$ in a single contracted boundary node, as {\it sphere components}, and we refer to the other irreducible components as {\it disk components}.

\begin{prop}\label{prop:existence_stable}
Suppose that $(C,\phi,\Sigma, \{z_1, \ldots, z_l\}, \{x_1, \ldots, x_k\})$ is a connected stable marked disk and $a_1,\ldots,a_l \in \{0, 1, \ldots, r-1\}$ are such that \eqref{eq:open_rank1} and \eqref{eq:open_rank2} hold with $|B| = k$ and $b_j = r-2$ for all $j.$  Then there exists a unique graded $r$-spin structure on this stable disk with internal twists given by the integers $a_i$.

The order of the automorphism group of the graded $r$-spin structure is $r^n$ where $2n$ is the number of internal nodes of $C$.
\end{prop}

\begin{proof}
The proof of the existence and uniqueness of the graded structure closely mirrors the proof of Proposition \ref{prop:existence_smooth}.  In particular, there exists a unique twisted $r$-spin structure on the closed genus-zero surface $C$ with the given twists at all internal, conjugate, and boundary marked points, so what remains is to produce the involution $\widetilde{\phi}: S \rightarrow S$ and the sections $v$ and $w$.

First, suppose that there is no contracted boundary node.  Then the involution $\widetilde\phi$ can be defined component-by-component.  Namely, starting from some boundary point $p$, we first construct the involution on a single component using the argument of Proposition \ref{prop:existence_smooth}.  At each node, there is an identification of the fibers of $S$ on the two half-nodes, so when we encounter a boundary half-node, the involution on one side induces the involution on the other side, and we use the other half-node as the basepoint for the construction of the involution in its component. An analogous treatment works for conjugate internal nodes.  The resulting involution is unique up to an isomorphism.

The lifting $v$ and the grading $w$ can be defined componentwise, uniquely up to isomorphism (as we prove below), using either the parity of the twist (in the case where $r$ is odd) or the parity constraint \eqref{parity} (in the case where $r$ is even) to determine whether it alternates at each half-node.  Ramond boundary nodes appear only when $r$ is odd, and then there is no choice in the lifting. For Neveu--Schwarz boundary nodes, there are two choices of lifting when $r$ is even, but they are equivalent via the isomorphism scaling the fibers of $S$ by $-1$.

To prove that there is a unique choice of grading satisfying the requisite condition at a Neveu--Schwarz boundary node $q$, let $p$ and $p'$ be the two branches with respective twists $c_p$ and $c_{p'}$.  When $r$ is odd, the fact that $w$ alternates at exactly one branch is immediate from item \ref{it:lifting and parity_odd} of Proposition \ref{prop:compatibility_lifting_parity}, since exactly one of $c_p$ and $c_{p'}$ is odd.  When $r$ is even, assume for simplicity that the nodal Riemann surface $C$ consists of two sphere components joined at the node $q$.  Then, adding the parity constraints \eqref{parity} on the two components, one obtains
\[\frac{4  + 2 \sum a_i + \sum b_i + c_p + c_{p'}}{r} \equiv k + \delta_p^{\text{alt}} + \delta_{p'}^{\text{alt}} \mod 2,\]
where $k$ is the number of boundary marked points, and $\delta_p^{\text{alt}}$ is defined to be $1$ if $w$ alternates at $p$ and $0$ otherwise.  Combining this with the parity constraint \eqref{parity} on the entire curve, one finds
\[\frac{2 + c_p + c_{p'}}{r} \equiv \delta_p^{\text{alt}} + \delta_{p'}^{\text{alt}} \mod 2,\]
so we indeed see that if $c_p + c_{p'} = r-2$, then $w$ alternates at exactly one branch.

This completes the proof of the existence of the graded $r$-spin structure in the case where there is no contracted boundary node.  When there is a contracted boundary node, it must be Ramond. To see this, apply constraint \eqref{eq:open_rank1} to see that the internal twists satisfy
\begin{equation}
\label{constraint1}
-2-2\sum a_i \equiv 0 \mod r.
\end{equation}
On a single component of $C$, on the other hand, the degree of the restriction of $S$ is
\[\frac{-2-\sum a_i - c_p}{r},\]
so we must have
\begin{equation}
\label{constraint2}
-2-\sum a_i - c_p \equiv 0 \mod r
\end{equation}
Combining \eqref{constraint1} and \eqref{constraint2}, we find that $2(c_p+1) \equiv 0 \mod r$.  When $r$ is odd, this is sufficient to conclude that $c_p\equiv r-1 \mod r$.  When $r$ is even, we apply \eqref{eq:open_rank2} to find
\[\frac{2+2\sum a_i}{r} \equiv 0 \mod 2.\]
Therefore, we have
\[\frac{1+\sum a_i}{r} \in \Z,\]
so $1+ \sum a_i \equiv 0 \mod r$.  Combining this with \eqref{constraint2} again shows that $c_p\equiv r-1 \mod r$.

When there is a contracted boundary node $q$, the involution $\widetilde\phi$ can be defined first in the fiber of $q$, by choosing a lifting of the conjugation on $\omega_{\log}\big|_q$.  There are $r$ such choices; when $r$ is odd, they are all equivalent, while when $r$ is even, there are $r/2$ equivalent choices that make the structure compatible, but in either case, we choose one.  We then extend the conjugation to the other components of $C$, if there are any, component-by-component, using the argument in the beginning of this proof. Since the contracted boundary node is Ramond, we can now choose the grading as above. We observe that, again, there is no choice when $r$ is odd and there are two equivalent choices when $r$ is even.

Finally, we compute the order of the automorphism group.  The closed case is known (see, for example, \cite[Proposition 1.18]{JKV} or \cite[Section 2.3]{CZ}): an $r$-spin structure on a closed, genus-zero stable curve with $N$ nodes has $r^{N+1}$ automorphisms.  Namely, each node contributes $r$ ghost automorphisms of the curve, which can each be lifted to the spin bundle as in Remark~\ref{rmk:auts_of_spin}, and each of the resulting automorphisms of the spin structure can be composed with a global fiberwise scaling by an $r$th root of unity.

Let us now consider which of these automorphisms respects the graded structure.  By compatibility with the involution, a ghost automorphism at an internal node determines the ghost automorphism at its conjugate node, but is otherwise unconstrained.  Nontrivial ghost automorphisms at boundary and contracted boundary nodes, on the other hand, cannot respect both the conjugation and the graded structure.  Indeed, at a boundary node, \eqref{eq:ghost} may constitute an automorphism only if $\xi$ is real. For odd $r$, this implies that $\xi=1$.  When $r$ is even, $\xi=-1$ is also a possibility; however, in this case, by Remark \ref{rmk:auts_of_spin}, this ghost automorphism acts on the spin bundle by changing the gluing of the fibers at the node by a factor of $\xi^m$, in which $m$ is the multiplicity at one half-node. By Observation \ref{obs:ramond_bdry_nodes}, the twists at the half-nodes are even, so $m$ is odd.  It follows that for $\xi^m= (-1)^m =-1$, and hence the ghost automorphism does not respect the grading.  At a contracted boundary node, we consider the automorphisms in local coordinates, in which they have the form
\[(x,y) \mapsto (\xi x, \zeta y),\]
where $\xi$ and $\zeta$ are $r$th roots of unity.  Such an automorphism is compatible with the conjugation---which, locally, has the form $(x,y) \mapsto (\bar{y}, \bar{x})$---only if $\zeta=\bar{\xi}=\xi^{-1}$.  But in this case, given the orbifold structure at the node, it is the identity automorphism.

It follows, then, that there are $r^n$ contributing ghost automorphisms, where $2n$ is the number of internal nodes.  Since we have already argued in Proposition \ref{prop:existence_smooth} that the fiberwise scalings on disk components do not respect the graded structure, there are no further automorphisms.
\end{proof}

\section{The moduli space of graded $r$-spin disks}\label{sec:moduli}

Henceforth, we usually denote an $r$-spin disk with a lifting simply by $\Sigma,$ the preferred half, suppressing most of the notation.

\subsection{Stable graded $r$-spin graphs}

It is useful to encode some of the combinatorial data of graded $r$-spin disks in a decorated dual graph.

\begin{definition}\label{def:non_spin_graph}
A {\it genus-zero pre-stable dual graph} is a tuple
\[\Gamma = (V, H, \sigma_0, \sim, H^{CB}, m),\]
in which
\begin{enumerate}[(i)]
\item $V$ is a finite set (the {\it vertices}) equipped with a decomposition $V = V^O \sqcup V^C$ into {\it open} and {\it closed vertices};
\item $H$ is a finite set (the {\it half-edges}) equipped with a decomposition $H = H^B \cup H^I$ into {\it boundary} and {\it internal half-edges};
\item $\sigma_0: H \rightarrow V$ is a function, viewed as associating to each half-edge the vertex from which it emanates;
\item $\sim$ is an equivalence relation on $H$, which decomposes as a pair of equivalence relations $\sim_B$ on $H^B$ and $\sim_I$ on $H^I$.  The equivalence classes are required to be of size $1$ or $2$, and those of size $1$ are referred to as {\it tails}.  We denote by $T^B \subseteq H^B$ and $T^I \subseteq H^I$ the sets of equivalence classes of size $1$ in $H^B$ and $H^I$, respectively;
\item $H^{CB}$ is a subset of $T^I$, the {\it contracted boundary tails};
\item $m$ is a function given by
\[m = m^B \sqcup m^I: T^B \sqcup (T^I \setminus H^{CB}) \rightarrow \{0\}\cup2^\N_{*},\]
where $m^B$ and $m^I$ (the \emph{boundary} and \emph{internal markings}) satisfy the definition of a marking when restricted to any connected component of $\Gamma$, and where $m^I$ is strict on connected components with an open vertex or a contracted boundary tail.
\end{enumerate}
Note that $(V,H, \sigma_0)$ defines a graph, and that we do not require this graph to be connected; denote its set of connected components by $\Conn(\Gamma) = \{\Lambda_i\}$. We require the above data to satisfy the following conditions:
\begin{enumerate}[(1)]
\item For each boundary half-edge $h \in H^B$, we have $\sigma_0(h) \in V^O$;
\item For each $\Lambda_i$, we have $h^1(\Lambda_i ) = 0$;
\item Each $\Lambda_i$ contains at most one half-edge in $H^{CB}$, and if $\Lambda_i$ contains such a half-edge, then all vertices of $\Lambda_i$ are closed;
\item For each $\Lambda_i$, the sub-graph formed by its open vertices (if any exist) and their incident boundary edges is connected.
\end{enumerate}
Conditions (2),~(3), and (4) guarantee that for any nodal graded $r$-spin surface with dual graph $\Gamma$, each connected component of the closed surface $C$ has genus zero.
\end{definition}

We refer to elements of $T^B$ as {\it boundary tails} and elements of $T^I \setminus H^{CB}$ as {\it internal tails}, and we denote
$T := T^I \sqcup T^B.$
Note that $\sim$ induces a fixed-point-free involution on $H \setminus T$, which we denote by $\sigma_1$.  Write
\[E^B := (H^B \setminus T^B)/\sim_B,\;\;\;\;\;\;\; E^I := (H^I \setminus T^I)/\sim_I\]
and refer to these as {\it boundary edges} and {\it internal edges}.
The set of {\it edges} is
$E := E^B \sqcup E^I.$
Denote by $\sigma_0^B$ the restriction of $\sigma_0$ to $H^B$, and similarly for $\sigma_0^I$.

For each vertex $v$, set
$k(v) := |(\sigma_0^B)^{-1}(v)|,$ and $l(v):= |(\sigma_0^I)^{-1}(v)|.$
We say that an open vertex $v \in V^O$ is {\it stable} if $k(v) + 2l(v) > 2$, and we say that a closed vertex $v \in V^C$ is {\it stable} if $l(v) > 2$. 
A graph is {\it stable} 
if all of its vertices are stable, 
and it is {\it closed} if $V^O = \emptyset.$ A graph is \emph{smooth} if there are no edges or contracted boundary tails.

\begin{definition}
An {\it isomorphism} between two genus-zero pre-stable dual graphs
\[\Gamma = (V,H, \sigma_0, \sim, H^{CB}, m) \; \text{ and } \; \Gamma' = (V', H', \sigma_0', \sim', H'^{CB}, m')\]
is a pair $f = (f^V, f^H)$, where
$f^V: V \rightarrow V'$ and $f^H: H \rightarrow H'$
are bijections satisfying
\begin{enumerate}[(1)]
\item $h_1 \sim h_2$ if and only if $f(h_1) \sim' f(h_2)$,
\item $f^{V} \circ \sigma_0 =  \sigma'_0 \circ f^H$,
\item $m = m' \circ f^H$,
\item $f(H^{CB}) = H'^{CB}$.
\end{enumerate}
We denote by $\text{Aut}(\Gamma)$ the group of automorphisms of $\Gamma$.
\end{definition}

Pre-stable dual graphs encode the discrete data of a marked orbifold Riemann surface with boundary.  In order to encode the additional data of a twisted spin structure and a lifting, we must add further decorations.

\begin{definition}
\label{def:graph}
A {\it genus-zero twisted $r$-spin dual graph with a lifting} is a genus-zero pre-stable dual graph $\Gamma$ as above, together with maps
\[\twE: H \rightarrow \{-1, 0, 1, \ldots, r-1\}\]
(the {\it twist}) and
\[\alt: H^B \rightarrow \Z/2\Z\]
and a subset $T^*\subseteq T^I$ (the {\it anchors}), satisfying the following conditions:
\begin{enumerate}[(i)]
\item
Any connected component of $\Gamma$ that is not stable consists either of (a) a single open vertex with a single internal tail, or (b) a single closed vertex with exactly two tails, one of which is in $H^{CB}$ and the other of which is in $H^I$.

\item
Every closed connected component contains exactly one tail in $T^*$.  All contracted boundary tails and all tails $t$ with $\twE(t) =-1$ belong to $T^*.$ Open connected components have no tails in $T^*$. Any element of $T^*\setminus H^{CB}$ is marked $0$ and is the only tail marked $0$ in its connected component.

\item
For any vertex $v$, the total number of incident half-edges $h$ with $h \in T^*$ or $\twE(h) = -1$ is at most one.

\item
\label{cond1} For any contracted boundary tail $t \in H^{CB}$, we have $\twE(t) =r-1$.

\item
 \label{item} For any open vertex $v \in V^O$,
\[2\sum_{h \in (\sigma_0^I)^{-1}(v)} \twE(h) + \sum_{h \in (\sigma_0^B)^{-1}(v)} \twE(h) \equiv r-2 \mod r\]
and
\[\frac{2\sum_{h \in (\sigma_0^I)^{-1}(v)}\ \twE(h) + \sum_{h \in (\sigma_0^B)^{-1}(v)} \twE(h) + 2}{r} \equiv \sum_{h \in (\sigma_0^B)^{-1}(v)} \alt(h) \mod 2.\]

\item
For any closed vertex $v \in V^C$,
\[\sum_{h \in \sigma_0^{-1}(v)} \twE(h) \equiv r-2 \mod r.\]

\item
\label{cond2} For any half-edge $h \in H \setminus T$, we have
\[\twE(h) + \twE(\sigma_1(h)) \equiv r -2 \mod r,\]
and at most one of $\twE(h)$ and $\twE(\sigma_1(h))$ equals $-1$.  No boundary half-edge $h$ satisfies $\twE(h)=-1.$ In case $h \in H^I \setminus T^I$ satisfies $\twE(h) \equiv -1\mod~r,$ then $\twE(h)=r-1$ precisely if, after detaching the edge, $h$ belongs to the connected component containing an anchor $t^* \in T^*$ (if $h$ is in a closed connected component of $\Gamma$) or an open vertex $v \in V^O$ (if $h$ is in an open connected component of $\Gamma$).

\item
\label{it:-1}For any boundary half-edge $h \in H^B \setminus T^B$, if $\twE(h) \neq r-1 $ we have
\[\alt(h) + \alt(\sigma_1(h)) = 1\]
and if $\twE(h) =r -1$ then $\alt(h) = \alt(\sigma_1(h)) = 0.$

\item
If $r$ is odd, then for any $h \in H^B$,
\[\alt(h) \equiv \twE(h) \mod 2,\]
and if $r$ is even, then for any $h \in H^B$,
\[\twE(h) \equiv 0 \mod 2.\]

\end{enumerate}
Boundary half-edges $h$ with $\alt(h) = 1$ are called {\it legal}, and those with $\alt(h) = 0$ are called {\it illegal}.  Half-edges $h$ with $\twE(h) \in \{-1, r-1\}$ are called {\it Ramond}, and those with $\twE(h) \in \{0, \ldots, r-2\}$ are called {\it Neveu--Schwarz}.  An edge is called {\it Ramond} if one (hence both) of its half-edges is Ramond, and Neveu--Schwarz otherwise.
\end{definition}

We say that a genus-zero twisted $r$-spin dual graph with a lifting is {\it stable} if the underlying dual graph is stable, in the sense specified above.  An {\it isomorphism} between genus-zero twisted $r$-spin dual graphs with liftings consists of an isomorphism in the sense of Definition~\ref{def:graph} that respects $\twE,~\alt,$ and $T^*.$  Analogously to Definition~\ref{def:stablegraded}, we define a {\it genus-zero graded $r$-spin graph} to be a genus-zero twisted $r$-spin dual graph with a lifting, such that every boundary tail $h \in H^B$ has
\[\twE(h) = r-2, \;\;\;\; \alt(h) = 1.\]
Any stable graded $r$-spin disk $\Sigma$ induces a stable genus-zero graded $r$-spin graph $\Gamma(\Sigma)$.

\subsection{Moduli of stable graded $r$-spin disks}\label{subsec:mod_const}

In the situation without boundary, there is a well-studied moduli space $\M_{g,n}^{1/r}$ of stable Riemann surfaces with $r$-spin structure, which is known to be a smooth Deligne--Mumford stack with projective coarse moduli for which the forgetful map to $\M_{g,n}$ is finite (see \cite{ChiodoStable} or, in the setting of a slightly different compactification, \cite{Jarvis}).  This moduli space admits a decomposition into open and closed substacks,
\begin{equation}
\label{eq:closedmoduli}
\M_{g,n}^{1/r} = \bigsqcup_{\vec{a} = (a_1, \ldots, a_n)} \M_{g,\vec{a}}^{1/r},
\end{equation}
where $a_i \in \{0,1,\ldots, r-1\}$ for each $i$ and $\M_{g,\vec{a}}^{1/r}$ denotes the substack of $r$-spin structures with twist $a_i$ at the $i$th marked point.  In genus zero, the situation is even simpler: according to Proposition~\ref{prop:existence_smooth}, for any choice of $\vec{a}$ such that \eqref{eq:close_rank1} holds, the moduli space $\M^{1/r}_{0,\vec{a}}$ has coarse moduli isomorphic to $\M_{0,n}$ and generic additional isotropy $\Z/r\Z$. The isomorphism of coarse moduli is given by the smooth map $\text{For}_{\text{spin}}$ that forgets the spin structure.

When there is no boundary, it is also straightforward to add the information of a grading.  Indeed, the moduli space of graded $r$-spin spheres for which the anchor does not have twist $r-1$ is canonically isomorphic to $\M_{0,n}^{1/r}$, whereas the moduli space of graded $r$-spin spheres with twist $\vec{a}$ and anchor twisted $r-1$ is, canonically, an $r$-to-$1$ cover of the moduli space $\M_{0,\vec{a'}}^{1/r}$ on which $\vec{a'}$ agrees with $\vec{a}$ except that the anchor has twist $-1$.

To generalize the construction of the moduli space to the open setting, we first note that in \cite{PST14}, the moduli space $\M_{0,k,l}$ of connected stable marked disks with boundary marked points marked by $\{1, \ldots, k\}$ and internal marked points marked by $\{1, \ldots, l\}$ was considered.  It is a smooth orientable manifold with corners in the sense of \cite{Joyce}, and its dimension is
\[\dim_{\R}(\M_{0,k,l}) = k+2l-3.\]
Let $\M_{0,k,l}^{1/r}$ denote the set of isomorphism classes of connected stable graded $r$-spin disks, with boundary and internal marked points as above.  There is a set-theoretic decomposition analogous to \eqref{eq:closedmoduli},
\[\M_{0,k,l}^{1/r} = \bigsqcup_{\vec{a}} \M_{0,k, \vec{a}}^{1/r},\]
in which $\M^{1/r}_{0,k,\vec{a}} \subset \M^{1/r}_{0,k,l}$ consists of those disks for which the $i$th internal marked point has twist $a_i$.  By Proposition~\ref{prop:existence_stable}, whenever $\M_{0,k,\vec{a}}^{1/r} \neq \emptyset,$ there is a bijection given by the forgetful map
\[\text{For}_{\text{spin}}: \M_{0,k,\vec{a}}^{1/r} \rightarrow \M_{0,k,l},\]
and we use these to give $\M_{0,k,l}^{1/r}$ the structure of a manifold with corners.

This describes the coarse underlying space of $\M_{0,k,l}^{1/r}$.  The main theorem of this section is that it can also be given an orbifold-with-corners structure in the sense of \cite[Section 3]{Zernik}:
\begin{thm}\label{thm:moduli}
The moduli space $\M_{0,k,l}^{1/r}$ of connected stable graded $r$-spin disks with boundary marked points marked by $\{1, \ldots, k\}$ and internal marked points marked by $\{1, \ldots, l\}$ is a compact smooth orientable orbifold with corners of real dimension $k+2l-3$. Its universal bundle admits a universal grading.
\end{thm}

We split the proof of Theorem \ref{thm:moduli} into three parts. Lemma \ref{lem:orbi_str} shows that the moduli space is a compact smooth orbifold with corners, Lemma \ref{lem:2-1 cover} proves the existence of the universal grading, and Proposition \ref{prop:or_moduli} together with Observation \ref{obs:or_with_without_spin} proves the orientability.

\begin{lemma}\label{lem:orbi_str}
$\M_{0,k,l}^{1/r}$  has the structure of a compact smooth orbifold with corners.
\end{lemma}
\begin{proof}  We describe a procedure that defines an orbifold-with-corners structure on $\M_{0,k,l}^{1/r}$, analogous to the procedure performed in \cite[Section 2]{Zernik}.  To define the procedure, we make reference to the following sequence:
\begin{equation}
\label{eq:OWCsequence}
\M_{0,k,l}^{1/r}\stackrel{(5)}{\to}\widehat{\mathcal{M}}_{0,k,l}^{1/r}\stackrel{(4)}{\hookrightarrow}\widetilde{\mathcal{M}}_{0,k,l}^{1/r}\stackrel{(3)}{\to}\widetilde{\mathcal{M}}_{0,k,l}^{1/r,\Z_2}\stackrel{(2)}{\to}
\overline{\mathcal{M}}_{0,k+2l}^{1/r,\Z_2}\stackrel{(1)}{\to}\overline{\mathcal{M}}_{0,k+2l}^{'1/r}.
\end{equation}
The moduli spaces and maps appearing in \eqref{eq:OWCsequence} are defined as they appear in what follows.

\textbf{Step 1: }
First, $\overline{\mathcal{M}}_{0,k+2l}^{'1/r}$ is the suborbifold of $\M_{0,k+2l}^{1/r}$ given by the condition that the first $k$ markings have twist $r-2$ and that the integer defined in \eqref{eq:close_rank1} is of the same parity as $k+1$.  Inside this space, $\M_{0,k+2l}^{1/r, \Z_2}$ is the fixed locus of the involution defined by
\[(C;w_1, \ldots, w_{k+2l}, S) \mapsto (\overline{C}; w_1, \ldots, w_k, w_{k+l+1},\ldots, w_{k+2l},w_{k+1},\ldots,w_{k+l}, \overline{S}),\]
where $\overline{C}$ and $\overline{S}$ are the same as $C$ and $S$ but with the conjugate complex structure (more details on the fixed point functor on stacks can be found in \cite{fixedPointStack}).  As the fixed locus of an anti-holomorphic involution, $\M_{0,k+2l}^{1/r, \Z_2}$ has the structure of a real orbifold.  It parameterizes isomorphism types of marked spin spheres with a real structure (an involution $\widetilde\phi$ covering the conjugation $\phi$ on $C$) and the prescribed twists, and it maps to $\M_{0,k+2l}^{1/r}$ (in general it is not a sub-orbifold, since some isotropy is lost), so it inherits a universal curve via pullback.

\begin{rmk}
\label{rem:isotropy}
Let us digress to discuss the isotropy of $\M_{0,k+2l}^{1/r,\Z_2}$, especially near nodal strata.\footnote{Without spin structure, the nodal strata of the real moduli space $\M_{0,k,l}$ are discussed in \cite[Section 3]{Liu}, and in the closed case with spin structure, the nodal strata of $\M_{0,l}^{1/r}$ are discussed in \cite[Section 4]{ChiodoGMS}.}  The generic point of $\M_{0,k+2l}^{1/r,\Z_2}$ has isotropy coming from scaling the fibers of the spin bundle by real $r$th roots of unity.   When $r$ is odd, there are no such roots and hence no generic isotropy, while when $r$ is even, there is generic $\Z/2\Z$ isotropy.  Nodal strata have additional $\Z/r\Z$ isotropy for each internal Neveu--Schwarz node, coming from the ghost automorphisms.

For boundary nodes, there is a difference in behavior for $r$ odd or even.  When $r$ is odd, boundary nodes also contribute no further isotropy.  Furthermore, if $U\times[-1,1]$ is a neighborhood in the moduli space of a curve with a single boundary node, such that $U\times\{0\}$ is the intersection with the nodal stratum and $(u,t)$ for $t\neq 0$ corresponds to a smooth real sphere, then the passage from $t<0$ to $t>0$ geometrically corresponds to flipping one of the two disk components and defining the involution on the spin bundle in the unique possible way.  Thus, in this case, the behavior near the node is exactly like in the real, non-spin case, and the spin moduli continues to be a trivial degree-one cover of the non-spin moduli generically.  When $r$ is even, on the other hand, boundary nodes contribute additional $\Z/2\Z$ isotropy, coming from ghost automorphisms of the form $(x,y)\to(-x,y)$; see the discussion in the proof of Proposition \ref{prop:existence_stable}.  On the moduli level, the picture is that a neighborhood of a nodal curve with a single boundary node looks locally like $U\times[-1,1]/(\Z/2\Z)$, where $U\times\{0\}$ is the nodal locus and the generator of $\Z/2\Z$ takes $(u,t)$ to $(u,-t).$

Contracted boundary nodes add no additional isotropy, as we saw in the end of the proof of Proposition \ref{prop:existence_stable}.  On the moduli level, again let $U\times[-1,1]$ be a neighborhood in the moduli space of a curve with a contracted boundary, such that $U\times\{0\}$ is the intersection with the nodal stratum and $(u,t)$ for $t> 0$ corresponds to a smooth real sphere on which the conjugation has nonempty fixed locus.  Then $(u,t)$ for $t<0$ corresponds to a real sphere on which the conjugation has no fixed points, so after taking the quotient by the conjugation, the result is a marked real projective plane.
\end{rmk}

\textbf{Step 2: }
Returning to our discussion of \eqref{eq:OWCsequence}, the next step is to cut $\overline{\mathcal{M}}_{0,k+2l}^{1/r,\Z_2}$ along the real simple normal crossings divisor consisting of curves with at least one real node, via the ``real hyperplane blow-up" as in \cite{Zernik}, yielding an orbifold with corners $\widetilde{\mathcal{M}}_{0,k,l}^{1/r, \Z_2}$.  We direct the reader to \cite[Section 3.3]{Zernik} for more details, but the idea is the following:

Near the real divisor consisting of curves with a contracted boundary node---or, when $r$ is odd, near the real divisor consisting of curves with a boundary node (a real node that is not an isolated fixed point of the conjugation)---the real blow-up is the standard cutting procedure that can be defined without a spin structure. In the notation of Remark \ref{rem:isotropy}, the real blow-up corresponds to the natural quotient map
\[U\times[-1,0]\sqcup U\times[0,1]\to U\times[-1,0]\cup U\times[0,1]= U\times[-1,1].\]
When $r$ is even, the blow-up near the real divisor consisting of curves with a boundary node is a topologically trivial operation but nontrivial on the orbifold level. In particular, in local charts, the real blow-up when $r$ is even is equivalent to blowing up before taking the extra $\Z/2\Z$ quotient mentioned in Remark \ref{rem:isotropy}, and then taking the quotient, so it kills the additional $\Z/2\Z$ isotropy on nodal strata.  In other words, the blow-up is locally the map
\[\left(U\times[-1,0]\sqcup U\times[0,1]\right)/(\Z/2\Z)\to (U\times[-1,0]\cup U\times[0,1])/(\Z/2\Z)= U\times[-1,1]/(\Z/2\Z),\] where the generator of $\Z/2\Z$ takes $(u,t)\in U\times[-1,0]\sqcup U\times[0,1]$ to $(u,-t).$

\textbf{Step 3: }
Consider the subset of $\widetilde{\mathcal{M}}_{0,k,l}^{1/r,\Z_2}$ whose generic point is a smooth marked real spin sphere with nonempty real locus.  Then $\widetilde{\mathcal{M}}_{0,k,l}^{1/r}$ is the disconnected $2$-to-$1$ cover of this subset given by the choice of a distinguished connected disk component of $C\setminus C^\phi$.  Equivalently, in the generic (smooth) situation, we first restrict to the connected components of $\widetilde{\mathcal{M}}_{0,k,l}^{1/r,\Z_2}$ consisting of real spheres on which the conjugation has nonempty fixed locus, and then we choose an orientation for $C^{\phi}$.  It is important to note, however, that this choice can be uniquely continuously extended to points in the boundary of $\widetilde{\mathcal{M}}_{0,k,l}^{1/r}.$

\textbf{Step 4: }
Inside $\widetilde{\mathcal{M}}_{0,k,l}^{1/r}$, we denote by $\widehat{\mathcal{M}}_{0,k,l}^{1/r}$ the union of connected components such that the marked points $w_{k+1},\ldots, w_{k+l}$ lie in the distinguished stable disk and, for even $r,$ the spin structure is compatible.

\textbf{Step 5: }
Finally, $\M_{0,k,l}^{1/r}$ is the cover of $\widehat{\mathcal{M}}_{0,k,l}^{1/r}$ given by a choice of grading. When $r$ is odd, this is the identity, while when $r$ is even, it is a $2$-to-$1$ cover given by forgetting the global $\Z/2\Z$ isotropy.

As a manifold with corners, $\M_{0,k,l}^{1/r}$ is indeed the same space defined previously, but now it has the additional structure of an orbifold with corners.  The proof that $\widehat{\mathcal{M}}_{0,k,l}^{1/r}$ is an orbifold with corners is identical to the proof of the analogous Theorem 2 in \cite{Zernik}. The space $\M_{0,k,l}^{1/r}$ then inherits the orbifold-with-corners structure from $\widehat{\mathcal{M}}_{0,k,l}^{1/r}$. It is moreover compact since compactness is preserved at every step.
\end{proof}

Over $\M_{0,k,l}^{1/r}$, there is a universal curve whose fibers are compatible stable spin disks. The content of Lemma \ref{lem:2-1 cover} below is that one may construct a graded structure, in a continuous way, on the fibers.
\begin{lemma}\label{lem:2-1 cover}
One can continuously choose a grading for the fibers of the universal curve of $\mathcal{M}_{0,k,l}^{1/r}.$ This choice is unique when $r$ is odd, while for even $r$ it is unique up to a global change of grading in each connected component of the moduli space.
\end{lemma}
\begin{proof}
We first prove the lemma for the universal curve over
\[\text{Int}(\mathcal{M}_{0,k,l}^{1/r}):=\oCMr\setminus\partial\oCMr.\]
The statement is clear when $r$ is odd, and its uniqueness up to a possible global change of grading in each component is immediate when $r$ is even, given the existence. We hence prove the existence for even $r.$

The fact that this choice can be made locally in a continuous way is straightforward.  The obstruction to making such a choice global in $\text{Int}(\mathcal{M}_{0,k,l}^{1/r})$ is the possible existence of a loop \[\gamma:\partial\bar\Delta\to \text{Int}(\oCMr),\] where $\bar\Delta$ is the closed unit disk, along which a continuous choice of grading alternates.

Since the strata corresponding to disks with an internal ``bubble" are of codimension two, it may be assumed that $\gamma(\partial\bar\Delta)\subset \CM_{0,k,\vec{a}}^{1/r}$.  We first show that $\gamma$ may be extended to $\hat\gamma:\bar\Delta\to \oCM_{0,k,\vec{a}}^{1/r}$ by verifying that $\pi_1(\oCM_{0,k,l})$ is trivial. This is true when $l=0,$ as every connected component of $\oCM_{0,k,0}$ is the contractible associahedron.  By considering the forgetful map $\oCM_{0,k,1}\to\oCM_{0,k,0},$, whose fiber is contractible, it is easy to see that $\oCM_{0,k,1}$ is also contractible. For $l>1,$ it is enough to consider an arbitrary $\gamma:\partial\bar\Delta\to \CM_{0,k,l}$ and to show that it can be extended to the disk. By working in the unit disk model where $z_1$ is mapped to the origin and $x_1$ to $1,$ one may write $\gamma(\theta)$ as \[(1,x_2(\theta),\ldots,x_k(\theta),0,z_2(\theta),\ldots,z_l(\theta)).\]
An extension $\hat\gamma$ may be written as
\begin{align*}
\hat{\gamma}(r,\theta)=(1, &e^{r\log{x_2(\theta)}+(1-r)\frac{2\pi i}{k}},e^{r\log{x_3(\theta)}+(1-r)\frac{4\pi i}{k}},\ldots,\\
&e^{r\log{x_k(\theta)}+(1-r)\frac{2(k-1)\pi i}{k}}
,0,r^{l-1}z_2(\theta),r^{l-2}z_3(\theta),\ldots,rz_l(\theta)),\end{align*}
where we define the logarithm by excluding the positive real ray.

As $\bar\Delta$ is contractible, standard homotopy arguments now show that we can uniquely extend the grading from an arbitrary grading at $\gamma(1)$ to a grading for all points of $\hat{\gamma}(\bar\Delta).$ Restricting to $\partial\bar\Delta,$ the grading defines a grading for the points of $\gamma.$ Thus, the grading does not alternate along $\gamma$, and therefore it can be defined globally.

Given a grading for the fibers of the universal curve over the interior of the moduli space, we extend it to fibers over the boundary by continuity.
If $\Sigma_t$ is a family of smooth graded $r$-spin surfaces converging to $\Sigma_0$, then the gradings of $\Sigma_t$ determine, by continuity, a compatible lifting on $\Sigma_0$ away from special points. An argument as in Proposition \ref{prop:existence_stable} shows that this lifting is in fact a grading, and it is independent of the family $\Sigma_t.$

Suppose $\Sigma\in\partial\overline{\mathcal{M}}_{0,k,l}^{1/r}$ has a contracted boundary node, which in particular forces that $k=0$.  Using the same argument as in Proposition \ref{prop:existence_stable} for the contracted boundary case, we observe that the contracted boundary node must be Ramond.  The limit of the grading in the smooth case, at the boundary stratum consisting of surfaces with a contracted boundary node, is precisely a grading in the sense of contracted boundary nodes defined above.
\end{proof}

\begin{rmk}
It is interesting to note that even for nodal spin disks with Neveu--Schwarz nodes, the choice of grading cannot be performed independently for different components, if it is required to be continuous. It is the real blow-up stage in the construction of the orbifold with corners that fixes this choice, up to a global change of grading in each connected component of the moduli space.
\end{rmk}

The above results can be carried out in greater generality.  First, if the images of the markings are any sets $B$ and $I$, one can clearly define the space $\M_{0,B,I}^{1/r}$ in the same way as above.    Furthermore, associated to each connected stable genus-zero twisted graded $r$-spin dual graph $\Gamma$, there is a closed suborbifold with corners $\M^{1/r}_{\Gamma} \subset \M_{0,B,I}^{1/r}$ whose general point is a graded $r$-spin disk with dual graph $\Gamma$. We also allow for the possibility that $\Gamma$ is disconnected, in which case $\M_{\Gamma}^{1/r}$ is defined as the product of the moduli spaces $\M_{\Gamma_i}^{1/r}$ associated to its connected components. Inside $\M^{1/r}_{\Gamma}$, we define $\CM_\Gamma^{1/r}$ as the open suborbifold consisting of graded $r$-spin disks whose dual graph is precisely $\Gamma.$

There are forgetful maps between the moduli spaces, but we note that marked points can only be forgotten if their twist is zero (otherwise \eqref{Sr} is not preserved), and boundary marked points can only be forgotten if they are in addition illegal (otherwise the grading does not descend to the moduli space with fewer marked points).  We define
\begin{equation}\label{eq:forgetful_map}\text{For}_{B',I'}: \M^{1/r}_{\Gamma} \rightarrow \M^{1/r}_{\Gamma'}\end{equation}
for $B',I' \subset \Z$ by forgetting all twist-zero internal marked points marked by $I'$ and all twist-zero illegal boundary marked points marked by $B'$.  This process may create unstable components; we repeatedly contract them.  If the process ends with some unstable components we remove them.  We denote by $\text{for}_{B',I'}(\Gamma)$ the graph $\Gamma'$ resulting from this procedure.

\subsection{The orientation of $\oCMr$}\label{sec:ap-or}
In the following subsection, we describe a natural orientation on the spaces $\oCMr$, thereby completing the proof of Theorem~\ref{thm:moduli}.  The ideas presented here are not new; in particular, they are similar to those presented in \cite[Section 2.5]{PST14} and are closely related to the earlier discussion in~\cite[Section 2.1.2]{FO09}.  First, we reduce the question of orientability to a simpler setting:

\begin{obs}\label{obs:or_with_without_spin}
We claim that the moduli space $\oCMr$ is orientable exactly if $\CM_{0,k,l}$ is orientable, and that an orientation on $\oCM_{0,k,l}$ induces one on $\oCMr$ by pullback under the map
\[\text{For}_{\text{spin}}:\oCMr\to\oCM_{0,k,l}\]
that forgets the graded spin structure.

To see this, note that the map $\text{For}_{\text{spin}}:\CMr\to\CM_{0,k,l}$ on the open moduli spaces is a diffeomorphism on the coarse underlying level, which means that the open locus $\CMr$ is indeed orientable precisely if $\CM_{0,k,l}$ is orientable.  To pass from the open locus to the full moduli space, one can construct $\oCMr$ (respectively, $\oCM_{0,k,l})$ from $\CMr$ (respectively, $\CM_{0,k,l}$) in two stages. First, add loci parameterizing disks without boundary nodes or contracted boundary components; such loci are of real codimension two and hence do not affect orientability.  Then, add the boundary of the moduli space; this contains strata of real codimension one, but the fact that they lie in the boundary means that they do not affect orientability.  This proves the claim.
\end{obs}

We henceforth discuss orientations on $\CM_{0,k,l}$, but, in light of Observation \ref{obs:or_with_without_spin}, all statements carry over to the spin case.  Furthermore, orientations can be studied one connected component at a time, and the the connected components of both $\M_{0,B,\vec{a}}^{1/r}$ and $\M_{0,B,\vec{a}}$ are indexed by cyclic orders of $B$:

\begin{nn}\label{nn:comps_of_moduli}
If $\pi: [|B|] \rightarrow B$ is an order of $B$, we denote its induced cyclic order by $\hat{\pi}$, and we denote by $\M_{0,B,\vec{a}}^{1/r,\hat\pi}$ (respectively, $\M_{0,B,\vec{a}}^{\hat\pi}$) the connected component of $\M_{0,B,\vec{a}}^{1/r}$ (respectively, $\M_{0,B,\vec{a}}$) that parameterizes disks for which the cyclic order of boundary markings taken along the boundary of the disk, with its canonical orientation, is $\hat\pi$. We denote by $\CMm_{0,k,l},\oCMm_{0,k,l} \subset \oCM_{0,k,l}$ the subspaces where the induced cyclic order on the boundary marked points is the cyclic order induced from the standard order $\pi^{\text{std}}$ on $[k].$
\end{nn}

We denote by $\text{Ord}(B)$ the set of all orders of $B$, by $\text{Cyc}(B)$ the set of cyclic orders, and by $S_B$ the group of permutations of $B$.   Note that $S_B$ acts both on $\text{Ord}(B)$, by composition, and on $\M_{0,B,\vec{a}}^{1/r}$, by permuting markings.

\begin{definition}
Let $\{ \tilde{\mathfrak{o}}^\pi=\tilde{\mathfrak{o}}^\pi_{0,B,I} \}$ be a family of orientations, where $B$ runs over all sets of size $k$, $I$ runs over all sets of size $l$, $\pi$ runs over all orders of $B$, and $\tilde{\mathfrak{o}}^\pi_{0,B,I}$ is an orientation for $\oCM_{0,B,I}^{\hat\pi}$.  We say such a family is \emph{covariant} if, whenever $f^B:B\to B'$ and $f^I:I\to I'$ are bijections and $F:\oCM_{0,B,I}\to\oCM_{0,B',I'}$ is the induced map, we have $\tilde{\mathfrak{o}}_{0,B,I}^\pi=F^* \tilde{\mathfrak{o}}_{0,B',I'}^{f^B\circ \pi}$.   A family $\{\tilde{\mathfrak{o}}^\pi=\tilde{\mathfrak{o}}^\pi_{0,B,\{\vec{a}\}}\}$ of orientations of $\oCM_{0,B,\{a_i\}_{i\in I}}^{\frac{1}{r},\hat\pi}$ is covariant if it is the pullback of a covariant family of orientations for $\oCM_{0,B,I}^{\hat\pi}.$
\end{definition}

The fiber of the forgetful map $\text{For}_{k+1}:\CM_{0,k,l+1} \to \CM_{0,k,l}$ is a punctured disk with a canonical complex orientation. For $k \geq 1,$ the fiber of the forgetful map $\CM_{0,k+1,l} \to \CM_{0,k,l}$ is a union of open intervals, so it is canonically oriented as the boundary of an oriented disk (as above). Denote this orientation by $o_{\text{For}_{k+1}^{-1}(\Sigma)}.$

\begin{prop}\label{prop:or_moduli}
Suppose $k+2l\geq 3.$ Then there exists a unique covariant family of orientations $\tilde{\mathfrak{o}}^{\pi}_{0,B,I}$ for the spaces $\oCM^{\hat{\pi}}_{0,B,I}$ with the following properties:
\begin{enumerate}[(1)]
\item\label{it:zero}
In the zero-dimensional case where $k = l=1$, the orientation is positive, while when $k = 3$ and $l=0$, the orientations are negative.
\item\label{it:perm}
Fix an integer $h$, and let $\pi\in \text{Ord}(B)$ and $g\in S_B$ be such that $g$ sends $x_{\pi(i)}$ to $x_{\pi(i+h)}$ cyclically.  Then $g$ preserves the orientation of $\oCM_{0,B,I}^{\hat\pi}$ if and only if $h(|B|-1)$ is even (this holds for any orientation of a moduli space of disks).   
\item\label{it:int}
The orientation $\tilde{\mathfrak{o}}^{\pi}_{0,k,l+1}$ agrees with the orientation induced from $\tilde{\mathfrak{o}}^{\pi}_{0,k,l}$ by the fibration $\oCM_{0,k,l+1} \to \oCM_{0,k,l}$ and the complex orientation on the fiber.
\item\label{it:bdry}
On $\oCMm_{0,k+1,l}$, we have $\tilde{\mathfrak{o}}_{0,k+1,l}^{\pi^{\text{std}}}=o_{\text{For}_{k+1}^{-1}(\Sigma)}\otimes \text{For}_{k+1}^*\tilde{\mathfrak{o}}_{0,k,l}{\pi^{\text{std}}}.$
\end{enumerate}
\end{prop}
\begin{rmk}\label{rmk:different_or_from_PST14}
For $k$ odd, the orientations described here differ from those of \cite{PST14} by $(-1)^{\frac{k-1}{2}}.$ This choice is more natural from the point of view of integrable hierarchies.
\end{rmk}
\begin{proof}[Proof of Proposition~\ref{prop:or_moduli}]
If the orientations $\tilde{\mathfrak{o}}_{0,k,l}^\pi$ exist, then properties~\eqref{it:zero} -- \eqref{it:bdry} imply that they are unique.  It remains to check existence.

For property~\eqref{it:perm} to hold, we must show that permutations of labels that map the component $\oCMm_{0,k,l}$ to itself affect the orientation according to their sign. This can be checked with respect to any orientation. Write
\[
U  = \left\{(z,w)\left|
\begin{array}{lll}
z = (z_1,\ldots,z_k) \in (S^1)^k, & z_i \neq z_j, & i\neq j \\
w = (w_1,\ldots,w_l) \in (\ior D^2)^l, & w_i \neq w_j, & i \neq j
\end{array}
\right.\right\}.
\]
Denote by $U^{\text{main}} \subset U$ the subset where the cyclic order of $z_1,\ldots,z_k$ on $S^1 = \partial D^2$ (with respect to the orientation induced from the complex orientation of $D^2$) agrees with the standard order of $[k]$.  Then
\[
\CMm_{0,k,l} = U^{\text{main}}/\text{PSL}_2(\R).
\]
When $k$ is odd, cyclic permutations of the boundary labels preserve the orientation of $U^{\text{main}}$ and thus also $\CM^{\text{main}}_{0,k,l}$ and $\oCM^{\text{main}}_{0,k,l}$.  When $k$ is even, a cyclic permutation of boundary labels that moves each boundary label by $h$ multiplies the orientation by the sign $(-1)^h$.  Renaming internal markings is a complex map that preserves orientations trivially, and similarly, arbitrary permutations of the interior labels preserve the orientation of $\oCM^{\text{main}}_{0,k,l}$.

A direct calculation shows that the orientation on $\oCM_{0,3,1}$ induced by property~\eqref{it:int} from $\tilde{\mathfrak{o}}_{0,3,0}^\pi$ agree with the orientation induced by property~\eqref{it:bdry} from $\tilde{\mathfrak{o}}^\pi_{0,1,1}$.  Thus, the required $\tilde{\mathfrak{o}}_{0,3,1}^\pi$ exists. Existence of $\tilde{\mathfrak{o}}_{0,k,l}^\pi$ satisfying properties~\eqref{it:int} and~\eqref{it:bdry} for other $k$ and $l$ follows from the commutativity of the diagram of forgetful maps
\[
\xymatrix{
\oCMm_{0,k+1,l+1} \ar[r]\ar[d] & \oCMm_{0,k+1,l} \ar[d] \\
\oCMm_{0,k,l+1} \ar[r] & \oCMm_{0,k,l}.
}
\]
\noindent Covariance, at this point, gives a unique way to extend the orientations to other connected components and to moduli spaces for different $B,I.$
\end{proof}

\begin{nn}\label{nn:or_for_moduli_spin}
Denote by $\tilde{\mathfrak{o}}^\pi_{0,B,\{a_i\}_{i\in I}}$ the orientation on $\oCM_{0,B,\{a_i\}_{i\in I}}^{1/r,\hat{\pi}}$ defined as $\text{For}_{\text{spin}}^*\tilde{\mathfrak{o}}^\pi_{0,B,I}.$
\end{nn}

\begin{lemma}\label{lem:moduli-induced_or}
The orientations $\tilde{\mathfrak{o}}^\pi_{0,B,\{a_i\}_{i\in I}}$ satisfy the following two properties:
\begin{enumerate}[(1)]
\item\label{it:induced or to boundary moduli}
Write $I = I_1 \sqcup I_2$, and take $B=B_1\sqcup B_2$.  Let $\Gamma$ be the graph with two open vertices, $v_1$ and $v_2$, connected by an edge $e$, where vertex $v_i$ has internal tails labeled by $I_i$ and $k_i$ boundary tails labeled by $B_i$, and let $h_i$ denote the half-edges of $v_i$. Let $\hat\pi$ be a cyclic ordering of $B$ in which all tails of $v_1$ are consecutive and all tails of $v_2$ are consecutive.  Denote by $\pi$ the unique ordering of $B$ such that, for any $\Sigma\in \CM_\Gamma^{\hat\pi}$ with normalization $\Sigma_1 \sqcup \Sigma_2$ (where $\Sigma_i$ corresponds to $v_i$), we have
\begin{itemize}
\item under $\pi$, the marked points of $\Sigma_1$ appear before those of $\Sigma_2$;
\item when $\pi$ is restricted to the points of $\Sigma_i$, it agrees with the order of the points on $\partial\Sigma_i$ with its natural orientation, starting after the node.
\end{itemize}
Let $\pi_1$ be the restriction of $\pi$ to the points of $\Sigma_1$, but adding the half-node $x_{h_1}$ in the end, and let $\pi_2$ be the restriction to points of $\Sigma_2$, but adding the half-node $x_{h_2}$ in the beginning.  Note that we have $\det(T\oCM^{\frac{1}{r}}_{0,B,\{a_i\}_{i\in I}})|_{\CM_\Gamma^{\hat\pi}} = \det(N)\otimes \det( T\CM_{\Gamma}^{\hat\pi}),$ where $N$ is the outward normal with canonical orientation $o_N$.  Then \[\tilde{\mathfrak{o}}^\pi|_{\CM_\Gamma^{\hat\pi}}=(-1)^{(|B_1|-1)|B_2|}o_N\otimes
(\tilde{\mathfrak{o}}^{\pi_1}_{0,B_1\cup\{h_1\},I_1}\boxtimes\tilde{\mathfrak{o}}^{\pi_2}_{0,\{h_2\}\cup B_2,I_2}).\]
\item\label{it:induced or to interior moduli}
Let $\Gamma\in\partial\Gammar$ be a graph with two vertices, an open vertex $v^o$ and a closed vertex $v^c$.  We have $\det(T\CM_\Gamma) = \det(N)\boxtimes \det(T\CM_{v^c})\boxtimes \det(T\CM_{v^o})$, where $N$ is again the normal bundle.  Then, for any order $\pi,$
\begin{equation}\label{eq:induced_or_interior_moduli}
\tilde{\mathfrak{o}}^\pi|_{\CM_\Gamma} = o_N\otimes (\tilde{\mathfrak{o}}^\pi_{v^o}\boxtimes \tilde{\mathfrak{o}}_{v^c}),
\end{equation}
where $o_N$ and $\tilde{\mathfrak{o}}_{v^c}$ are the canonical complex orientations.
\end{enumerate}
\end{lemma}
\begin{proof}
We prove the first item (which is analogous to Lemma 3.16 in \cite{PST14}) by induction on the dimension $\oCM_{0,B,I}^{1/r}$. By Observation \ref{obs:or_with_without_spin}, it is enough to prove the analogous claim for $\oCM_{0,k,l}$.  The proof of the second item uses exactly the same arguments, so it is omitted.

Covariance shows that it is enough to prove the claim when $B=[k]$ and $\pi$ is the order $1,\ldots, k$.  The base cases where the moduli space has dimension one or two can be checked by hand.  Suppose, then, that the desired statement holds for all moduli spaces of dimension less than $n$, where $n \geq 3$.  After possibly interchanging the roles of $v_1$ and $v_2$, we can assume that a tail of $v_2$ can be forgotten without affecting stability; here, we use item \eqref{it:perm} of Proposition \ref{prop:or_moduli} and the fact that $\dim_\R \CM_{v_i} \equiv |B_i| \mod 2$ to see that the interchanging affects the equation with the correct sign.  Let $\Gamma'$ be the stable graph obtained by removing a tail of $v_2$.  If this tail is a boundary tail, assume it is labeled $k$.  Write $v'_2$ for the resulting vertex of $\Gamma'$ and $\pi'$ for the restriction of $\pi$ to $[k-1].$

Consider first the case where the forgotten tail corresponds to a boundary point $x_k$.  With the notation of Proposition \ref{prop:or_moduli}, we have $\tilde{\mathfrak{o}}^\pi_{0,k,l}=o_{\text{For}_{k}^{-1}(\Sigma)}\otimes \text{For}_k^*\tilde{\mathfrak{o}}^{\pi'}_{0,k,l}$, so
\begin{equation}
\label{eq:1_for_mod}
\tilde{\mathfrak{o}}^\pi|_{\CM_\Gamma}=o_{\text{For}_{k}^{-1}(\Sigma)}\otimes \text{For}_k^*\tilde{\mathfrak{o}}^{\pi'}|_{\CM_{\Gamma'}}.
\end{equation}
By induction, we have
\begin{equation}\label{eq:2_for_mod}
\tilde{\mathfrak{o}}^{\pi'}|_{\CM_{\Gamma'}} = (-1)^{(|B_1|-1)(|B_2|-1)}o_N\otimes(\tilde{\mathfrak{o}}^{\pi_1}_{v_1}\boxtimes\tilde{\mathfrak{o}}^{\pi_2}_{v'_2}).
\end{equation}
We can identify $\text{For}_{k}^{-1}(N|_{\CM_{\Gamma'}}) \cong N_{\CM_\Gamma},$ and the identification preserves natural orientations.
Finally,
\begin{equation}\label{eq:3_for_mod}
\tilde{\mathfrak{o}}^{\pi_2}_{v_2}=o_{\text{For}_{k}^{-1}(\Sigma)}\otimes \text{For}_k^*\tilde{\mathfrak{o}}^{\pi'_2}_{v_2'}.
\end{equation}
Putting equations \eqref{eq:1_for_mod}, \eqref{eq:2_for_mod}, and \eqref{eq:3_for_mod} together, and recalling that $\dim_\R \CM_{v_1} \equiv |B_1| \mod 2$, we obtain the result.

If the forgotten tail corresponds to an internal marked point labeled $i$, then
\begin{equation}\label{eq:1'_for_mod}
\tilde{\mathfrak{o}}^\pi_{0,k,l}=o_{\text{For}_{i}^{-1}(\Sigma)}\otimes \text{For}_i^*\tilde{\mathfrak{o}}^{\pi}_{0,k,l-1}\Rightarrow
\tilde{\mathfrak{o}}^\pi|_{\CM_\Gamma}=o_{For_{i}^{-1}(\Sigma)}\otimes For_i^*\tilde{\mathfrak{o}}^{\pi}|_{\CM_{\Gamma'}},
\end{equation}
where we abuse notation somewhat by using $\text{For}_i$ to denote the map forgetting the $i$th internal marked point.  By induction, we have
\begin{equation}\label{eq:2'_for_mod}
\tilde{\mathfrak{o}}^{\pi'}|_{\CM_{\Gamma'}} = (-1)^{(|B_1|-1)|B_2|}o_N\otimes\tilde{\mathfrak{o}}^{\pi_1}_{v_1}\boxtimes\tilde{\mathfrak{o}}^{\pi_2}_{v'_2}.
\end{equation}
Observe that $\text{For}_{i}^{-1}(N|_{\CM_{\Gamma'}})\cong N_{\CM_\Gamma},$ and the equation preserves natural orientations.  Finally,
\begin{equation}\label{eq:3'_for_mod}
\tilde{\mathfrak{o}}^{\pi_2}_{v_2}=o_{\text{For}_{i}^{-1}(\Sigma)}\otimes \text{For}_i^*\tilde{\mathfrak{o}}^{\pi'_2}_{v_2'}.
\end{equation}
Putting equations \eqref{eq:1'_for_mod}, \eqref{eq:2'_for_mod}, and \eqref{eq:3'_for_mod} together, and noting that $\dim_\R \text{For}_{i}^{-1}(\Sigma)=2$, we obtain the result.
\end{proof}

\section{Associated vector bundles}\label{sec:bundles}

\subsection{The Witten bundle}\label{subsec:Wittenbdl}

In closed genus-zero $r$-spin theory, the virtual fundamental class is defined by way of the {\it Witten bundle} $(R^1\pi_*\mathcal{S})^{\vee}$, where $\pi: \mathcal{C} \rightarrow \M_{0,n}^{1/r}$ is the universal family and $\mathcal{S} \rightarrow \mathcal{C}$ is the universal twisted $r$-spin bundle.  This is an orbifold vector bundle with fiber
\[H^1(C, S)^{\vee} \cong H^0(C, J),\]
of complex rank $\frac{\sum a_i - (r-2)}{r}$.

\begin{obs}
\label{obs:twist-1}
As observed in \cite{JKV2} and studied in detail in \cite{BCT_Closed_Extended}, the Witten bundle is a bundle as long as at most one marked point has twist $-1$ and all other twists are non-negative, since this ensures that $R^0\pi_*\mathcal{S} = 0$.
\end{obs}

We now define the open analogue of the Witten bundle. 
Denote by $\pi: \mathcal{C} \rightarrow \M^{1/r,\Z_2}_{0,k+2l}$ the universal curve over the moduli space of real spin spheres defined above, and by $\cS \rightarrow \mathcal{C}$ the universal spin bundle.  Then $R^1\pi_* \cS$ is a vector bundle, since \eqref{Sr} implies that spin structures have negative degree and hence $R^0\pi_*\cS = 0$.  There are universal involutions
\[\phi: \mathcal{C} \rightarrow \mathcal{C} \; \text{ and } \; \widetilde{\phi}: \cS \rightarrow \cS,\]
which induce an involution on $R^1\pi_*\mathcal{S}$.  Let
\[\mathcal{W}^{\text{pre}} := (R^0\pi_*\mathcal{J})_+ =(R^1\pi_*\mathcal{S})^{\vee}_-\]
be the vector bundle of $\widetilde{\phi}$-invariant sections of $J$, where $\mathcal{J}:= \mathcal{S}^{\vee} \otimes \omega_{\pi}$; the second equality uses Serre duality, under which invariant sections become anti-invariant.  From here, the {\it open Witten bundle} is defined as
\[\mathcal{W}:= \varpi^*\mathcal{W}^{\text{pre}},\]
where $\varpi: \M^{1/r}_{0,k,l} \rightarrow \M^{1/r,\Z_2}_{0,k+2l}$ is the composition of the morphisms defined above.

\begin{rmk}\label{rmk:subscript}
We pull back $\cW$ from the moduli space $\M^{1/r,\Z_2}_{0,k+2l}$ of spheres in order to avoid the need to define derived pushforward in the orbifold-with-corners context.  To avoid cluttering the notation in what follows, however, we often write $\mathcal{W} = (R^0\pi_*\mathcal{J})_+$ even on $\M^{1/r}_{0,k,l}$.  Whenever we write such expressions, they should be understood as pulled back under $\varpi$. In addition, we sometimes write $\mathcal{W} = (R^0\pi_*\mathcal{J})_+$ on $\M^{1/r}_{0,n}$, where no involution is needed; in this case, the subscript should be ignored and the equation is to be interpreted as $\mathcal{W} = R^0\pi_*\mathcal{J}.$
\end{rmk}

\begin{rmk}\label{rmk:Witten-coarseVSorbi}
Due to the canonical isomorphism
\begin{equation}
\label{eq:H0(J)}
H^0(J) = H^0(|J|),
\end{equation}
the fibers of the Witten bundle can equivalently be viewed as sections of $|J|$ on $|C|$.  Furthermore, if $C$ is a graded $r$-spin disk and $p \in C$ is a non-orbifold point, then the fiber $J_p$ is identified with the fiber $|J|_{\rho(p)}$ over the image point $\rho(p) \in |C|$.  In particular, if $s \in H^0(J)$ is an element of the fiber of the Witten bundle over $C$ corresponding under \eqref{eq:H0(J)} to $\rho_*(s) \in H^0(|J|)$, then the evaluation of $s$ at $p$ agrees under the identification $J_p = |J|_{\rho(p)}$ with the evaluation of $\rho_*(s)$ at $\rho(p)$.  Because of these observations, we view the fibers of the Witten bundle interchangeably as $H^0(J)$ or as $H^0(|J|)$ in what follows.
\end{rmk}

The real rank of $\mathcal{W}$ (which is the number $e$ defined in \eqref{eq:open_rank1} after setting $g=0$) is
\begin{equation}
\label{eq:rank_of_W}\frac{2\sum a_i + \sum b_j - (r-2)}{r}
\end{equation}
indeed, a standard Riemann--Roch calculation shows that \eqref{eq:rank_of_W} is the complex rank of $R^0\pi_*\mathcal{J}$, and taking involution-invariant parts reduces the real rank by half.  Furthermore, in the notation of \eqref{eq:forgetful_map}, there is a canonical isomorphism
\begin{equation}\label{eq:ForW}\mathcal{W} = \text{For}_{B',I'}^*\mathcal{W}\end{equation}
for any subsets $B',I' \subset \Z$.  This identification involves replacing a graded $r$-spin disk $C$ by a partially coarsened (and possibly stabilized) disk $C'$.  However, in the same vein as Remark~\ref{rmk:Witten-coarseVSorbi}, if $s \in H^0(J)$ is an element of the fiber of $\mathcal{W}$ over $C$ that corresponds under \eqref{eq:ForW} to an element $\text{For}^*_{B',I'}(s)$ in the fiber of $\text{For}_{B',I'}^*\mathcal{W}$ over $C'$, and if $p \in C$ is a non-special point whose component is not stabilized when passing to $C'$, then the evaluation of $s$ at $p$ coincides under \eqref{eq:ForW} with the evaluation of $\text{For}^*_{B',I'}(s)$ at the image of $p$.

%

\subsection{Decomposition properties of the Witten bundle}\label{subsec:decomp}

The open Witten bundle, like its closed analogue, satisfies a decomposition property along nodes.  In order to state the property, we must define a normalization morphism on the moduli spaces, which can be described by a ``detaching" operation on graphs.

\begin{definition}
\label{def:detach}
Let $\Gamma$ be a genus-zero graded $r$-spin graph, and let $e$ be an edge of $\Gamma$ with half-edges $h$ and $h'$.  Then the {\it detaching} of $\Gamma$ at $e$ is the disconnected graph
\[\detach_e(\Gamma) = (V', H', \sigma_0', \sim', H'^{CB}, m'),\]
defined to agree with $\Gamma$ except that $h \not \sim' h'$.  We keep $\alt$ and $\twE$ the same, and we extend the marking and the anchor as follows.   If $e$ is a boundary edge, set $m'(h) = m'(h') = 0$.  If $e$ is an internal edge, then exactly one of the components of $\detach_e(\Gamma)$ is closed and unanchored; suppose, without loss of generality, that this is the component containing $h$.  Then we set $h$ to be the anchor of its component, and we set $m'(h) = 0$ and $m'(h')$ to be the union of the markings of the internal tails $h'' \neq h$ in the same component as $h$.

If $t \in H^{CB}$ is a contracted boundary tail, then the {\it detaching} of $\Gamma$ at $t$ is the graph $\detach_t(\Gamma)$ defined to agree with $\Gamma$ except that $t \in (T')^I \setminus (H')^{CB}$.  We keep $\alt$ and $\twE$ the same and leave $t$ as the anchor.
\end{definition}

Note that the new internal and boundary markings still satisfy the requirements of Definitions \ref{def:non_spin_graph} and \ref{def:graph}.  In particular, since there is a canonical identification of $E(\Gamma) \setminus \{e\}$ with the edges of $\detach_e(\Gamma)$ and of $E(\Gamma)$ with the edges of $\detach_t(\Gamma)$, one can also iterate the detaching process.  For any subset $N \subset E(\Gamma) \cup H^{CB}(\Gamma)$, we denote by $\detach _{N}(\Gamma)$ the graph obtained by performing $\detach_f$ for each element $f \in N$; the result is independent of the order in which the detachings are performed. When we write $\detach(\Gamma)$ without any subscript, we mean $\detach_{E(\Gamma)}(\Gamma).$

Let $\Gamma$ be a stable $r$-spin dual graph, and let $\widehat{\Gamma} = \detach_{N}(\Gamma)$ for some set $N \subset E(\Gamma) \cup H^{CB}(\Gamma)$ of edges and contracted boundary tails.  Unlike the moduli space of curves, the $r$-spin moduli space does not always have a gluing map $\M_{\widehat\Gamma}^{1/r} \rightarrow \M_{0,k,l}^{1/r}$, because there is no canonical way to glue the fibers of the spin bundle at the internal nodes.  Instead, we consider the following diagram of morphisms:
\begin{equation}\label{eq:instead_of_product}\M_{\widehat\Gamma}^{1/r} \xleftarrow{q} \M_{\widehat\Gamma} \times_{\M_{\Gamma}} \M_{\Gamma}^{1/r} \xrightarrow{\mu} \M_{\Gamma}^{1/r} \xrightarrow{i_{\Gamma}} \M_{0,k,l}^{1/r}.\end{equation}
Here, $\M_{\Gamma},~\M_{\widehat\Gamma}$ are the moduli spaces of marked disks with dual graphs $\Gamma,\widehat{\Gamma}$ respectively. The morphism $q$ is defined by sending the spin structure $S$ to $\widehat{S}$ as in \eqref{eq:widehatS}.

The map $\mu$ is an isomorphism, though we distinguish between its domain and codomain because they have different universal objects.  While the map $q$ is in general not an isomorphism (see Remark \ref{rmk:q not iso} below), it has degree one.  This fact is known in the closed case (see, for example, \cite{CZ}); the key point is that each element of $\M_{\widehat\Gamma}^{1/r}$ has a single geometric preimage under $q$, and in both the domain and codomain of $q$, the order of the automorphism group is $r^c$, where $c$ is the number of components.  By the analysis of automorphisms in Proposition \ref{prop:existence_stable}, the same argument applies in the open case.

\begin{rmk}
\label{rmk:q not iso}
The reason that $q$ is not, in general, an isomorphism is that it does not induce an isomorphism on automorphism groups.  Automorphisms on a normalized curve are given by separate fiberwise scalings on each sphere component; on a nodal curve, however, scalings by $\zeta$ and $\eta$ at opposite branches of a node only glue to give a global automorphism if one can act by a ghost automorphism of $C$ to make the scalings at the two branches agree---in other words (by Remark~\ref{rmk:auts_of_spin}), if there exists $\xi \in \Z/r\Z$ for which $\zeta \cdot \xi^m = \eta$, where $m$ is the multiplicity of $S$ at the $\zeta$-branch of the node.  Thus, $q$ is an isomorphism only if each internal node satisfies $\text{gcd}(r,m)=1$, so that such $\xi$ exists.  (Boundary nodes need not satisfy any condition, since fiberwise scaling on a disk component is not an automorphism of the graded spin structure.)
\end{rmk}

\begin{nn}\label{nn:Detach}
For $\Gamma$ and $N$ as above, we denote the map $q\circ \mu^{-1}:\M_{\Gamma}^{1/r}\to\M_{\widehat\Gamma}^{1/r}$ by $\text{Detach}_{N}.$ When $N=\{e\}$ is a singleton, we denote this map by $\text{Detach}_{e}.$
\end{nn}

There are two natural universal curves over the fiber product $\M_{\widehat\Gamma} \times_{\M_{\Gamma}} \M_{\Gamma}^{1/r}$: we define $\mathcal{C}_{\Gamma}$ by the fiber diagram
\[\xymatrix{
\mathcal{C}_{\Gamma} \ar[r]\ar[d]_{\pi} & \mathcal{C}\ar[d]\\
\M_{\widehat\Gamma} \times_{\M_{\Gamma}} \M_{\Gamma}^{1/r} \ar[r]^-{i_{\Gamma} \circ \mu} & \M_{0,k,l}^{1/r},}\]
and $\mathcal{C}_{\widehat{\Gamma}}$ by the fiber diagram
\[\xymatrix{
\widehat{\mathcal{C}}\ar[d] & \mathcal{C}_{\widehat{\Gamma}}\ar[l]\ar[d]^{\widehat{\pi}}\\
\M_{\widehat{\Gamma}}^{1/r} & \M_{\widehat\Gamma} \times_{\M_{\Gamma}} \M_{\Gamma}^{1/r}\ar[l]_-{q},}\]
in which $\widehat{\mathcal{C}}$ is the universal curve over $\M_{\widehat\Gamma}^{1/r}$.  There are universal bundles $\mathcal{S}$ and $\widehat{\mathcal{S}}$ on these two universal curves, and they are related by a universal normalization morphism
\[n: \mathcal{C}_{\widehat\Gamma} \rightarrow \mathcal{C}_{\Gamma}.\]

We can now state the decomposition properties of the Witten bundle.  We state the properties in the case where $N = \{e\}$ for a single edge $e$, but all can be readily generalized to the setting where more than one edge is detached.

\begin{pr}
\label{pr:decomposition}
Let $\Gamma$ be a stable genus-zero twisted $r$-spin dual graph with a lifting.  Suppose that $\Gamma$ has a single edge $e$, so the general point of $\M_{\Gamma}^{1/r}$ is a stable $r$-spin disk with two components $C_1$ and $C_2$ meeting at a node $p$.  Let $\widehat{\Gamma} = \detach_e(\Gamma)$.  Let $\cW$ and $\widehat\cW$ denote the Witten bundles on $\M_{0,k,l}^{1/r}$ and $\M_{\widehat\Gamma}^{1/r}$, respectively.

Then, topologically, the Witten bundle decomposes as follows along the node $p$:
\begin{enumerate}[(i)]
\item\label{it:NS} If $e$ is a Neveu--Schwarz edge, then $\mu^*i_{\Gamma}^*\cW = q^*\widehat\cW$.

\item\label{it:Ramondbdryedge} If $e$ is a Ramond boundary edge, then there is an exact sequence
\begin{equation}
\label{eq:decompses}0 \rightarrow \mu^*i_{\Gamma}^*\mathcal{W} \rightarrow q^*\widehat{\mathcal{W}} \rightarrow \TTT_+ \rightarrow 0,
\end{equation}
where $\TTT_+$ is a trivial real line bundle.

\item If $e$ is a Ramond internal edge connecting two closed vertices, write $q^*\widehat\cW = \widehat\cW_1 \boxplus \widehat\cW_2$, in which $\widehat\cW_1$ is the Witten bundle on the component containing the anchor of $\Gamma$ (defined via $\widehat\cS|_{\mathcal{C}_1}$) and $\widehat\cW_2$ is the Witten bundle on the other component.  Then there is an exact sequence
\begin{equation}
\label{eq:decompses2}
0 \rightarrow \widehat\cW_2 \rightarrow \mu^*i_{\Gamma}^*\cW \rightarrow \widehat\cW_1 \rightarrow 0.
\end{equation}
Furthermore, if $\widehat\Gamma'$ is defined to agree with $\widehat\Gamma$ except that the twist at each Ramond tail is $r-1$, and $q': \M_{\widehat{\Gamma}} \times_{\M_{\Gamma}} \M_{\Gamma}^{1/r} \rightarrow \M_{\widehat\Gamma'}^{1/r}$ is defined analogously to $q$, then there is an exact sequence
\begin{equation}
\label{eq:decompses3}
0 \rightarrow \mu^*i_{\Gamma}^*\mathcal{W} \rightarrow (q')^*\widehat{\mathcal{W}}' \rightarrow \TTT \rightarrow 0,
\end{equation}
where $\widehat{\mathcal{W}}'$ is the Witten bundle on $\M_{\widehat\Gamma'}^{1/r}$ and $\TTT$ is a line bundle whose $r$th power is trivial.

\item If $e$ is a Ramond internal edge connecting an open vertex to a closed vertex, write $q^*\widehat\cW = \widehat\cW_1 \boxplus \widehat\cW_2$, in which $\widehat\cW_1$ is the Witten bundle on the disk component (defined via $\widehat\cS|_{\mathcal{C}_1}$) and $\widehat\cW_2$ is the Witten bundle on the sphere component.  Then the exact sequences \eqref{eq:decompses2} and \eqref{eq:decompses3} both hold.
\item\label{it:cont_bdry_tail} Suppose that $\Gamma$ has a single vertex, no edges, and a contracted boundary tail $t$, and let $\widehat{\Gamma} = \detach _t(\Gamma)$.  If $\cW$ and $\widehat\cW$ denote the Witten bundles on $\M_{0,k,l}^{1/r}$ and $\M_{\widehat\Gamma}^{1/r}$, respectively, then the sequence \eqref{eq:decompses} also holds in this case.
\end{enumerate}

\begin{rmk}
\label{rmk:coarse}
We say that the Witten bundle decomposes ``topologically" as above to emphasize that, while the coarse underlying Witten bundle behaves as above, the action of the isotropy groups of the moduli space on the fibers may not respect these identifications.  (This is only relevant for internal nodes, since boundary and contracted boundary nodes do not contribute isotropy to the moduli space.)  For example, in the case of a Neveu--Schwarz internal node, an element of the fiber of $q^*\widehat\cW$ is acted on by independently scaling the sections of $J$ on the two components, while Remark~\ref{rmk:q not iso} shows that this is not in general possible for elements of the fiber of $\mu^*i_{\Gamma}^*\cW$.
\end{rmk}

\begin{proof}[Proof of Proposition~\ref{pr:decomposition}]
First, let us fix some notation.  Letting $\cJ = \cS^{\vee} \otimes \omega_{\pi}$ and decomposing $\mathcal{C}_{\widehat{\Gamma}}$ into components $\mathcal{C}_1$ and $\mathcal{C}_2$, we define:
\[\pi_i = \widehat{\pi}|_{\mathcal{C}_i},\;\;\; \cS_i = n^*\cS|_{\mathcal{C}_i},\;\;\;
\widehat\cS_i = \widehat\cS|_{\mathcal{C}_i},\;\;\;
\widehat\cJ_i = \widehat\cS_i^{\vee} \otimes \omega_{\pi_i}\]
for $i=1,2$.  We view $\mu^*i_{\Gamma}^*\cW = (R^0\pi_*\cJ)_+$ and $q^*\widehat{\cW} = (R^0\pi_{1*}\widehat\cJ_1 \oplus R^0\pi_{2*}\widehat\cJ_2)_+$.  

Suppose that $e$ is Neveu--Schwarz.  Then the normalization exact sequence yields
\[0 \rightarrow \cS \rightarrow n_*n^*\cS \rightarrow \cS|_{\Delta_p} \rightarrow 0,\]
where $\Delta_p \subset \mathcal{C}_{\Gamma}$ is the orbifold divisor corresponding to the node $p$.  Since the twist of every tail of $\Gamma$ is non-negative, except at most one tail that may have twist $-1$, we have $R^0\pi_{1*}\cS_1 = R^0\pi_{2*}\cS_2 = 0$, and we obtain
\begin{equation}
\label{eq:normS}
0 \rightarrow \sigma_p^*\cS \rightarrow R^1\pi_*\cS \rightarrow R^1\pi_{1*}\cS_1 \oplus R^1\pi_{2*}\cS_2 \rightarrow 0.
\end{equation}
The assumption that $e$ is Neveu--Schwarz implies both that $\widehat{\cS}_i = \cS_i$ for $i=1,2$ and that $\sigma_p^*\cS = 0$, since sections of an orbifold line bundle necessarily vanish at Neveu--Schwarz points.  Thus, dualizing and taking involution-invariant parts of \eqref{eq:normS} shows that the fibers of $\mu^*i_{\Gamma}^*\cW$ and $q^*\widehat\cW$ are canonically identified.

Suppose, now, that $e$ is a Ramond boundary edge.  Then
\[\widehat{\cS}_i = \cS_i \otimes \O(-r\Delta_{i,p}),\]
where $\Delta_{i,p} \subset \mathcal{C}_i$ is the orbifold divisor corresponding to the half-node in $C_i$.  The normalization exact sequence for $\cJ$ yields
\begin{equation}
\label{eq:normJ}
0 \rightarrow R^0\pi_*\mathcal{J} \rightarrow R^0\widehat{\pi}_*(n^*\mathcal{J}) \rightarrow \sigma_p^*\mathcal{J} \rightarrow 0.
\end{equation}
Now, passing to coarse underlying bundles (which does not affect cohomology), the middle term can be re-written:
\[\big|n^*\mathcal{J}|_{\mathcal{C}_i}\big| = |n|^*\bigg(|\mathcal{S}|^{\vee} \otimes \omega_{|\pi|}\bigg)\bigg|_{|\mathcal{C}_i|} = |\mathcal{S}_i|^{\vee} \otimes \O(\Delta_{i,p}) \otimes \omega_{|\pi|} = |\widehat{\mathcal{J}}_i|.\]
Thus, the sequence \eqref{eq:normJ} can be re-expressed as
\begin{equation}
\label{eq:Jhat12}
0 \rightarrow R^0\pi_*\mathcal{J} \rightarrow  R^0\pi_{1*}\widehat{\mathcal{J}}_1 \oplus R^0\pi_{2*}\widehat{\mathcal{J}}_2 \rightarrow \sigma_p^*\mathcal{J} \rightarrow 0.
\end{equation}
Because $e$ is Ramond, the bundle $\TTT:= (\sigma_p^*\cJ)$ has trivial $r$th power, using that the restriction of both $\omega_{\pi}$ and $\omega_{\pi,\log}$ to the locus of nodes is trivial. Taking involution-invariant parts in \eqref{eq:Jhat12} yields \eqref{eq:decompses}, where the real line $\TTT_+$ is trivialized using the grading.

Next, suppose that $e$ is a Ramond internal edge connecting two closed vertices.   Denote the two vertices of $\Gamma$ by $v_1$ and $v_2$, and let $v_1$ be the vertex supporting the anchor.  Then $R^0\pi_{2*}\cS_2 = 0$, and if $v_1$ contains at least one tail of positive twist, then $R^0\pi_{1*}\cS_1 = 0$, as well.  In this case, we still have the sequence \eqref{eq:normS}.  We also have
\begin{equation}
\label{eq:SShat}
0 \rightarrow \sigma_p^*\cS \rightarrow R^1\pi_{1*}\widehat\cS_1 \rightarrow R^1\pi_{1*}\cS_1 \rightarrow 0
\end{equation}
The sequences \eqref{eq:normS} and \eqref{eq:SShat} fit together into a commutative diagram
\[\xymatrix{
0 \ar[r] & \sigma_p^*\mathcal{S} \ar[d]\ar[r] & R^1\pi_*\widehat{\mathcal{S}}_1 \ar[r]\ar[d] & R^1\pi_*\mathcal{S}_1 \ar[d]\ar[r] &0\\
0 \ar[r] & \sigma_p^*\mathcal{S}\ar[r] & R^1\pi_*\mathcal{S} \ar[r] & R^1\pi_*\mathcal{S}_1 \oplus R^1\pi_*\mathcal{S}_2 \ar[r] & 0,
}\]
in which the middle vertical arrow can be constructed by dualizing the sequence \eqref{eq:Jhat12}.  By the Snake Lemma the cokernel of the middle vertical arrow is $R^1\pi_{2*}\cS_2 = R^1\pi_{2*}\widehat\cS_2$, so
\[0 \rightarrow R^1\pi_{1*}\widehat\cS_1 \rightarrow R^1\pi_*\cS \rightarrow R^1\pi_{2*}\widehat\cS_2 \rightarrow 0.\]
Dualizing and taking involution-invariant parts proves \eqref{eq:decompses2}.

Still assuming that $e$ is a Ramond internal edge connecting two closed vertices, suppose that every tail of $v_1$ except the anchor has twist zero. In this case the anchor must have twist $-1.$  Then $R^0\pi_{1*}\cS_1$ is one-dimensional, and in the normalization exact sequence
\[0 \rightarrow R^0\pi_{1*}\cS_1 \oplus R^1\pi_{2*}\cS_2 \rightarrow \sigma_p^*\cS \rightarrow R^1\pi_*\cS \rightarrow R^1\pi_{1*}\cS_1 \oplus R^1\pi_{2*},\]
the first map is an isomorphism.  Hence, $R^1\pi_*\cS \cong R^1\pi_{1*}\cS_1 \oplus R^1\pi_{2*}\cS_2$.  A similar argument shows that the first map in
\[0 \rightarrow R^0\pi_{1*}\cS_1 \rightarrow \sigma_p^*\cS_1 \rightarrow R^0\pi_{1*}\widehat\cS_1 \rightarrow R^0\pi_{1*}\cS_1 \rightarrow 0\]
is an isomorphism, so we also have $R^1\pi_{1*}\widehat\cS_1 \cong R^1\pi_{1*}\cS_1$.  Since $\widehat\cS_2 = \cS_2$ by construction, we conclude that $R^1\pi_*\cS \cong R^1\pi_{1*}\widehat\cS_1 \oplus R^1\pi_{2*}\widehat\cS_2$, which implies that \eqref{eq:decompses2} holds (and, in fact, splits) in this case.  The proof of \eqref{eq:decompses3} is identical to the proof of \eqref{eq:decompses} above.

If $e$ is a Ramond internal edge connecting an open vertex to a closed vertex, then the proof that \eqref{eq:decompses2} holds is identical to the proof for an edge joining two closed vertices (with $R^0\pi_{1*}\cS_1 = R^0\pi_{2*}\cS_2 = 0$), and the proof that \eqref{eq:decompses3} holds is exactly as in \eqref{eq:decompses}.  Finally, in the situation where $\Gamma$ has a single vertex, no edges, and a contracted boundary tail, the exact sequence \eqref{eq:decompses} still holds, by the same proof.
\end{proof}
\end{pr}

One further observation will be needed later:
\begin{obs}\label{obs:cases when pulled back from one component}
Suppose $I'\subseteq [l]$ is such that $\sum_{i\in I'} a_i\leq r-1$.  Let $\Gamma$ be a graph with exactly two vertices, a closed vertex $v^c$ containing exactly the internal tails labeled by $I'$ and an open vertex $v^o$, connected by a single edge.  Then $\M_{\widehat\Gamma}^{1/r} = \M_{v^c} \times \M_{v^o}$.
One can compute that the twist at the node is $r-2-\sum_{i\in I'} a_i$ and the Witten bundle $\cW_{v^c}$ on $\M_{v^c}$ has rank zero, so Proposition~\ref{pr:decomposition} implies that
\[\mu^*i_{\Gamma}^*\mathcal{W}_{\Gamma} = q^*\pi_2^*\mathcal{W}_{v^o},\]
where $\pi_2$ is the projection to the second factor.
\end{obs}

\subsection{The cotangent line bundles}

For each $i \in I$, a cotangent line bundle $\mathbb{L}_i$ is defined on the moduli space of stable marked disks as the line bundle whose fiber over $(C, \phi, \Sigma, \{z_i\}, \{x_j\}, m^I, m^B)$ is the cotangent line $T^*_{z_i}\Sigma$.  Alternatively, $\mathbb{L}_i$ is the pullback via the doubling map of the usual line bundle $\mathbb{L}_i$ on $\M_{0,k+2l}$.  We define cotangent line bundles $\mathbb{L}_i$ on $\M^{1/r}_{0,k,l}$ by pullback under the morphism forgetting the spin structure, and for any graph $\Gamma$, we let $\mathbb{L}_i^{\Gamma}$ be the pullback of $\mathbb{L}_i$ to $\M^{1/r}_{\Gamma}$.

We mention a few important properties of these bundles in the observation below. The proofs are identical to the proofs of the analogous statements in \cite[Section 3.5]{PST14}:

\begin{obs}\label{obs:L_i_properties}
\begin{enumerate}[(i)]
\item If $e$ is an edge of $\Gamma$, $\Gamma_1$ and $\Gamma_2$ are the two connected components of $\detach_e(\Gamma)$, and $i$ is a marking of an internal tail of $\Gamma_1$, then
\[\mathbb{L}_i^{\Gamma} = {\text{Proj}}_1^*\mathbb{L}_i^{\Gamma_1},\]
where ${\text{Proj}}_1: \M^{1/r}_{\Gamma} \rightarrow \M^{1/r}_{\Gamma_1}$ is the projection.
\item If $B' \subset \Z \setminus \{i\}$, $I' \subset \Z$, and $\Gamma' = \text{for}_{B',I'}(\Gamma)$, then there exists a canonical morphism
\[t_{\Gamma,B',I'} : \text{For}_{B',I'}^*\mathbb{L}_i^{\Gamma'} \rightarrow \mathbb{L}_i^{\Gamma},\]
which is an isomorphism away from strata where the component containing $z_i$ is contracted by the forgetful map, and which vanishes identically on the remaining strata.
\item $\mathbb{L}_i^{\Gamma}$ is canonically oriented as a complex line bundle.
\end{enumerate}
\end{obs}

\section{Orientation: constructions and properties}\label{sec:or}
\subsection{Relative orientation of the Witten bundle}\label{subsec:rel_or}
The open Witten bundle does not carry a canonical complex orientation.  Nevertheless, it is canonically relatively oriented, relative to the moduli space of stable graded $r$-spin disks.  Before stating this precisely, we require some notation.

\begin{nn}\label{nn:Gammar}
Denote by $\Gammar$ the connected graded smooth graph with a single vertex $v,$ boundary tails marked by $[k]$, and internal tails marked $[l]$, such that the $i$th internal tail has twist $a_i$.
\end{nn}
Let $\Gamma\in\d\Gammar$, which is an $r$-spin dual graph with a lifting consisting of two vertices $v_1$ and $v_2$ connected by a single edge $e$.  We have
\begin{equation}
\label{eq:TMdecomp0}
i_{\Gamma}^* \det(T\M^{1/r}_{0,B,\{a_i\}_{i \in I}}) \cong \det(N)  \otimes \det(T\oCM_{\detach_e(\Gamma)}^{1/r}),
\end{equation}
where $N$ is the normal bundle of $\M^{1/r}_{\Gamma}$; this can be identified with the tensor product of the tangent line bundles at the half-nodes on the moduli spaces $\M_{v_1}^{1/r}$ and $\M_{v_2}^{1/r}$ associated to the vertices.\footnote{To be more precise, the isomorphism \eqref{eq:TMdecomp0} holds only on the coarse level, as the actions of the isotropy groups of the moduli space on the fibers are not the same; see Remarks~\ref{rmk:q not iso} and \ref{rmk:coarse}.  However, these actions are clearly orientation-preserving, so this issue does not affect our orientation analysis, and will therefore be ignored below.} These two tangent lines are canonically oriented complex lines when $e$ is internal, and are canonically oriented real lines when $e$ is boundary. In both cases, $N$ carries a canonical orientation, where in the second case we orient $N$ by the outward-pointing normal.

Furthermore, for any decompositions $I = I_1 \sqcup I_2$ and $B = B_1 \sqcup B_2$ such that
\begin{align}
\label{eq:B12}
|B_1| - 1 = \sum_{i \in I_1} a_i \mod r, \;\;\; |B_2| = \sum_{i \in I_2} a_i \mod r
\end{align}
and hence
\begin{align*}
\frac{2\sum_{i \in I_2} a_i +|B_2|(r-2)}{r}\equiv |B_2|\mod 2,
\end{align*}
let $\Gamma = \Gamma_{I_1,I_2; B_1,B_2}$ be the graph with two open vertices $v_1$ and $v_2$ joined by a boundary edge $e$, in which $v_i$ contains the internal tails marked by $I_i$ and the boundary tails marked by $B_i$.  If $h_i$ is the half-edge of $e$ incident to $v_i$, then a straightforward computation shows that
\[\text{tw}(h_1) = \alt(h_1) = 0.\]
We let $\mathcal{B}\Gamma$ denote the graph with vertices $v'_1$ and $v_2$ obtained by detaching $e$ and forgetting $h_1$; this is a very special case of the notion of the ``base" of a graph $\Gamma$, defined in general in 
the sequel. Let $\M^{1/r}_{v'_1}$, $\M^{1/r}_{v_1}$, and $\M^{1/r}_{v_2}$ be the moduli spaces of stable graded $r$-spin disks corresponding to the vertices $v'_1$, $v_1$, and $v_2,$ thought of as $r$-spin graphs.

By the discussion in Section~\ref{subsec:mod_const}, 
the maps $q$ and $\mu$ are isomorphisms in the case of a single boundary edge, which implies that
\[\M_{\Gamma}^{1/r} \cong \M_{v_1}^{1/r} \times \M_{v_2}^{1/r}.\]
Composing with the forgetful map forgetting the twist-zero tail $h_1$ on $v_1$, we obtain a map
\[F_{\Gamma}: \M^{1/r}_{\Gamma} \rightarrow \M^{1/r}_{\CB\Gamma}\cong \M_{v'_1}^{1/r} \times \M_{v_2}^{1/r}.\]
By Proposition~\ref{pr:decomposition} and the fact that the Witten class pulls back under the forgetful map, 
\begin{equation}
\label{eq:Wdecomp}
i_{\Gamma}^*\mathcal{W} = F_{\Gamma}^*\mathcal{W}_{\CB\Gamma} = \mathcal{W}_{v_1'} \boxplus \mathcal{W}_{v_2}.
\end{equation}
In this situation, we also have
\begin{equation}
\label{eq:TMdecomp}
i_{\Gamma}^* \det(T\M^{1/r}_{0,B,\{a_i\}_{i \in I}}) \cong \det(N) \otimes \det(f_{\Gamma}) \otimes F_\Gamma^*\det(T\M^{1/r}_{\mathcal{B}\Gamma}),
\end{equation}
where $N$ is again the normal bundle of $\M^{1/r}_{\Gamma},$ and $f_{\Gamma}$ is the fiber of $F_{\Gamma}$ (which can be identified with the fiber of the map on $\M_{v_1}$ that forgets the marked point $x_{h_1}$ associated to $h_1$).  Note that $f_{\Gamma}$ also carries a canonical orientation, as a subset of the boundary of a disk.

The main theorem of this section is the following:

\begin{thm}\label{thm:or}
The morphism $\mathcal{W} \rightarrow \M_{0,B ,\{a_i\}_{i \in I}}^{1/r}$ is canonically relatively oriented.  More precisely, for any sets $B$ and $I$ of boundary and internal markings, and any set $\{a_i\}_{i \in I}$ of the internal twists, there exists a distinguished orientation $o_{0, B, \{a_i\}_{i \in I}}$ of $T\M_{0,B ,\{a_i\}_{i \in I}}^{1/r} \oplus \mathcal{W}$.  These orientations satisfy the following properties:

\begin{enumerate}[(i)]

\item\label{it:renaming} If $f^B: B \rightarrow B'$ and $f^I: I \rightarrow I'$ are bijections such that $f^I$ preserves twists, and if
\[F: \M_{0,B, \{a_i\}_{i \in I}}^{1/r} \rightarrow \M_{0,B', \{a_i\}_{i \in I'}}^{1/r}\] is the induced map,
then $o_{0, B, \{a_i\}_{i \in I}} = F^* o_{0, B', \{a_i\}_{i \in I'}}$.

\item\label{it:or-induced to bdry} If $I = I_1 \sqcup I_2$ and $B = B_1 \sqcup B_2$ satisfy \eqref{eq:B12} and $\Gamma = \Gamma_{I_1, I_2; B_1, B_2}$, then under the isomorphisms \eqref{eq:Wdecomp} and \eqref{eq:TMdecomp}, we have
\[i_{\Gamma}^*o_{0,B,\{a_i\}_{i \in I}}= o_N \otimes o_{h_1} \otimes F_\Gamma^*(o_{0, B_1, I_1} \boxtimes o_{\{h_2\} \cup B_2, I_2}),\]
in which $o_N,~o_{h_1}$ are the canonical orientations on $N$ and $f_{\Gamma}$ respectively.
\item\label{it:or_ind_to_interior} If $\Gamma$ is a graph with one open vertex $v^o$ and one closed vertex $v^c$, then under the identification
\[i_{\Gamma}^*\det(T\M^{1/r}_{0,B, \{a_i\}_{i \in I}}) = \det(N) \otimes \left(\det(T\M_{v^o}) \boxtimes \det(T \M_{v^c})\right),\]
we have
\begin{equation}\label{eq:123456}i_{\Gamma}^*o_{0,B,\{a_i\}_{i \in I}}= o_N \otimes \left(o_{v^o} \boxtimes o_{v^c}\right),\end{equation}
where $o_N$ is canonically defined as above and $o_{v^o},o_{v^c}$ are the orientations on the total space associated to the two vertices, $o_{v^c}$ is the standard complex one.
\end{enumerate}
\end{thm}
The proof of this theorem is the content of remainder of the section.

\begin{rmk}\label{rmk:orientation and exact sequence}
Note that in the last item we have used that if $0\to A\xrightarrow{f} B\to C\to 0$ is an exact sequence of vector bundles, then there is a canonical isomorphism
\[\det(A)\otimes \det(C)\cong \det(B).\]
In case $e$ is Neveu--Schwarz, the Witten bundle decomposes as a direct sum by Proposition \ref{pr:decomposition}, but when $e$ is Ramond, we use \eqref{eq:decompses2} in order to write \eqref{eq:123456}.  In both cases, the decomposition may not respect the isotropy group actions, but these actions are orientation-preserving as they are induced from scaling the closed Witten bundle summand by roots of unity.
\end{rmk}

Theorem \ref{thm:or} defines a canonical relative orientation $o_\Gamma$ on $\mathcal{W} \rightarrow \M_{\Gamma}^{1/r}$ whenever $\Gamma$ is a connected smooth graded graph.  We extend the definition of relative orientation to the case where $\Gamma$ is a smooth graded graph, but not necessarily connected, by putting
\[o_\Gamma = \boxtimes_{\Lambda\in\Conn(\Gamma)}o_\Lambda,\]
in which we use that the Witten bundle $\cW$ on $\M_{\Gamma}^{1/r}$ is canonically isomorphic to the direct sum of the Witten bundles on the moduli spaces $\M_{\Lambda}^{1/r}$ associated to $\Lambda\in \Conn(\Gamma)$.

\begin{rmk}\label{rmk:uniqueness1}
Theorem \ref{thm:or} determines a unique family of orientations, up to the following ambiguity: if $o_{0,B,\{a_i\}_{i\in I}}$ is a family satisfying the conditions of the theorem, then $(-1)^{|B|-1}o_{0,B,\{a_i\}_{i\in I}}$ also satisfies these conditions, see Lemma \ref{lem:almost_uniq}. Remark \ref{rmk:uniqueness} also explains how we choose the distinguished orientation $o_{0,B,\{a_i\}_{i\in I}}$.
\end{rmk}

Throughout this section, we write $\cW_{0,k,\vec{a}}$ for the Witten bundle on $\M_{0,k,\vec{a}}^{1/r}$ for any tuple $\vec{a} = \{a_1, \ldots, a_l\}$.  Recall the action of $S_B$ on $\M_{0,B,\vec{a}}^{1/r}$, obtained by permuting markings. This action lifts canonically to $\cW_{0,k,\vec{a}}$.

\subsection{Extending the internal twists}

Thus far, we have assumed that in the open $r$-spin moduli space $\oCMr$, the internal twists satisfy $a_i \in \{-1,0,1,\ldots, r-1\}$.    However, when defining orientations, it is useful to allow more general twists.  In fact, there exists a smooth orientable orbifold with corners $\M_{0,k,\vec{a}}^{1/r}$ for {\it any} tuple of non-negative integers $\vec{a} = \{a_1, \ldots, a_l\}$, parameterizing stable nodal marked orbifold Riemann surfaces with boundary together with an orbifold line bundle $S$, an isomorphism
\[|S|^{\otimes r} \cong \omega_{|C|} \otimes \O\left(-\sum_{i\in I} a_i[z_i] - \sum_{i \in I} a_i [\overline{z}_i] - \sum_{j \in B} (r-2)[x_j]\right),\]
and a grading.  Moreover, there is a Witten bundle on this moduli space, defined as before.

The relationship between the Witten bundle with twist $a_i$ and twist $a_i+r$ was observed by Jarvis--Kimura--Vaintrob \cite{JKV} in the closed case, and the same is true in the open theory:

\begin{lemma}\label{lem:Witten_for_twists_more_than_r}
Let $\vec{a} = \{a_1, \ldots, a_l\}$ be a tuple of non-negative integers, and let $\vec{a}' = \{a_1+r,a_2, \ldots, a_l\}$.  Then there is an isomorphism
\[\kappa: \M_{0,k,\vec{a}'}^{1/r} \rightarrow \M_{0,k,\vec{a}}^{1/r},\]
and the Witten bundles on $\M_{0,k,\vec{a}'}^{1/r}$ and $\M_{0,k,\vec{a}}^{1/r}$ are related by a short exact sequence
\[0 \rightarrow \widetilde{\mathbb{L}} \rightarrow \mathcal{W}_{0,k,\vec{a}'} \rightarrow \kappa^*\mathcal{W}_{0,k,\vec{a}} \rightarrow 0,\]
in which $\widetilde{\mathbb{L}}$ is an orbifold line bundle satisfying
\[\widetilde{\mathbb{L}}^{\otimes r} \cong \left(\mathbb{L}_1^{\vee}\right)^{\otimes (a_1+1)}.\]
\begin{proof}
The isomorphism $\kappa$ is given by sending the twisted $r$-spin structure $S'$ to
\[S:=S'\otimes \O\left(r[z_1]+r[\overline{z}_1]\right),\]
where $[z_1]$ and $[\overline{z}_1]$ denote the orbifold divisors of the first internal marked point and its conjugate, so that
\[|S|=|S'|\otimes \O\left([z_1]+[\overline{z}_1]\right)\]
on $|C|$.  If $\mathcal{S}'$ denotes the universal twisted $r$-spin structure on the universal curve $\mathcal{C}$ over $\M^{1/r}_{0,k,\vec{a}'}$ and $\mathcal{S}:= \mathcal{S}'\otimes \O\left(\Delta_{z_1} + \Delta_{\overline{z}_1}\right)$ for the divisors $\Delta_{z_1}, \Delta_{\overline{z}_1} \subset \mathcal{C}$ corresponding to $z_1$ and $\overline{z}_1$, then the short exact sequence
\[0 \rightarrow |\mathcal{S}'| \rightarrow |\mathcal{S}| \rightarrow |\mathcal{S}|\big|_{\Delta_{z_1} + \Delta_{\overline{z}_1}} \rightarrow 0\]
implies
\[0 \rightarrow R^0\pi_*\left(|\mathcal{S}|\big|_{\Delta_{z_1} + \Delta_{\overline{z}_1}}\right) \rightarrow R^1\pi_*|\mathcal{S}'| \rightarrow R^1\pi_*|\mathcal{S}| \rightarrow 0,\]
Taking $\widetilde{\phi}$-invariants part yields
\[0 \rightarrow \sigma_1^*|\mathcal{S}| \rightarrow \left(R^1\pi_*|\mathcal{S}'|\right)_+ \rightarrow \left(R^1\pi_*|\mathcal{S}|\right)_+ \rightarrow 0,\]
where $\sigma_1$ is the section of $|\mathcal{C}|$ associated to the first internal marked point.  The first bundle in this sequence satisfies
\[(\sigma_1^*|\mathcal{S}|)^{\otimes r} \cong \sigma_1^*\left(\omega_{|\pi|,\log} \left(-\sum_{i=1}^l (a_i+1)\Delta_{z_i}\right)\right) \cong \mathbb{L}_1^{\otimes (a_1+1)},\]
since $\sigma_1^*\omega_{|\pi|,\log}$ is trivial.  Dualizing thus proves the claim.
\end{proof}
\end{lemma}

\subsection{Orientability of the Witten bundle}

Before we show how to choose canonical orientations, we first confirm that the Witten bundles are orientable.

\begin{lemma}\label{lem:Witten_is_orientable}
The bundles $\cW_{0,k,\{a_1\ldots,a_l\}}$ are orientable.

Moreover, suppose $g\in S_B$ preserves the component $\oCM_{0,B,\vec{a}}^{1/r,\hat{\pi}}$, meaning that there exists an integer $h$ such that
\[g(\pi(i)) = \pi(i+h),\]
interpreted cyclically.  Then, for any orientation $\mathfrak{o}$ of $\cW_{0,k,\{a_1\ldots,a_l\}}\to\oCM_{0,B,\vec{a}}^{1/r,\hat{\pi}}$, the permutation $g$ acts with sign $\varepsilon^h,$ where $\varepsilon = \varepsilon_k=(-1)^{k-1}.$
The action of $g\in S_I$ by permuting the internal markings, which also lifts canonically to $\cW_{0,k,\{a_1\ldots,a_l\}}$, is orientation-preserving.
\end{lemma}

\begin{proof}
Suppose, first, that $l \geq 1$ and $k-1\geq \sum a_i$.  Observe that
\[k-1\equiv \sum a_i\!\mod r.\]
There exists a (usually non-unique) tuple $(\tilde{a}_1,\ldots,\tilde{a}_l)$ with
\[\tilde{a}_i\geq 0, \;\;\; \tilde{a}_i\equiv a_i \!\mod r, \;\;\; \sum \tilde{a}_i = k-1.\]
By Lemma~\ref{lem:Witten_for_twists_more_than_r}, the bundles $\cW:=\cW_{0,k,\{a_1\ldots,a_l\}}$ and $\widetilde{\cW}:= \cW_{0,k,\{\tilde{a}_1\ldots,\tilde{a}_l\}}$ differ by a bundle $E$ that is a direct sums of complex line bundles. Thus, one has
\[\det(E)\otimes \det(\cW) \cong \det(\widetilde{\cW}),\]
so, since $E$ is canonically oriented, orienting $\cW$ is equivalent to orienting $\widetilde{\cW}.$ Note that $S_B$ and $S_I$ act with a trivial sign on $E$.  We can thus reduce to the case where $k-1= \sum a_i.$

In this case, the Witten bundle is of real rank $e = \sum a_i=k-1.$  Since the boundary of the base space does not affect the orientability of a bundle, $\cW\to\oCMr$ is orientable precisely if $\cW\to\oCMr\setminus\partial\CMr$ is orientable.  Furthermore, this is the case exactly if $\cW$ is orientable on the moduli space $\CMr$ of smooth graded $r$-spin disks, since compact strata of real codimension two do not affect orientability.  It is therefore enough to prove that the bundle $\cW$ is orientable over each connected component $\CM_{0,k,\{a_1,\ldots,a_l\}}^{\frac{1}{r},\hat{\pi}}$ of $\CMr$ associated to an order $\pi: [k] \to[k]$ with induced cyclic order $\hat{\pi}$.

The following gives explicit sections of $\cW\to\CM_{0,k,\{a_1,\ldots,a_l\}}^{\frac{1}{r},\hat{\pi}}$ that form a basis in any fiber.

\begin{construction}\label{cons:explicit basis of Witten}
Let $\Sigma$ be a smooth graded $r$-spin disk.  Identifying $\Sigma$ with the upper half-plane, let
\[\xi_{ij}=\xi_{ij}^\pi = \frac{(x_{\pi(j)}-x_{\pi(i)})dw}{(w-x_{\pi(i)})(w-x_{\pi(j)})},~i,j\in[k], \;\;\;\;\; \xi_j = \frac{\sqrt{-1}(\bar{z_j}-z_j)dw}{(w-z_j)(w-\bar{z}_j)},~j\in[l],\]
where $\sqrt{-1}$ is the root in the upper half plane.
Define sections
\[(\sigma_j)_\Sigma = (\sigma^\pi_j)_\Sigma=(\sigma^{\pi}_{0,k,\{a_1,\ldots,a_l\};j})_\Sigma= \left((-1)^{r+1}\prod_{i\in [l]}\xi_i^{a_i}\prod_{i\in[k]} \xi_{i(i+1)}^{-1}\xi_{1(1+j)}^r\right)^{\frac{1}{r}}\in H^{0}(\Sigma,|J|),\]
for $j \in [k-1]$, where $k+1$ is taken to be $1$.  More precisely, inside the parentheses in the above formula, a global section of
\[\omega_{|C|}^{\otimes r-1}\otimes \O\left(\sum_{i=1}^l a_i[z_i] + \sum_{i=1}^l a_i[\overline{z}_i] + \sum_{j=1}^k (r-2)[x_j]\right)\cong |J|^{\otimes r}\] is written. This section is $\phi$-invariant and positive on the arc from $x_1$ to $x_{1+j}$, with respect to the canonical orientation. One can construct an $r$th root in the sense of a global section $\sigma_j$ of $|J|$ whose $r$th power is mapped to the expression in parentheses. Indeed, this can be done locally away from special points, and the order of zeroes and poles at special points guarantee that the construction extends globally and univalently.  When $r$ is odd, the real $r$th root is unique.  When $r$ is even, there are two real $r$th roots, one of which agrees with the grading on the arc from $x_1$ to $x_{1+j}$ and the other of which agrees on the complementary arc; we choose $\sigma_j$ to be the one that agrees on the arc from $x_1$ to $x_{1+j}$.
\end{construction}

We claim that, for any $\Sigma\in\CM_{0,k,\{a_1,\ldots,a_l\}}^{\frac{1}{r}\hat{\pi}}$, the sections $(\sigma^\pi_j)_{\Sigma}$ for $j\in[k-1]$ form a basis of the fiber of $\cW$ over $\Sigma$.  To see this, observe first that the forms $\xi_j$ and $\xi_{ij}$ are $\text{PSL}_2(\R)$-invariant and conjugation-invariant.  It is immediate that for all $j,~\sigma_j$ is nowhere-vanishing on $\CM_{0,k,\{a_1,\ldots,a_l\}}^{\frac{1}{r},\hat{\pi}}.$ As the number of those sections is $k-1$, it is enough to show that they are linearly independent. But this is clear, since $(\sigma_{j})_{\Sigma}$ has poles only at $x_{\pi(1)}$ and $x_{\pi(1+j)}$, and by calculating coefficients of poles (for example at $x_{\pi(i)}$ for $i\in\{2,\ldots,k\}$), we see that if
\[\sum c_j(\sigma_j)_{\Sigma} = 0,\]
then $c_j=0$ for all $j\in[k-1]$. Thus, $\cW\to\CM_{0,k,\{a_1,\ldots,a_l\}}^{\frac{1}{r},\hat{\pi}}$ is orientable.

The next case is $l =0$.  Recall from equation \eqref{eq:forgetful_map} that $\cW_{0,k,\{0\}}$ is the pullback of $\cW_{0,k,\emptyset}$ by the forgetful map $\oCM_{0,k,\{0\}}\to\oCM_{0,k,\emptyset}.$ The fiber of this forgetful map is an open disk, and in particular, it is contractible. Thus, $\cW_{0,k,\emptyset}$ is orientable exactly if $\cW_{0,k,\{0\}}$ is orientable, and the latter is orientable by the first case.

The last case is when $l\geq 1$ and $k-1< \sum a_i.$ Let $k'= 1+\sum a_i$, so that $\cW_{0,k',\{a_1,\ldots,a_l\}}$ is orientable by the first case. Consider the boundary stratum $\oCM^{1/r}_\Gamma$ defined by a graph $\Gamma$ with two open vertices $v_1$ and $v_2$, in which $v_1$ has $k'-k$ boundary tails and no internal tails, while $v_2$ has the rest. Note that the half-edge $h_1$ of $v_1$ has $\tw(h_1)=\alt(h_1)=0.$ By Proposition \ref{pr:decomposition} and equation \eqref{eq:forgetful_map}, we have an identification
\[i_{\Gamma}^*\cW_{0,k',\{a_1,\ldots,a_l\}} \cong \cW_{0,k'-k+1,\emptyset} \boxplus \cW_{0,k,\{a_1,\ldots,a_l\}}.\]
Hence,
\[\det(\cW_{0,k',\{a_1,\ldots,a_l\}})|_{\oCM_\Gamma}\simeq  \det(\cW_{0,k'-k+1,\emptyset}) \boxtimes \det(\cW_{0,k,\{a_1,\ldots,a_l\}}).\]
As $\cW_{0,k',\{a_1,\ldots,a_l\}}$ and $\cW_{0,k'-k+1,\emptyset}$ are orientable by the first two steps, $\cW_{0,k,\{a_1,\ldots,a_l\}}$ is also orientable.

Regarding the signs of actions of $S_B$ and $S_I$, consider first the case where $l=1$.  In this situation, $k\equiv 1+a_1\mod r$.  As noted above, the bundles $\cW_{0,k,a_1}$ and $\cW_{0,k,\{k-1\}}$ differ by a direct sum of complex line bundles by Lemma \ref{lem:Witten_for_twists_more_than_r}, and this decomposition is invariant under $S_B$.  The claim in this case follows, then, if we prove it for the bundle $\cW_{0,k,\{k-1\}}.$  Moreover, it suffices to prove the Lemma for $h=-1$, so we assume this in what follows.

Observe that, since the bundle is orientable and all components of the moduli space are isomorphic, the claim can be verified for any specific component by fixing any orientation $\mathfrak{o}$ for $\cW$ over this component and comparing $\mathfrak{o}_{g\cdot\Sigma}$ to $g^*\mathfrak{o}_{\Sigma}$, where $g\cdot\Sigma$ is obtained from $\Sigma$ by renaming markings according to $g$ and $\Sigma$ is an arbitrary point in the specified component.

Using these comments, fix a generic smooth $\Sigma$ and write $\Sigma'=g\cdot\Sigma$.  Let $\mathfrak{o}^\pi$ be the orientation determined by the ordered basis $(\sigma^\pi_1,\ldots,\sigma^\pi_{k-1})$.  Comparing $\mathfrak{o}^\pi_{g\cdot\Sigma}$ with $g^*\mathfrak{o}^\pi_{\Sigma}$ is equivalent to calculating the ratio of $\mathfrak{o}^\pi$ to $\mathfrak{o}^{g\cdot\pi}$ at the same point of the moduli space.  To this end, first observe that the residues of $(\sigma^\pi_j)_{\Sigma}$ at $x_{\pi(1)}$ and $x_{\pi(1+j)}$ are real and of opposite sign. The sign of the residue at $x_{\pi(1)}$ is independent of $j,$ and moreover, it equals the sign of the residue of $(\sigma^{\pi'}_j)_{\Sigma}$ at $x_{\pi'(1)}$ for $\pi'= g\cdot\pi.$
Thus, the coordinate change between the ordered bases $(e_1,\ldots,e_{k-1})=(\sigma^\pi_1,\ldots,\sigma^\pi_{k-1})_{g\cdot\Sigma},~(e'_1,\ldots,e'_{k-1})=(\sigma^{\pi'}_1,\ldots,\sigma^{\pi'}_{k-1})_{g\cdot\Sigma}$ is given by
 \begin{gather*}
 e'_1 = -\alpha_1e_1+\beta_1e_2,\ldots, e'_{k-2} = -\alpha_{k-2}e_1+\beta_{k-2}e_{k-1}, ~
 e'_{k-1} = -\alpha_{k-1}e_1,
 \end{gather*}
 where $\alpha_i,\beta_i$ are positive numbers depending on $\Sigma.$ Hence, the induced sign on the orientation is equal to the sign of $\prod(-\alpha_i)$, which is $(-1)^{k-1}.$

The case $l=0$ follows by the same argument as in the proof of orientability above: we first add an internal point with twist zero, and then we reduce to the case $l=1.$

Finally, for $l>1,$ we compare the two orientations at a generic point of $\CM_\Gamma,$ where $\Gamma$ is a graph with two vertices, a closed vertex $v_1$ with all internal tails and an open vertex $v_2$ with only boundary tails. By Proposition \ref{pr:decomposition}, we have $\mu^*\cW|_{\CM_{\Gamma}}\cong q^*(\cW_{v_1}\boxplus\cW_{v_2}),$ and the decomposition is invariant under the actions of $S_B$ and $S_I$.  Now, $\mu^*$ and $q^*$ are complex maps and hence preserve signs. Furthermore, $\cW_{v_1}$ is canonically oriented as a complex bundle, $g\in S_B$ acts on it trivially, and $g\in S_I$ (which preserves twists) acts on it as a complex isomorphism, hence preserves sign.  Thus, the sign induced by $g$ on the orientation of $\cW_{0,k,\{a_1,\ldots,a_l\}}$ (or, equivalently, on the orientation of $\cW|_{\CM_{\Gamma}}$) is the sign it induces on the orientation of $\cW_{v_2}$.  It is $1$ for $g \in S_I$, and it equals $(-1)^{k-1}$ for $g\in S_B$ by the first case.
\end{proof}

\subsection{Choosing canonical orientations}

Throughout what follows, we continue with the notation from the proof of Lemma \ref{lem:Witten_is_orientable}.

\begin{definition}\label{def:covariance_bundle_or}
Fix a set of twists $\vec{a}$ of size $l$ and an integer $k\geq 0.$
Let $\{\mathfrak{o}^\pi\}$ be a family of orientations, one for each $\cW\to\oCM_{0,B,\vec{a}}^{\hat{\pi},1/r}$, in which $\pi$ runs over all orderings of $B$ and $B$ runs over all sets of size $k$.  We say that the family $\{\mathfrak{o}^\pi\}$  is \emph{covariant} if, whenever $f^B:B\to B'$ and $f^I:I\to I'$ are bijections such that $f^I$ preserves twists, we have $\mathfrak{o}_{0,B,\vec{a}}^\pi=F^* \mathfrak{o}_{0,B',\vec{a}}^{f^B\circ \pi},$ where $F:(\cW\to\oCM_{0,B,\{a_i\}_{i\in I}}^{1/r})\to(\cW\to\oCM_{0,B',\{a_i\}_{i\in I'}}^{1/r})$ is the induced map.
\end{definition}

\begin{rmk}\label{it:1_for_covariance}
By the second part of Lemma \ref{lem:Witten_is_orientable}, whenever $\{\mathfrak{o}^{\pi}\}$ is a covariant family and $g \in S_B$ cyclically satisfies $g(\pi(i))=\pi(i+h)$ for some integer $h$, we have the equality
$\mathfrak{o}^\pi=\varepsilon^h\mathfrak{o}^{g\cdot\pi}$, where $\varepsilon = \varepsilon_k=(-1)^{k-1}.$  Thus, associated to any multiset $(\vec{a},k)$ are exactly two covariant orientations, determined by choosing $ \mathfrak{o}^\pi$ for a single $\pi$ to be any orientation of the Witten bundle of $\oCM_{0,k,\vec{a}}^{1/r,\hat\pi}$, and then extending covariantly.
\end{rmk}

\begin{definition}\label{def:or_in_the_right_number}
Suppose $a_i\leq r-1$ for each $i$ and $k-1=\sum_{i\in[l]} a_i.$ When $k>1$, for an order $\pi$ of the boundary marked points, denote by $\mathfrak{o}^\pi_{0,k,\vec{a}}$ the orientation of $\cW_{0,k,\vec{a}}^{1/r,\hat{\pi}}$ defined using the ordered basis $\mathbf{\sigma}^\pi = (\sigma_1^\pi,\ldots,\sigma^\pi_{k-1})$ on the smooth locus. For any $I$ of size $l$ and $B$ of size $k=1+\sum_{i\in I} a_i$, this uniquely defines a covariant family of orientations $\{\mathfrak{o}^\pi_{0,B,\{a_i\}_{i\in I}}\}_\pi.$ When $k=1$ and hence all $a_i=0$, the Witten bundle is zero-dimensional; in this case, define its orientations $\{\mathfrak{o}_{0,\{b\},\{0\}_{i\in I}}\}$ to be the positive orientations.
\end{definition}

In case $l=0$, $\vec{a}$ is the empty vector, and $k = r+1$, Remark \eqref{it:1_for_covariance} implies that there are exactly two covariant families, induced by the two choices of orientation of the Witten bundle on the space $\oCM_{0,B,\emptyset}^{\frac{1}{r}}$ with $|B|=r+1$. Fix one of these orientations, to be chosen later, and denote it by $\mathfrak{o}_{r+1}.$

Now, let $l,k\geq 0$, and fix $a_1,\ldots,a_l\in\{0,\ldots,r-1\}$ with $\oCMr\neq\emptyset$.  Suppose we have a covariant family of orientations with twists $\{a_1, \ldots, a_l\}$ and $B$ of size $k$.  We now show how to induce, given $\mathfrak{o}_{r+1}$ as above, a covariant family of orientations with the same twists but for $B'$ of size $k':=k \pm r$, assuming $k'$ is non-negative.

Suppose, first, that $k'=k-r$.  Let $B'$ be a set of size $k-r$, and let $\pi'$ be an order of $B'$.  Let $B$ be a set of size $k$ containing $B'$, and let $\pi$ be an order of $B$ that extends $\pi'$, such that the first $k-r$ element of $B$ with respect to $\pi$ are those of $B'$.

Let $\Gamma\in\partial\Gammar$ be a graph with two open vertices and no closed vertices, for which there is one vertex $v_0$ that contains all of the internal tails, all of the boundary tails marked by $B'$, and one half-edge $h_0$. The other vertex, $v_1,$ contains the remaining $r$ boundary tails marked by $B\setminus B'$ and the half-edge $\sigma_1h_0.$   By construction, $\tw(h_0)=\alt(h_0)=0$.  Observe that in $\CM^{1/r,\hat\pi}_\Gamma$, the boundary marked points on the component corresponding to $v_0$ are cyclically ordered by $\hat{\pi'}$.  The Witten bundle $\cW_{\Gamma}$ on $\M_{\Gamma}^{1/r}$ can be identified with $\cW_{v_0}\boxplus\cW_{v_1}$, so we have an identification
\[\det(\cW_{\Gamma}) = \det(\cW_{v_0})\boxtimes \det(\cW_{v_1}).\]  Define an orientation $\text{ind}_{\Gamma,\pi\to\pi'}\mathfrak{o}^\pi$ on $\cW\to\oCM^{\hat{\pi}',\frac{1}{r}}_{0,B',\{a_i\}_{i\in[l]}}$ as the unique orientation satisfying
\[\mathfrak{o}^\pi|_{\cW_{\oCM_\Gamma^{1/r,\hat{\pi}}}} =(\text{ind}_{\Gamma,\pi\to\pi'}\mathfrak{o}^{\pi})\boxtimes\mathfrak{o}^{\pi_1}_{r+1},\]
where $\pi_1$ is the order on the half-edges of $v_1$ induced from $\pi$, starting from $\sigma_1h_0.$

This procedure defines $\mathfrak{o}^{\pi'}$ uniquely from $\mathfrak{o}^\pi.$ The construction is easily seen to be independent of choices, and yields a covariant family. Moreover, inverting the steps allows us to define $\mathfrak{o}^\pi$ uniquely from $\mathfrak{o}^{\pi'},$ and if the latter comes from a covariant family, so will the former. Thus, the case $k'=k+r$ is also treated.

Therefore, given twists $\{a_i\}_{i\in I}$, we can uniquely define a covariant family of orientations $\mathfrak{o}^\pi$ for any set $B$ such that $\oCM_{0,B,\vec{a}}^{1/r}\neq \emptyset$ by inducing the orientations iteratively, starting from the covariant family of Definition \ref{def:or_in_the_right_number} (the case $k=\sum a_i+1$) or, when $I=\emptyset$, from $\mathfrak{o}_{r+1}$.

\begin{definition}\label{def:or_for_witten_based_on_o_r+1}
Given $\mathfrak{o}_{r+1}$ and the orientations of Definition \ref{def:or_in_the_right_number}, define $\{\mathfrak{o}^{\pi}\}$ as the unique covariant family of orientations induced by the above procedure from the family $\{\mathfrak{o}^\pi_{0,B,\{a_i\}_{i\in I}}\}_\pi$, defined when $k-1=\sum a_i.$
\end{definition}
%
\begin{obs}\label{obs:inducing}
If $I\subset I'$ and $a_i=0$ for all $i\in I'$ then $\mathfrak{o}_{0,B,\{a_i\}_{i\in I'}}^\pi=\text{For}_{I'\setminus I}^*\mathfrak{o}_{0,B,\{a_i\}_{i\in I'}}^\pi.$
\end{obs}
Indeed, the case $|B|-1=\sum a_i$ follows from the definitions, while the general case follows easily by noting that the forgetful map behaves well with respect to the induction procedure.

\subsection{Properties of the orientations}

The family of orientations in Definition \ref{def:or_for_witten_based_on_o_r+1} respects the decomposition properties of the Witten bundle.

To make this precise, we first state two lemmas, whose proofs are postponed to the next subsection. We first fix decompositions $I = I_1\sqcup I_2$, $B = B_1 \sqcup B_2$, and $k=k_1+k_2$ such that
\[k_1-1\equiv\sum_{i\in I_1} a_i \mod r,\]
\[k_2 \equiv \sum_{i\in I_2} a_i \mod r.\]
Let $\Gamma$ be a graph with two open vertices, $v_1$ and $v_2$, connected by an edge $e$, in which the vertex $v_i$ has internal tails labeled by $I_i$ and $k_i$ boundary tails labeled by $B_i$.  If $h_i$ are the half-edges of $v_i$, then a simple calculation shows that
\[\text{tw}(h_1) = \alt(h_1) = 0.\]
Let $\pi$ be an order in which the elements of $B_1$ are consecutive and come before the elements of $B_2.$  Consider $\Sigma\in \CM_\Gamma^{\hat\pi}$, where the normalization of $\Sigma$ has components $\Sigma_1$ and $\Sigma_2$ corresponding to $v_1$ and $v_2$, respectively.
Let $\pi_1$ be the restriction of $\pi$ to the points of $\Sigma_1$, and let $\pi_2$ be the restriction to points of $\Sigma_2$ but with the half-node $x_{h_2}$ added as the first point.

On $\CM_\Gamma^{\hat\pi}$, the bundle $\cW_{\Gamma}$ is again identified with $ \cW_{v_1}\boxplus\cW_{v_2}$, so we have an identification,
\begin{equation}
\label{eq:detW}
\det(\cW_{\Gamma}) = \det(\cW_{v_1}) \boxtimes \det(\cW_{v_2}),
\end{equation}
and this is respected by the orientations in Definition \ref{def:or_for_witten_based_on_o_r+1}.  That is:

\begin{lemma}\label{lem:induced or to boundary bundle}
For $\Sigma$ as above, the orientation $\mathfrak{o}^\pi$ agrees with
$\mathfrak{o}^{\pi_1}_{0,B_1,I_1}\boxtimes\mathfrak{o}^{\pi_2}_{0,\{h_2\}\cup B_2,I_2}$ under the isomorphism \eqref{eq:detW}.
\end{lemma}

Similarly, the family of orientations in Definition \ref{def:or_for_witten_based_on_o_r+1} satisfies a decomposition property along internal nodes.  To specify this, write
\[m^c(\{a_i\}_{i\in I}) = \frac{\sum a_i -(r-2)+ ((r-2-\sum a_i) \; (\text{mod } r)')}{r},\]
where $x \;(\text{mod }r)'$ is the unique element of $\{-1,0,\ldots,r-2\}$ congruent to $x$ modulo $r$.

Now, let $\Gamma\in\partial\Gammar$ be a graph with two vertices, an open vertex $v^o$ and a closed vertex $v^c$.  By Proposition \ref{pr:decomposition} and Remark \ref{rmk:orientation and exact sequence}, the Witten bundle $\cW_{\Gamma}$ on $\M_{\Gamma}^{1/r}$ satisfies
\begin{equation}\label{eq:det_and_exact_seq_Witten}
\det(\cW_{\Gamma})\cong (q \circ \mu^{-1})^* \bigg( \det(\cW_{v^o})\boxtimes \det(\cW_{v^c}) \bigg),
\end{equation}
and we have the following:

\begin{lemma}\label{lem:induced_or_of_Witten-internal}
There exists $\delta=\pm 1$, depending only on the choice of $\mathfrak{o}_{r+1}$, such that for any $\Gamma,v^o$, and $v^c$ as above and any order $\pi$, the orientation $\mathfrak{o}^\pi|_{\CM_\Gamma^{\hat\pi}}$ agrees with $\delta^{m^c(\{a_i\}_{i\in I^C})}\mathfrak{o}^\pi_{v^o}\boxtimes\mathfrak{o}_{v^c}$ under the isomorphism \eqref{eq:det_and_exact_seq_Witten}, where $\mathfrak{o}_{v^c}$ is the canonical complex orientation and $I^C$ are the labels of $v^c$.  Moreover, changing $\mathfrak{o}_{r+1}$ to the opposite orientation changes $\delta$ to $-\delta.$
\end{lemma}

\begin{definition}\label{def:or_Witten}
Define $\mathfrak{o}_{r+1}$ to be the unique covariant family of orientations of $\cW\to\oCM_{0,r+1,\emptyset}^{\frac{1}{r}}$ for which the $\delta=1$ in Lemma \ref{lem:induced_or_of_Witten-internal}.  This induces orientations $\mathfrak{o}_{0,B,\vec{a}}^\pi$ on the Witten bundle for all $B$ and $\vec{a}$ by Definition \ref{def:or_for_witten_based_on_o_r+1}.
\end{definition}
We can now complete the proof of Theorem \ref{thm:or}.

\begin{proof}[Proof of Theorem \ref{thm:or}]
Let $\tilde{\mathfrak{o}}^\pi_{0,B,\vec{a}}$ be the orientations on the moduli spaces $\M_{0,B,\vec{a}}^{\frac{1}{r}}$, described explicitly in Notation \ref{nn:or_for_moduli_spin}, and let $\mathfrak{o}_{0,B,\vec{a}}^\pi$ be the orientations on the Witten bundle described in Definition~\ref{def:or_Witten}.  By Lemma \ref{lem:induced or to boundary bundle} and Proposition \ref{prop:or_moduli}, the relative orientation
\begin{equation}\label{eq:def_of_or}o_{0,B,\{a_i\}_{i\in I}} = \tilde{\mathfrak{o}}_{0,B,\{a_i\}_{i\in I}}^\pi\otimes\mathfrak{o}_{0,B,\{a_i\}_{i\in I}}^\pi\end{equation}
for the Witten bundle on $\oCM_{0,B,\{a_i\}_{i\in I}}^{\frac{1}{r},\hat\pi}$ is independent of the choice of $\pi$.  The same argument, together with Remark \ref{it:1_for_covariance}, shows the invariance property.

For the second item, first note that, under the notation of Lemma \ref{lem:induced or to boundary bundle} and using item \eqref{it:perm} of Proposition \ref{prop:or_moduli}, we may write
\[\mathfrak{o}^{\pi'_1}_{v_1}=o_{h_1}\boxtimes\mathfrak{o}^{\pi_1}_{v'_1},\]
where $\pi'_1$ is the extension of $\pi_1$ defined by writing $h_1$ as the last element, and the pullback of $\mathfrak{o}^{\pi_1}_{v'_1}$ is with respect to the map that forgets the half-node $n_{h_1}$.  Note also that
\[\text{dim}_\R\oCM_{v_1'}^{1/r}\equiv |B_1|-1\mod 2,\;\;\;\;\;\;\;\text{rank}_\R\cW_{v_2}\equiv |B_2| \mod 2.\]
From here, the second item is a consequence of Lemma \ref{lem:induced or to boundary bundle} and Lemma \ref{lem:moduli-induced_or}, where the sign $(-1)^{(|B_1|-1)|B_2|}$ disappears when commuting $\tilde{\mathfrak{o}}^{\pi_2}_{v_2}$ with $\mathfrak{o}^{\pi_1}_{v'_1}$.

The last item is a direct consequence of Lemmas \ref{lem:induced_or_of_Witten-internal} and \ref{lem:moduli-induced_or}.
\end{proof}
\begin{lemma}\label{lem:almost_uniq}
The properties of Theorem \ref{thm:or} characterize precisely two families of orientations: the family $\{o_{0,B,\{a_i\}_{i\in I}}\}_{B,\{a_i\}_{i\in I}}$ and the family $\{(-1)^{|B|-1}o_{0,B,\{a_i\}_{i\in I}}\}_{B,\{a_i\}_{i\in I}}$.
\end{lemma}
\begin{proof}
Suppose that $\{o'_{0,B,\{a_i\}_{i\in I}}\}_{B,\{a_i\}_{i\in I}}$ is a different family of orientations satisfying the requirements of Theorem \ref{thm:or}. Let $\delta_{B,\vec{a_i}}\in\{\pm 1\}$ be the ratio of $o'_{0,B,\{a_i\}_{i\in I}}$ to $o_{0,B,\{a_i\}_{i\in I}}$.
Then item \eqref{it:renaming} shows that $\delta_{B,\vec{a_i}}=\delta_{|B|,\vec{a_i}}$.
Item \eqref{it:or_ind_to_interior} shows that $\delta_{|B|,\vec{a_i}} = \delta_{|B|,\{\sum a_i \mod r\}}$.  Since $\sum a_i=|B|-1 \mod r$ ,we denote $\delta_{|B|,\{\sum a_i \mod r\}}$ by $\delta_{|B|}$.  Finally, item \eqref{it:or-induced to bdry} shows that $\delta_{a+b}=\delta_{a+1}\delta_{b}.$ Thus, $\delta_1=0$ and $\delta_k=\delta_2^{k-1}$, where $\delta_2\in\pm1.$ The claim follows.
\end{proof}

\begin{rmk}\label{rmk:uniqueness}
This ambiguity from Lemma \ref{lem:almost_uniq} is killed by specifying the orientation of Witten's bundle for a single real one-dimensional moduli space: the additional requirement is that if we orient $\oCM_{0,2,\{1\}}^{\hat\pi}$ for $\pi=(1,2)$ by $\tilde{\mathfrak{o}}^\pi$, then the bundle $\cW\to\oCM_{0,2,\{1\}}^{\hat\pi}$ is oriented so that sections that are positive with respect to the grading on the arc from $x_1$ to $x_2$.
\end{rmk}

\subsection{Proof of Lemmas \ref{lem:induced or to boundary bundle} and \ref{lem:induced_or_of_Witten-internal}}\label{subsec:ap_some_more_or}

We now return to the proofs of the two lemmas from the previous subsection.

\begin{proof}[Proof of Lemma \ref{lem:induced or to boundary bundle}]
We first treat the case \[k_1-1=\sum_{i\in I_1}a_i,\qquad k_2=\sum_{i\in I_2}a_2.\]
For convenience, assume that $B=[k]$, that $B_1=[k_1]$, and that $\pi$ is the standard order.  We show that if $j < k_1$, then up to rescaling by a positive function, the $j$th basis element of $\mathbf{\sigma}^\pi$ converges to $\mathbf{\sigma}^{\pi_1}_{0,k_1,\{a_i\}_{i\in I_1};j}$ as one approaches a point in $\CM_\Gamma^{1/r,\hat\pi}$, and that if $j \geq k_1$, it converges to $\mathbf{\sigma}^{\pi_1}_{0,1+k_2,\{a_i\}_{i\in I_2};j-k_1+1}$.  This verification implies the lemma.

We assume that $k_1,k_2>0$, since otherwise, the result is straightforward.  Recall that a vector $u\in(\cW_{0,k,\vec{a}})_{\Sigma}$ for $\Sigma\in\CM_\Gamma^{1/r,\hat\pi}$ can be written as $u=u_1\boxplus u_2$ with $u_i\in(\cW_{v_i})_{\Sigma_i}$.  To calculate a coordinate expression for $u_i$, let $\{\Sigma_t\}_{t \in (0,1/2)}$ be a path in $\CMr$ such that $\lim_{t\to 0}\Sigma_t=\Sigma.$ One can model $(\Sigma_t)_{t>0}$ on the upper half-plane, preserving the complex orientation, such that all markings of $\Sigma_{3-i}$ tend to $0$ as $t\to 0$ but all markings of $\Sigma_i$ tend to finite, nonzero limits.  The resulting marked upper half-plane is a model for $\Sigma_i$ in which $x_{h_i}$ is mapped to the origin. If the vectors $u_t\in\cW_{\Sigma_t}$ converge to $u$, then their expressions in the coordinates induced from the upper half-plane model converge to the coordinate expression for $u_i$.  Moreover, as $t\to 0,$ the ratio between any two markings going to zero is bounded away from zero, since otherwise, $\Sigma\in\oCM_\Gamma^{1/r}\setminus\CM_\Gamma^{1/r}$.  

Suppose, first, that $j\geq k_1$, so that $1+j$ is a marked point of $\Sigma_2$.  Recall that
\[\sigma^\pi_j= \left((-1)^{r+1}\prod_{i\in [l]}\xi_i^{a_i}\prod_{i\in[k]}\xi_{i(i+1)}^{-1}\xi_{1(1+j)}^r\right)^{\frac{1}{r}}\]
and we have
\[\sigma^{\pi_2}_{1+j-k_1}= \left((-1)^{r+1}\prod_{i\in I_2}\xi_i^{a_i}\prod_{i\in[k_2+1]}(\xi^{\pi_2}_{i(i+1)})^{-1}(\xi^{\pi_2}_{1(2+j-k_1)})^r\right)^{\frac{1}{r}},\]
where $x_{\pi_2(1)}=x_{h_2}$.  Then, by the discussion above,
\[\left(\frac{\prod_{h\in[k_1-1]}(x_{\pi(h+1)}-x_{\pi(h)})}{\prod_{h\in I_1}(\sqrt{-1}(\bar{z}_h-z_h))^{a_h}}\right)^{\frac{1}{r}}(\sigma^\pi_j)_{\Sigma_t}\to\sigma^{\pi_2}_j.\]
The function in the above by which we multiply $(\sigma^\pi_j)_{\Sigma_t}$ is positive. Moreover, for $t$ close enough to zero, the ratios between any two factors $x_{\pi(h+1)}-x_{\pi(h)}$ and $\bar{z}_h-z_h$ are bounded, hence they tend to zero uniformly. In fact, since $\sum_{h\in I_1}a_h=k_1-1$, the expression \[\frac{\prod_{h\in[k_1-1]}(x_{\pi(h+1)}-x_{\pi(h)})}{\prod_{h\in I_1}(\sqrt{-1}(\bar{z}_h-z_h))^{a_h}}\] has a nonzero limit.  Calculating the same expression but in a gauge-fixing for which the points of $\Sigma_2$ approach zero as $t\to 0$, while those of $\Sigma_1$ have distinct finite images, the same argument shows that the limiting section, when projected to $\cW_{v_1}$, has zero of order $\frac{1}{r}.$

Suppose, now, that $j< k_1$, so that $1+j$ is a marked point of $\Sigma_1$.  Recall that
\[\sigma^{\pi_1}_j= \left((-1)^{r+1}\prod_{i\in I_1}\xi_i^{a_i}\prod_{i\in[k_1]}(\xi^{\pi_1}_{i(i+1)})^{-1}(\xi^{\pi_1}_{1(1+j)})^r\right)^{\frac{1}{r}}.\]
Write $x_{h_1}$ for the half-node of $\Sigma_1$.  Again, we work with the upper half-plane model for $\Sigma_1,$ where $x_{h_1}$ is mapped to the origin and the orientation of the image of $\partial\Sigma_1$ agrees with the standard real orientation of $\R$.  Consider again a path of smooth surfaces $\Sigma_t$, where $t\in(0,\frac{1}{2})$, that converges to $\Sigma$. Choose $i_j\in I_j$ and write
\[C_{i_1i_2}=-\frac{(\bar{z}_{i_1}-z_{i_1})(\bar{z}_{i_2}-z_{i_2})}{(z_{i_1}-z_{i_2})(\bar{z}_{i_1}-\bar{z}_{i_2})}.\]
This is a well-defined, positive function on $\CMr$, so it has a positive $r$th root; furthermore, it vanishes at $\CM_\Gamma$.  The same considerations as above reveal that the limit of $\frac{\sigma^\pi_j(\Sigma_t)}{C_{i_1i_2}^{1/r}}$, when projected on $\cW_{v_1}$, is nonzero, but its projection on $\cW_{v_2}$ vanishes to order to order $1-1/r$.
Moreover, the projection of the limit on $\cW_{v_1}$ has the zero profile of $\sigma_j^{\pi_1},$ hence they agree up to multiplication by a real function, by degree reasons and involution invariance.  This function is positive, since both sections are positive on the arc from $x_1$ to $x_{1+j}.$

We now turn to the general case. Write
\[m_1 =\frac{k_1-1-\sum_{i\in I_1} a_i}{r}, \qquad m_2 =\frac{k_2-\sum_{i\in I_2} a_i}{r}  .\]
The proof is by induction on $m=|m_1|+|m_2|$.  The case $m=0$ has been treated above. Note that, perhaps by applying Observation \ref{obs:inducing}, we may assume $I_1$ and $I_2$ are nonempty. Suppose that the claim has been proven for $m-1$, and assume $m_2<0$.

Let $B'_2\supseteq B_2$ be a set of size $k'_2=|B_2|+r,$ and write $B' = B_1\cup B'_2.$ Consider the graph $\Gamma'$ obtained by attaching to $\Gamma$ an open vertex, connected to $v_2$ by an edge whose half-edge $h$ in $v_2$ has $\tw(h)=\alt(h)=0$, and that has $r$ boundary tails labeled by $B'\setminus B$.  Let $\pi'$ be an order on $B'$ extending $\pi$, such that the last elements are those of $B\setminus B'$ and are ordered so that tails belonging to the same open vertex are labeled consecutively. Let $\Gamma_2$ be the component of $v_2$ in $\detach_{e}(\Gamma').$
By the construction of the induced orientation, we have
\begin{equation}\label{eq:1_for_induced_or}
\mathfrak{o}^{\pi'}|_{\oCM_{\Gamma'}^{\hat\pi'}}=\mathfrak{o}^\pi|_{\oCM_\Gamma^{\hat\pi}} \boxtimes \mathfrak{o}'_{r+1},\end{equation}
where $\mathfrak{o}'_{r+1}$ is the orientation for the Witten bundle of the new component with respect to the order induced from $\pi'$ after putting the node as the first element.  On the other hand, by the induction assumption, we have
\begin{equation}\label{eq:2_for_induced_or}\mathfrak{o}^{\pi'}|_{\oCM_{\Gamma'}^{\hat\pi'}} = \mathfrak{o}^{\pi_1}_{0,B_1,I_1}\boxtimes\mathfrak{o}^{\pi'_2}_{0,\{h_2\}\cup B'_2,I_2},\end{equation} where $\pi'_2$ is the restriction of $\pi_2$ to $B'_2.$
By the construction of the induced orientation, this time with respect to $\oCM_{0,\{h_2\}\cup B'_2,\{a_i\}_{i\in I_2}}^{1/r,\pi'_2},$ we have
\begin{equation}\label{eq:3_for_induced_or}\mathfrak{o}^{\pi'_2}|_{\oCM_{\Gamma_2}^{\hat\pi'_2}}=\mathfrak{o}^{\pi_2} \boxtimes\mathfrak{o}'_{r+1}.\end{equation}
Putting these observations together, we see that
\begin{equation}\label{eq:4_for_induced_or}\mathfrak{o}^\pi|_{\CM_\Gamma^{\hat\pi}}=\mathfrak{o}^{\pi_1}_{0,B_1,I_1}\boxtimes\mathfrak{o}^{\pi_2}_{0,\{h_2\}\cup B_2,I_2}.\end{equation} 

The case $m_2>0$ is treated similarly to the above, so we omit it.  The remaining case is $m_2=0$.  In this case, $m_1\neq 0$, and the proof is similar, so we merely remark on the changes.  First, we work with a graph $\Gamma'$ obtained from $\Gamma$ by attaching the new vertex with $r$ boundary tails, no internal tails and one legal half edge of twist $r-2$ to $v_1$.  We choose the order $\pi'$ so that boundary tails of $v_1$ comes first and those of $v_2$ come last. If we let $e_1$ be the edge of $v_2$ and $e_2$ the edge of the new vertex, then we compare the orientation expressions for $\cW\to\oCM_{\Gamma'}^{\hat{\pi'}}$ obtained in two ways. First, we induce orientation from $\cW\to\oCM_{d_{e_1,e_2}\Gamma'}^{\hat{\pi'}}$ to $\cW\to\oCM_{d_{e_1}\Gamma'}^{\hat{\pi'}}$ via $\text{Ind}$, and then to $\cW\to\oCM_{\Gamma'}^{\hat{\pi'}},$ to obtain (using the above notation)
\[\mathfrak{o}^{\pi'}|_{\oCM_{\Gamma'}^{\hat{\pi'}}}=\epsilon_1(-1)^{rk_2}\mathfrak{o}^{\pi_1}\boxtimes\mathfrak{o}^{\pi_2}\boxtimes\mathfrak{o}'_{r+1}.\]
Here, $\epsilon_1$ is the sign appearing from inducing the orientation from $\cW\to\oCM_{d_{e_1}\Gamma'}^{\hat{\pi'}}$ to $\cW\to\oCM_{\Gamma'}^{\hat{\pi'}}$ and $(-1)^{rk_2}$ comes from the fact that, in $\pi'$, the tails of the new vertex are not last, so we have to perform a shift in $\pi'$ in order to calculate the sign induced using $\text{Ind}$.  Using Lemma \ref{lem:Witten_is_orientable}, we see that the total added sign is $(-1)^{rk_2}.$  Next, we induce orientation from $\cW\to\oCM_{d_{e_1,e_2}\Gamma'}^{\hat{\pi'}}$ to $\cW\to\oCM_{d_{e_2}\Gamma'}^{\hat{\pi'}},$ and then, via $\text{Ind}$, to $\cW\to\oCM_{\Gamma'}^{\hat{\pi'}}$.  The result is
\[\mathfrak{o}^{\pi'}|_{\oCM_{\Gamma'}^{\hat{\pi'}}}=\epsilon_2\mathfrak{o}^{\pi_1}\boxtimes\mathfrak{o}'_{r+1}\boxtimes\mathfrak{o}^{\pi_2}=
\epsilon_2(-1)^{rk_2}\mathfrak{o}^{\pi_1}\boxtimes\mathfrak{o}^{\pi_2}\boxtimes\mathfrak{o}'_{r+1},\]
in which $\epsilon_2$ is the sign that appears from inducing the orientation to $\cW\to\oCM_{d_{e_2}\Gamma'}^{\hat{\pi'}}$ and the sign $(-1)^{rk_2}$ comes from changing the order of $\mathfrak{o}^{\pi_2}$ and $\mathfrak{o}'_{r+1}.$
One of $\epsilon_1$ or $\epsilon_2$ is determined by induction to be $1$, so the other is also $1$. The case $m_2=0$ is thus proven.
\end{proof}

\begin{proof}[Proof of Lemma \ref{lem:induced_or_of_Witten-internal}]
Write
\[\mathfrak{o}^\pi|_{\CM_\Gamma^{\hat\pi}} = \epsilon^\pi_{\{a_i\}_{i\in I^C},\{a_i\}_{i\in I^O},B}\mathfrak{o}^\pi_{v^o}\boxtimes\mathfrak{o}_{v^c},\]
where $I^O$ are the labels in the open part and $I^C$ in the closed part.  By covariance $\epsilon^\pi_{\{a_i\}_{i\in I^C},\{a_i\}_{i\in I^O},B}$ is independent of $\pi$ and can be written as $\epsilon_{\{a_i\}_{i\in I^C},\{a_i\}_{i\in I^O},|B|}.$
We prove the lemma by showing:
\begin{enumerate}[(1)]
\item $\epsilon_{\{a_i\}_{i\in I^C},\{a_i\}_{i\in I^O},B}=\epsilon_{\{a_i\}_{i\in I^C}},$ meaning that it depends only on $\{a_i\}_{i\in I^C}.$
\item If $I^C=I_1\cup I_2$ where neither $I_1$ nor $I_2$ is empty and $|I_1|\geq 2$, then
\[\epsilon_{\{a_i\}_{i\in I^C}}=\epsilon_{\{a_i\}_{i\in I_1}}\epsilon_{\{a_i\}_{i\in I_2}\cup\{\sum_{i\in I_1}a_i \text{ mod } r\}}.\]
Thus, $\epsilon$ is fully determined by its value on pairs of elements.
\item If $a+b<r$, then $\epsilon_{\{a,b\}}=1$.
\item $\epsilon_{\{a,b+c (\text{mod } r)\}}\epsilon_{\{b,c\}} = \epsilon_{\{a+b (\text{mod } r),c\}}\epsilon_{\{a,b\}},$
and when $a<r-1$ but $a+b \geq r$
\[\epsilon_{\{a,b\}}=\epsilon_{\{1,a+b-r\}}\epsilon_{\{a,b\}} = \epsilon_{\{1+a,b\}}\epsilon_{\{1,a\}}=\epsilon_{\{1+a,b\}},\]
so that for $r\leq a+b\leq 2r-2,~\epsilon_{\{a,b\}}=\delta\in\{\pm1\}$ is a constant.
\item
$\epsilon_{\{a_i\}_{i\in I^C}}=\delta^{m^c(\{a_i\}_{i\in I^C})}.$
\end{enumerate}

For the first item, it is enough to show that $\epsilon_{\{a_i\}_{i\in I^C},\{a_i\}_{i\in I^O},B} = \epsilon_{\{a_i\}_{i\in I^C},\{a_i\}_{i\in I^{'O}},B'},$ where $B\subset B'$ and $I^O\subset I^{'O}.$ To show this, consider the graph $\Gamma'$ consisting of two open vertices, $v^o$ and $v'$, and a closed vertex, where the open vertices are connected by an edge $e'$ whose half-edge $h$ in $v^o$ has $\tw(h)=\alt(h)=0$, and the component of $v^o$ in $\detach_{e'}(\Gamma')$ is $\Gamma.$ Denote the internal edge of $\Gamma$ (and of $\Gamma'$) by $e.$ Let $\pi'$ be an order extending $\pi$ to $B'$ such that the first elements are those of $B.$ By abuse of notation, denote also by $\pi'$ the restriction of $\pi$ to $B'\setminus B$.  We calculate the orientation of $\cW|_{\CM_{\Gamma'}}$ in two ways.
First, by the definition of $\epsilon$ applied to the moduli space of the smoothing of $\Gamma'$ along $e'$, we have
\[\mathfrak{o}^{\pi'}|_{\CM^{\hat{\pi'}}_{d_{e'}\Gamma'}} = \epsilon_{\{a_i\}_{i\in I^C},\{a_i\}_{i\in I^{ O'}},B'}\mathfrak{o}^{\pi'}_{0,B',\{a_i\}_{i\in I^{'O}\cup\{\sum_{i\in I^C} a_i(\text{ mod }r)\}}}\boxtimes\mathfrak{o}_{v^c}.\]
Applying Lemma \ref{lem:induced or to boundary bundle} to $\oCM^{\frac{1}{r}}_{0,B',\{a_i\}_{i\in I^{'O}}\cup\{\sum_{i\in I^C} a_i \text{ mod } r\}}$ gives
\[\mathfrak{o}^{\pi'}|_{\CM^{\hat{\pi'}}_{\Gamma'}}=\epsilon_{\{a_i\}_{i\in I^C},\{a_i\}_{i\in I^{ O'}},B'}
\mathfrak{o}^\pi_{v^o}\boxtimes\mathfrak{o}^{\pi'}_{v'}\boxtimes
\mathfrak{o}_{v^c}.
\]
On the other hand, by Lemma \ref{lem:induced or to boundary bundle} applied to $\oCM^{\frac{1}{r}}_{0,B',\{a_i\}_{i\in I^C\cup I^{'O}}}$, we have
\[\mathfrak{o}^{\pi'}|_{\CM^{\hat{\pi'}}_{d_{e}\Gamma'}}=\mathfrak{o}^\pi_{{0,B,\{a_i\}_{i\in I^C\cup I^{O}}}}\boxtimes\mathfrak{o}^{\pi'}_{v'}.\]
The claim now follows since, by the definition of $\epsilon$, we have
\[\mathfrak{o}^{\pi'}|_{\CM^{\hat{\pi'}}_{\Gamma'}}=\epsilon_{\{a_i\}_{i\in I^C},\{a_i\}_{i\in I^{ O}},B}
\mathfrak{o}^\pi_{v^o}\boxtimes\mathfrak{o}_{v^c}\boxtimes\mathfrak{o}^{\pi'}_{v'}.
\]

For the second item, first note that if $\Lambda$ is a connected, closed $r$-spin graph consisting of two vertices $v_1$ and $v_2$ and an edge between them, then
\[\mathfrak{o}^\pi|_{\CM_\Lambda} = \mathfrak{o}_{v_1}\boxtimes\mathfrak{o}_{v_2},\] where all the orientations in the equation are the canonical complex orientations.  Consider a graph $\Gamma'$ with two internal closed vertices $v_1$ and $v_2$ and an open vertex $v^o$, where $v^o$ is connected to $v_2$ by $e_2$ and $v_2$ is connected to $v_1$ by $e_1$.  We again calculate the orientation of $\cW\to\CM_{\Gamma'}$ in two ways.  First, by the definition of $\epsilon$,
\[\mathfrak{o}^\pi|_{\CM^{\hat{\pi}}_{d_{e_1}\Gamma}} = \epsilon_{\{a_i\}_{i\in I^C}}\mathfrak{o}^\pi_{v^o}\boxtimes\mathfrak{o}_{v^c},\] where $v^c$ is the closed vertex of $d_{e_1}\Gamma$ and $I^C$ are the labels of its tails.  Then
\[\mathfrak{o}^\pi|_{\CM^{\hat{\pi}}_{\Gamma}} = \epsilon_{\{a_i\}_{i\in I^C}}\mathfrak{o}^\pi_{v^o}\boxtimes\mathfrak{o}_{v_1}\boxtimes\mathfrak{o}_{v_2}.\]
On the other hand, by the definition of $\epsilon,$
\[\mathfrak{o}^\pi|_{\CM^{\hat{\pi}}_{d_{e_2}\Gamma}} = \epsilon_{\{a_i\}_{i\in I_1}}\boxtimes\mathfrak{o}^\pi_{v^{'o}}\boxtimes\mathfrak{o}_{v_1},\] where $v^{'o}$ is the open vertex of $d_{e_2}\Gamma$ and $I_i$ are the labels of the tails on $v_i$.  Again by the definition of $\epsilon$,
\[\mathfrak{o}^\pi|_{\CM^{\hat{\pi}}_{\Gamma}} = \epsilon_{\{a_i\}_{i\in I_1}}\epsilon_{\{a_i\}_{i\in I_2}\cup\{\sum_{i\in I_1}a_i(\text{ mod }r)\}}\mathfrak{o}^\pi_{v^o}\boxtimes\mathfrak{o}_{v_1}\boxtimes\mathfrak{o}_{v_2},\] as claimed.

For the third item, we use Observation \ref{obs:cases when pulled back from one component}.  Let $I=[l]$ and $B=[k]$ and assume that $k-1=\sum_{i\in [l]}a_i$.  Suppose $I'\subseteq [l]$ is such that $\sum_{i\in I'} a_i\leq r-1$, and let $\Gamma$ be the graph with a closed vertex $v^c$ containing exactly the tails labeled $I'$ and an internal edge to the open vertex $v^o$ with internal tails labeled by $[l]\setminus I'$ and boundary tails labeled $[k]$.  We claim, in this situation, that the two orientations $\mathfrak{o}^\pi|_{\oCM_{\Gamma}^{\hat\pi}}$ and $\mathfrak{o}^\pi_{0,k,\{a_i\}_{i\in[l]\setminus I'}\cup\{\sum a_i\}_{i\in I'}}$ on $\cW\to\oCM_{\Gamma}^\pi$ agree (where for convenience we omit the pullback maps from the notation).

To prove this claim, write $a = \sum_{i\in I'} a_i$, and let $\sigma_j = \sigma^\pi_{0,k,\{a_i\}_{i\in[l]};j}$ and $\sigma'_j=\sigma^\pi_{0,k,\{a_i\}_{i\in[l]\setminus I'}\cup\{a\};j}$. Then it is a direct computation to verify that, when $\Sigma'\to\Sigma \in\CM_\Gamma^\pi$,
the section $(\sigma_j)_{\Sigma'}$ converges to a section $(\hat{\sigma}_j)_{\Sigma}$ of $\cW_{\Sigma}$ that is the pullback of $\sigma'_j.$ Indeed, denoting by $z_{h^o}$ the half-node in the disk component of $\Sigma,$ the projection to $\cW_{v^o}$ of the limit of $(\sigma_j)_{\Sigma'}$ is
\[\left(
\left(\frac{\sqrt{-1}(\bar{z}_{h^o}-z_{h^o})dw}{(w-z_{h^o})(w-\bar{z}_{h^o})}\right)^{a}
\prod_{h\in [l]\setminus I'}\xi_h^{a_h/r}\prod_{h\in[k]}\xi_{h(h+1)}^{-1}\xi_{1(1+j)}^r\right)^{\frac{1}{r}},\]
as claimed. Thanks to the first item, this claim implies the third item.

The fourth item is a direct consequence of the second item, by partitioning the set $\{a,b,c\}$ in two different ways.  Its second part uses the third item, applied twice to $\{1,a,b\}$.

The last item follows by induction on $|I^C|\geq 2$.  The third and fourth items serve as the base case. Suppose, then, that we have shown the claim for $|I^C|=n$, and let $|I^C|=n+1.$ Write $I^C = I_1\cup I_2$ with $I_2=\{a,b\}$, and write $c= a+b \mod r$.  Then
\begin{equation*}\epsilon_{\{a_i\}_{i\in I^C}} = \epsilon_{\{a_i\}_{i\in I_1}\cup\{c\}}\epsilon_{a,b}=(-1)^{p},\end{equation*}
where the power $p$ is given by
\begin{align*}
p = &\frac{c+\sum a_i-(r-2)+(r-2-c-\sum_{i\in I_1} a_i\;(\text{mod }r)')}{r}\\
& \hspace{1cm} + \frac{a + b -(r-2) +(r-2-a-b)(\text{mod } r)'}{r}.
\end{align*}
The induction now follows from the definition of $m^c$ and the two equations
\[r-2-c-\sum_{i\in I_1} a_i \;(\text{mod }r)'=r-2-\sum_{i\in I^C} a_i \;(\text{mod }r)',\]
\[c + (r-2-a-b)(\text{mod }r)' = r-2.\]

The ``moreover" statement of the lemma is straightforward from the definitions. Indeed, when $k=1+\sum a_i$, the orientation is defined without $\mathfrak{o}_{r+1}$, so it does not change when $\mathfrak{o}_{r+1}$ changes. For general $k'$, the construction of orientations in Definition \ref{def:or_for_witten_based_on_o_r+1} uses the map $\text{ind}_{k\to k'}$, and it is immediate from the definition of $\text{ind}_{k\to k'}$ that changing $\mathfrak{o}_{r+1}$ changes the orientation by a factor of $(-1)^{\frac{k'-\sum_{i\in I} a_i -1}{r}}$.  The claim now follows from noting that
\[\frac{\sum_{i\in I} a_i - k'+1}{r}=
m^c\left(\{a_i\}_{i\in I^C}\right)+\frac{\left\{\sum_{i\in I^C}a_i \; (\text{mod } r)\right \}+\sum_{i\in I^O} a_i - k'+1}{r}
.\]
\end{proof}

\bibliographystyle{abbrv}
\bibliography{OpenBiblio}

\end{document}